\documentclass[11pt]{article}
\usepackage{amsmath,amsfonts,amsthm,amssymb,dsfont}
\usepackage{times}
\usepackage{comment}
\usepackage[margin=1in]{geometry}
\usepackage{soul,color}
\usepackage{graphicx,float,wrapfig}
\usepackage{mathrsfs}
\usepackage[usenames,dvipsnames]{pstricks}
\usepackage{epsfig}
\usepackage{pst-grad} 
\usepackage{pst-plot} 
\usepackage[linesnumbered,boxed,ruled,vlined]{algorithm2e}
\usepackage{tikz}
\usetikzlibrary{arrows,automata}

\newtheorem{prob}{Problem}
\newtheorem{problem}{Problem}
\newtheorem{theo}{Theorem}[section]

\newtheorem{lemma}[theo]{Lemma}

\newtheorem{prop}[theo]{Proposition}
\newtheorem{cor}[theo]{Corollary}
\newtheorem{defi}[theo]{Definition}
\newtheorem{rem}[theo]{Remark}

\newtheorem{ob}{Observation}
\newtheorem{as}{Assumption}
\newtheorem{assumption}{Assumption}

\newtheorem{conjecture}{Conjecture}

\newenvironment{proofof}[1]{\begin{proof}[Proof of #1]}{\end{proof}}
\newenvironment{proofsketch}{\begin{proof}[Proof Sketch]}{\end{proof}}

\newcommand{\spz}[1]{|#1\rangle}

\newcommand{\rpz}[1]{\langle #1 |}

\newcommand{\ct}[1]{{#1}^{\dagger}}

\newcommand{\Tr}{\mathrm{Tr}}

\newcommand{\outpt}[1]{\spz{#1}\rpz{#1}}

\newcommand{\Ex}{\operatorname*{\mathbb{E}}}
\everymath{\displaystyle}
\newcommand{\betw}{\ |\ }

\newcommand{\SampBPP}{\textsf{SampBPP}}
\newcommand{\SampBQP}{\textsf{SampBQP}}
\newcommand{\PSPACE}{\textsf{PSPACE}}
\newcommand{\PH}{\textsf{PH}}
\newcommand{\TQBF}{\textsf{TQBF}}
\newcommand{\BosonSampling}{\textsf{BosonSampling}}

\newcommand{\Sipser}{\textsf{Sipser}}
\newcommand{\ACz}{\textsf{AC}_0}

\newcommand{\oracle}{\mathcal{O}}
\newcommand{\distr}{\mathcal{D}}
\newcommand{\distrO}{\distr_\oracle}

\newcommand{\pj}{\mathcal{P}}
\newcommand{\heb}{\mathcal{H}}

\newcommand{\floor}[1]{\lfloor #1 \rfloor}
\newcommand{\ceil}[1]{\lceil #1 \rceil}

\newcommand{\eps}{\varepsilon}

\renewcommand{\epsilon}{\varepsilon}

\newcommand{\ffishing}{\textsf{Fourier Fishing}}
\newcommand{\fsampling}{\textsf{Fourier Sampling}}
\newcommand{\ffish}{\textsf{Ffishing}}
\newcommand{\pffish}{\textsf{promise-Ffishing}}
\newcommand{\fsamp}{\textsf{Fsampling}}
\newcommand{\OR}{\mathsf{OR}}

\newcommand{\adv}{\mathsf{adv}}
\newcommand{\Qsucc}{\mathsf{Succ}_{Q}}
\newcommand{\Rsucc}{\mathsf{Succ}_{R}}

\newcommand{\ppoly}{\textsf{P/poly}}
\newcommand{\PTIME}{\textsf{P}}
\newcommand{\NP}{\textsf{NP}}

\newcommand{\BPP}{\textsf{BPP}}

\newcommand{\SZK}{\textsf{SZK}}
\newcommand{\BQP}{\textsf{BQP}}
\newcommand{\DTIME}{\textsf{DTIME}}
\newcommand{\GapMaj}{\mathsf{GapMaj}}
\newcommand{\Hwrs}{H_{\textrm{wrs}}}

\newcommand{\PRP}{\mathsf{PRP}}
\newcommand{\PRF}{\mathsf{PRF}}
\newcommand{\key}{\mathcal{K}}
\newcommand{\domain}{\mathcal{X}}
\newcommand{\image}{\mathcal{Y}}

\newcommand{\PRPraw}{\PRP^{\mathsf{raw}}}
\newcommand{\PRFmod}{\PRF^{\mathsf{mod}}}
\newcommand{\keyraw}{\key^{\mathsf{raw}}}
\newcommand{\keymod}{\key^{\mathsf{mod}}}
\newcommand{\domainraw}{\domain^{\mathsf{raw}}}
\newcommand{\domainmod}{\domain^{\mathsf{mod}}}
\newcommand{\imageraw}{\domainraw}
\newcommand{\imagemod}{\image^{\mathsf{mod}}}
\newcommand{\moduliset}{\mathcal{A}}

\def\ShowAuthNotes{1}
\ifnum\ShowAuthNotes=1
\newcommand{\authnote}[2]{\ \\ \textcolor{red}{\parbox{0.9\linewidth}{[{\footnotesize {\bf #1:} { {#2}}}]}}\newline}
\else
\newcommand{\authnote}[2]{}
\fi

\begin{document}
	\title{Complexity-Theoretic Foundations of Quantum Supremacy Experiments}
	\author{Scott Aaronson\thanks{The University of Texas at Austin.
			\ aaronson@cs.utexas.edu. \ Supported by a Vannevar Bush Faculty Fellowship
			from the US Department of Defense, and by the Simons Foundation
			\textquotedblleft It from Qubit\textquotedblright\ Collaboration.}
		\and Lijie Chen\thanks{Tsinghua University.
			\ wjmzbmr@gmail.com. \ Supported in part by the National Basic Research Program of China Grant 2011CBA00300, 2011CBA00301, the National Natural Science Foundation of China Grant 61361136003.}}
	\date{}

	\clearpage\maketitle
	\thispagestyle{empty}

	\begin{abstract}
		In the near future, there will likely be special-purpose quantum computers
		with 40-50 high-quality qubits. \ This paper lays general theoretical
		foundations for how to use such devices to demonstrate \textquotedblleft
		quantum supremacy\textquotedblright: that is, a clear quantum speedup for
		\textit{some} task, motivated by the goal of overturning the Extended
		Church-Turing Thesis as confidently as possible.
		
		First, we study the hardness of sampling the output distribution of a random
		quantum circuit, along the lines of a recent proposal by the Quantum AI group at
		Google. \ We show that there's a natural average-case hardness assumption,
		which has nothing to do with sampling, yet implies that no polynomial-time
		classical algorithm can pass a statistical test that the quantum sampling
		procedure's outputs do pass. \ Compared to previous work---for example, on
		$\BosonSampling$ and $\mathsf{IQP}$---the central advantage is that we can now
		talk directly about the observed outputs, rather than about the distribution
		being sampled.
		
		Second, in an attempt to refute our hardness assumption, we give a new
		algorithm, inspired by Savitch's Theorem, for simulating a general quantum
		circuit with $n$ qubits and depth $d$ in polynomial space and $d^{O\left(
			n\right)  }$ time. \ We then discuss why this and other known algorithms fail
		to refute our assumption.
		
		Third, resolving an open problem of Aaronson and Arkhipov, we show that any
		strong quantum supremacy theorem---of the form \textquotedblleft if
		approximate quantum sampling is classically easy, then the polynomial
		hierarchy collapses\textquotedblright---must be non-relativizing.
		\ This sharply contrasts with the situation for exact sampling.
		
		Fourth, refuting a conjecture by Aaronson and Ambainis, we show that there is a
        sampling task, namely $\fsampling$, with a \textit{1 versus linear} separation
		between its quantum and classical query complexities.
		
		Fifth, in search of a \textquotedblleft happy medium\textquotedblright%
		\ between black-box and non-black-box arguments, we study quantum supremacy
		relative to oracles in $\mathsf{P/poly}$. \ Previous work implies that, if
		one-way functions exist, then quantum supremacy is possible relative to such
		oracles. \ We show, conversely, that \textit{some} computational assumption is
		needed: if $\mathsf{SampBPP}=\mathsf{SampBQP}$ and $\mathsf{NP}\subseteq
		\mathsf{BPP}$, then quantum supremacy is impossible relative to oracles with
		small circuits.
		
	\end{abstract}
	
	\addtocounter{page}{-1}
	\newpage
	
	\section{Introduction\label{INTRO}}
	
	The \textit{Extended Church-Turing Thesis}, or ECT, asserts that every
	physical process can be simulated by a deterministic or probabilistic Turing
	machine with at most polynomial overhead. \ Since the 1980s---and certainly
	since the discovery of Shor's algorithm \cite{shor}\ in the 1990s---computer
	scientists have understood that quantum mechanics might refute the ECT in
	principle. \ Today, there are actual experiments being planned (e.g.,
	\cite{martinis:supremacy})\ with the goal of severely challenging the ECT in
	practice. \ These experiments don't yet aim to build full,
	fault-tolerant, universal quantum computers, but \textquotedblleft
	merely\textquotedblright\ to demonstrate \textit{some} quantum speedup over
	the best known or conjectured classical algorithms, for some
	possibly-contrived task, as confidently as possible. \ In other words, the
	goal is to answer the skeptics \cite{kalai,levin}\ who claim that genuine
	quantum speedups are either impossible in theory, or at any rate, are
	hopelessly out of reach technologically. \ Recently,\ the term
	\textquotedblleft quantum supremacy\textquotedblright\ has come into vogue for
	such experiments,\footnote{As far as we know, the first person to use the term in print was John Preskill \cite{preskill:solvay}.} although the basic goal goes back several decades, to the
	beginning of quantum computing itself.
	
	Before going further, we should address some common misunderstandings about
	quantum supremacy.
	
	The ECT is an asymptotic claim,\ which of course means that \textit{no}
	finite experiment could render a decisive verdict on it, even in principle.
	\ But this hardly makes experiments irrelevant. \ If
	
	\begin{enumerate}
		\item[(1)] a quantum device performed some task (say) $10^{15}$\ times faster
		than a highly optimized simulation written by \textquotedblleft
		adversaries\textquotedblright\ and running on a classical computing cluster,
		with the quantum/classical gap appearing to increase exponentially with the
		instance size across the whole range tested, and
		
		\item[(2)] this observed performance closely matched theoretical results that
		\textit{predicted} such an exponential quantum speedup for the task in
		question, and
		
		\item[(3)] all other consistency checks passed (for example: removing quantum
		behavior from the experimental device destroyed the observed speedup),
	\end{enumerate}
	
	\noindent this would obviously \textquotedblleft raise the
	stakes\textquotedblright\ for anyone who still believed the ECT! \ Indeed,
	when some quantum computing researchers have
	criticized previous claims to have experimentally achieved quantum speedups
	(see, e.g., \cite{aar:dwave}), it has typically been on the ground that, in those researchers' view,
	the experiments failed to meet one or more of the conditions above.
	
	It's sometimes claimed that any molecule in Nature or the laboratory,
	for which chemists find it computationally prohibitive to solve the
	Schr\"{o}dinger equation and calculate its ground state, \textit{already}
	provides an example of \textquotedblleft quantum supremacy.\textquotedblright%
	\ \ The idea, in other words, is that such a molecule constitutes a
	\textquotedblleft useful quantum computer, for the task of simulating
	itself.\textquotedblright
	
	For us, the central problem with this idea is that in theoretical computer
	science, we care less about individual \textit{instances} than about solving
	\textit{problems} (i.e., infinite collections of instances) in a more-or-less
	uniform way. \ For any one molecule, the difficulty in simulating classically
	it \textit{might} reflect genuine asymptotic hardness, but it might also
	reflect other issues (e.g., a failure to exploit special structure in the
	molecule, or the same issues of modeling error, constant-factor overheads, and
	so forth that arise even in simulations of classical physics).
	
	Thus, while it's possible that complex molecules could form the basis
	for a convincing quantum supremacy demonstration, we believe more work would need to be
	done. \ In particular, one would want a device that could synthesize
	\textit{any} molecule in some theoretically infinite class---and one would
	then want complexity-theoretic evidence that the general problem, of
	simulating a given molecule from that class, is asymptotically hard for a
	classical computer. \ And in such a case, it would seem more natural to call
	the synthesis machine the \textquotedblleft quantum
	computer,\textquotedblright\ rather than the molecules themselves!
	
	In summary, we regard quantum supremacy as a central milestone for quantum
	computing that hasn't been reached yet, but that might be reached in the near
	future. \ This milestone is essentially negative in character: it has no
	obvious signature of the sort familiar to experimental physics, since it
	simply amounts to the \textit{nonexistence} of an efficient classical
	algorithm to simulate a given quantum process. \ For that reason, the tools of
    theoretical computer science will be essential to understand when quantum supremacy has or
	hasn't been achieved. \ So in our view, even if it were uninteresting as
	TCS, there would still be an urgent need for TCS to
	contribute to the discussion about which quantum supremacy experiments to do,
	how to verify their results, and what should count as convincing evidence that
	classical simulation is hard. \ Happily, it turns out that there \textit{is} a great deal
	here of intrinsic TCS interest as well.
	
	\subsection{Supremacy from Sampling}
	
	In recent years, a realization has crystallized that, if our goal is to
	demonstrate quantum supremacy (rather than doing anything directly useful),
	then there are good reasons to shift our attention from decision and function
	problems to \textit{sampling} problems: that is, problems where the goal is to
	sample an $n$-bit string, either exactly or approximately, from a desired
	probability distribution.
	
	A first reason for this is that demonstrating quantum supremacy via a sampling
	problem doesn't appear to require the full strength of a universal quantum
	computer. \ Indeed, there are now at least a half-dozen proposals
	\cite{aark,bjs,farhiharrow,td,mff,jvn,abkm}\ for special-purpose devices that
	could efficiently solve sampling problems believed to be classically
	intractable, \textit{without} being able to solve every problem in the class
	$\mathsf{BQP}$, or for that matter even every problem in $\mathsf{P}$.
	\ Besides their intrinsic physical and mathematical interest,
	these intermediate models might be easier to realize than a universal quantum
	computer. \ In particular, because of their simplicity, they might let us
	avoid the need for the full machinery of quantum fault-tolerance \cite{ab}: something
	that adds a theoretically polylogarithmic but enormous-in-practice overhead to
	quantum computation. \ Thus, many researchers now expect that the first
	convincing demonstration of quantum supremacy will come via this route.
	
	A second reason to focus on sampling problems is more theoretical: in the
	present state of complexity theory, we can arguably be \textit{more} confident
	that certain quantum sampling problems really are classically hard, than we
	are that \textit{factoring} (for example) is classically hard, or even that
	$\mathsf{BPP}\neq\mathsf{BQP}$. \ Already in 2002, Terhal and DiVincenzo
	\cite{td} noticed that, while constant-depth quantum circuits can't solve any
	classically intractable decision problems,\footnote{This is because any qubit
		output by such a circuit depends on at most a constant number of input
		qubits.} they nevertheless have a curious power: namely, they can sample
	probability distributions that can't be sampled in classical polynomial time,
	unless $\mathsf{BQP}\subseteq\mathsf{AM}$, which would be a surprising
    inclusion of complexity classes. \ Then, in 2004, Aaronson showed
	that $\mathsf{PostBQP}=\mathsf{PP}$, where $\mathsf{PostBQP}$\ means
	$\mathsf{BQP}$\ with the ability to postselect on exponentially-unlikely
	measurement outcomes. \ This had the immediate corollary that, if there's an
	efficient classical algorithm to sample the output distribution of an
	arbitrary quantum circuit---or for that matter, any distribution whose
	probabilities are multiplicatively close to the correct ones---then%
	\[
	\mathsf{PP}=\mathsf{PostBQP}=\mathsf{PostBPP}\subseteq\mathsf{BPP}%
	^{\mathsf{NP}}.
	\]
	By Toda's Theorem \cite{toda}, this implies that the polynomial hierarchy
	collapses to the third level.

    Related to that, in 2009, Aaronson \cite{aar:ph}\ showed that,
	while it was (and remains) a notorious open problem to construct an oracle
	relative to which $\mathsf{BQP}\not \subset \mathsf{PH}$, one can construct
	oracular \textit{sampling} and \textit{relation} problems that are solvable in
	quantum polynomial time, but that are provably not solvable in randomized
	polynomial time augmented with a $\mathsf{PH}$\ oracle.
	
	Then, partly inspired by that oracle separation, Aaronson and Arkhipov
	\cite{aark}\ proposed $\BosonSampling$: a model that uses identical photons
	traveling through a network of beamsplitters and phaseshifters to solve
	classically hard sampling problems. \ Aaronson and Arkhipov proved that a polynomial-time exact
	classical simulation of $\BosonSampling$ would collapse $\mathsf{PH}$. \ They also gave a
	plausible conjecture implying that even an \textit{approximate} simulation would have
	the same consequence. \ Around the same time, Bremner, Jozsa, and Shepherd
	\cite{bjs}\ independently proposed the Commuting Hamiltonians or
	$\mathsf{IQP}$\ (\textquotedblleft Instantaneous Quantum
	Polynomial-Time\textquotedblright) model, and showed that it had the same
	property, that exact classical simulation would collapse $\mathsf{PH}$.
	\ Later, Bremner, Montanaro, and Shepherd \cite{bms,bms2} showed that, just
	like for $\BosonSampling$, there are plausible conjectures under which even a
	fast classical \textit{approximate} simulation of the $\mathsf{IQP}$\ model
	would collapse $\mathsf{PH}$.
	
	Since then, other models have been proposed with similar behavior. \ To take a
	few examples: Farhi and Harrow \cite{farhiharrow} showed that the so-called
	Quantum Approximate Optimization Algorithm, or QAOA, can sample distributions
	that are classically intractable unless $\mathsf{PH}$\ collapses. \ Morimae,
	Fujii, and Fitzsimons \cite{mff} showed the same for the so-called One Clean
	Qubit or $\mathsf{DQC1}$\ model, while\ Jozsa and Van den Nest \cite{jvn}%
	\ showed it for stabilizer circuits with magic initial states and nonadaptive
	measurements, and Aaronson et al.\ \cite{abkm} showed it for a model based on
	integrable particle scattering in $1+1$\ dimensions. \ In retrospect, the
	constant-depth quantum circuits considered by Terhal and DiVincenzo \cite{td}
	also have the property that fast exact classical simulation would collapse
	$\mathsf{PH}$.
	
	Within the last four years, quantum supremacy via sampling has made the leap
	from complexity theory to a serious experimental prospect. \ For example,
	there have by now been many small-scale demonstrations of $\BosonSampling$ in
	linear-optical systems, with the current record being a $6$-photon experiment
	by Carolan et al.\ \cite{carolan}. \ To scale up to (say) $30$ or $40$
	photons---as would be needed to make a classical simulation of the experiment
	suitably difficult---seems to require more reliable single-photon sources than
	exist today. \ But some experts (e.g., \cite{rudolph:optimistic,pgha}) are
	optimistic that optical multiplexing, superconducting resonators, or other
	technologies currently under development will lead to such photon sources. \ In the
	meantime, as we mentioned earlier, Boixo et al.\ \cite{martinis:supremacy}%
	\ have publicly detailed a plan, currently underway at Google, to perform a
	quantum supremacy experiment involving random circuits applied to a 2D array
	of 40-50 coupled superconducting qubits. \ So far, the group at
	Google has demonstrated the preparation and measurement of entangled states on
	a linear array of 9 superconducting qubits \cite{kelly}.
	
	\subsection{Theoretical Challenges\label{CHALLENGES}}
	
	Despite the exciting recent progress in both theory and experiment, some huge
	conceptual problems have remained about sampling-based quantum supremacy.
	\ These problems are not specific to any one quantum supremacy
	proposal (such as $\BosonSampling$, $\mathsf{IQP}$, or random quantum circuits),
	but apply with minor variations to all of them.\bigskip
	
	\textbf{Verification of Quantum Supremacy Experiments.} \ From the beginning,
	there was the problem of \textit{how to verify the results} of a
	sampling-based quantum supremacy experiment. \ In contrast to (say)
	factoring and discrete log, for sampling tasks such as $\BosonSampling$, it seems
	unlikely that there's any $\mathsf{NP}$ witness certifying the quantum
	experiment's output, let alone an $\mathsf{NP}$ witness\ that's also the
	experimental output itself. \ Rather, for the sampling tasks, not only
	simulation but even verification might need classical exponential time. \ Yet,
	while no one has yet discovered a general way around this,\footnote{In principle, one could use so-called
    \textit{authenticated quantum computing} \cite{abe,bfk}, but the known schemes for that might be much harder
    to realize technologically than a basic quantum supremacy experiment, and in any case, they all presuppose the validity of quantum mechanics.} it's far from the fatal
	problem that some have imagined. \ The reason is simply that experiments can
	and will target a \textquotedblleft sweet spot,\textquotedblright\ of (say)
	40-50 qubits, for which classical simulation and verification of the results
	is \textit{difficult but not impossible}.
	
	Still, the existing verification methods have a second drawback. \ Namely,
	once we've fixed a specific verification test for sampling from a probability
	distribution $\mathcal{D}$, we ought to consider, not merely all
	classical algorithms that sample exactly or approximately from $\mathcal{D}$,
	but \textit{all classical algorithms that output anything that passes the
		verification test.} \ To put it differently, we ought to talk not about the
	sampling problem itself, but about an associated \textit{relation problem}:
	that is, a problem where the goal is to produce any output that satisfies a
	given condition.
	
	As it happens, in 2011, Aaronson \cite{aaronson2014equivalence} proved an extremely general
	connection between sampling problems and relation problems. \ Namely, given
	any approximate sampling problem $S$, he showed how to define a relation
	problem $R_{S}$\ such that, for every \textquotedblleft
	reasonable\textquotedblright\ model of computation (classical, quantum, etc.),
	$R_{S}$\ is efficiently solvable in that model if and only if $S$ is. \ This
	had the corollary that%
	\[
	\mathsf{SampBPP}=\mathsf{SampBQP}\iff\mathsf{FBPP}=\mathsf{FBQP,}%
	\]
	where $\mathsf{SampBPP}$ and $\mathsf{SampBQP}$\ are the classes of
	approximate sampling problems solvable in polynomial time by randomized and
	quantum algorithms respectively, and $\mathsf{FBPP}$\ and $\mathsf{FBQP}$\ are
	the corresponding classes of relation problems. \ Unfortunately, Aaronson's
	construction of $R_{S}$\ involved Kolmogorov complexity: basically, one asks
	for an $m$-tuple of strings, $\left\langle x_{1},\ldots,x_{m}\right\rangle $,
	such that%
	\[
	K\left(  x_{1},\ldots,x_{m}\right)  \geq\log_{2}\frac{1}{p_{1}\cdots p_{m}%
	}-O\left(  1\right)  ,
	\]
	where $p_{i}$\ is the desired probability of outputting $x_{i}$\ in the
	sampling problem. \ And of course, verifying such a condition is
	extraordinarily difficult, even more so than calculating the
	probabilities $p_{1},\ldots,p_{m}$.\footnote{Furthermore, this is true even if
		we substitute a resource-bounded Kolmogorov complexity, as Aaronson's result
		allows.} \ For this reason, it's strongly preferable to have a condition that
	talks only about the largeness of the $p_{i}$'s, and not about the
	algorithmic randomness of the $x_{i}$'s. \ But then hardness for the
	sampling problem no longer necessarily implies hardness for the relation
	problem, so a new argument is needed.\bigskip
	
	\textbf{Supremacy Theorems for Approximate Sampling.} \ A second difficulty is
	that any quantum sampling device is subject to noise and decoherence.
	\ Ultimately, of course, we'd like hardness results for quantum sampling that
	apply even in the presence of experimentally realistic errors. \ Very
	recently, Bremner, Montanaro, and Shepherd \cite{bms2}\ and Fujii
	\cite{fujii}\ have taken some promising initial steps in that direction. \ But
	even if we care only about the \textit{smallest} \textquotedblleft
	experimentally reasonable\textquotedblright\ error---namely, an error that
	corrupts the output distribution $\mathcal{D}$\ to some other distribution
	$\mathcal{D}^{\prime}$\ that's $\varepsilon$-close to $\mathcal{D}$\ in
	variation distance---Aaronson and Arkhipov \cite{aark} found that we already
	engage substantial new open problems in complexity theory, if we want evidence
	for classical hardness. \ So for example, their hardness argument for approximate
	$\BosonSampling$ depended on the conjecture that there's no $\mathsf{BPP}%
	^{\mathsf{NP}}$\ algorithm to estimate the permanent of an i.i.d.\ Gaussian
	matrix $A\thicksim N\left(  0,1\right)  _{\mathbb{C}}^{n\times n}$, with high
	probability over the choice of $A$.
	
	Of course, one could try to close that loophole by proving that this Gaussian
	permanent estimation problem is $\mathsf{\#P}$-hard, which is indeed a major challenge that
	Aaronson and Arkhipov left open. \ But this situation also raises more general
	questions. \ For example, is there an implication of the form
	\textquotedblleft if $\mathsf{SampBPP}=\mathsf{SampBQP}$, then $\mathsf{PH}%
	$\ collapses,\textquotedblright\ where again $\mathsf{SampBPP}$\ and
	$\mathsf{SampBQP}$\ are the \textit{approximate} sampling versions of $\mathsf{BPP}%
	$\ and $\mathsf{BQP}$\ respectively? \ Are there oracles relative to which
	such an implication does \textit{not} hold?
    \bigskip
	
	\textbf{Quantum Supremacy Relative to Oracles.} \ A third problem goes to
	perhaps the core issue of complexity theory (both quantum and classical): namely, we don't at
	present have a proof of $\mathsf{P}\neq\mathsf{PSPACE}$, much less of
	$\mathsf{BPP\neq BQP}$\ or\ $\mathsf{SampBPP\neq SampBQP}$, less still of the hardness
	of specific problems like factoring or estimating Gaussian permanents. \ So
	what reason do we have to believe that \textit{any} of these problems are hard?
	\ Part of the evidence has always come from oracle results, which we often
	\textit{can} prove unconditionally. \ Particularly in quantum complexity
	theory, oracle separations can already be highly nontrivial, and give us a
	deep intuition for why all the \textquotedblleft standard\textquotedblright%
	\ algorithmic approaches fail for some problem.
	
	On the other hand, we also know, from results like $\mathsf{IP}%
	=\mathsf{PSPACE}$ \cite{shamir}, that oracle separations can badly mislead us
	about what happens in the unrelativized world. \ Generally speaking, we might
	say, relying on an oracle separation is more dangerous, the less the oracle
	function resembles what would actually be available in an explicit
	problem.\footnote{Indeed, the \textit{algebrization barrier} of Aaronson and
		Wigderson \cite{awig} was based on precisely this insight: namely, if we force
		oracles to be \textquotedblleft more realistic,\textquotedblright\ by
		demanding (in that case) that they come equipped with algebraic extensions of
		whichever Boolean functions they represent, then many previously
		non-relativizing results become relativizing.}
	
	In the case of sampling-based quantum supremacy, we've known strong oracle
	separations since early in the subject. \ Indeed, in 2009, Aaronson
	\cite{aar:ph} showed that $\fsampling$---a quantumly easy sampling problem
	that involves only a \textit{random} oracle---requires classical exponential
	time, and for that matter, sits outside the entire polynomial hierarchy. \ But
	of course, in real life random oracles are unavailable. \ So a question
	arises: can we say anything about the classical hardness of $\fsampling$
	with a \textit{pseudorandom} oracle? \ More broadly, what hardness results can
	we prove for quantum sampling, relative to oracles that are efficiently
	computable? \ Here, we imagine that an algorithm doesn't have access to a
	succinct representation of the oracle function $f$, but it does know that a
	succinct representation \textit{exists} (i.e., that $f\in\mathsf{P/poly}$).
	\ Under that assumption, is there any hope of proving an
	\textit{unconditional} separation between quantum and classical sampling? \ If
	not, then can we at least prove quantum supremacy under weaker (or more
	\textquotedblleft generic\textquotedblright) assumptions than would be needed
	in the purely computational setting?
	
	\subsection{Our Contributions}
	\label{OURCONT}
	
	In this paper, we address all three of the above challenges. \ Our
	results might look wide-ranging, but they're held together by a
	single thread: namely, \textit{the quest to understand the classical hardness
		of quantum approximate sampling problems, and especially the meta-question of
		under which computational assumptions such hardness can be proven.} \ We'll be
	interested in both \textquotedblleft positive\textquotedblright\ results, of
	the form \textquotedblleft quantum sampling problem $X$ is classically hard
	under assumption $Y$,\textquotedblright\ and \textquotedblleft
	negative\textquotedblright\ results, of the form \textquotedblleft proving the
	classical hardness of $X$ \textit{requires} assumption $Y$.\textquotedblright%
	\ \ Also, we'll be less concerned with specific proposals such as $\BosonSampling$,
	than simply with the general task of approximately sampling the output
	distribution of a given quantum circuit $C$. \ Fortuitously, though, our
	focus on quantum circuit sampling will make some of our results an excellent
	fit to currently planned experiments---most notably, those at Google
	\cite{martinis:supremacy}, which will involve random quantum circuits on a 2D
	square lattice of $40$ to $50$ superconducting qubits. \ Even though we won't
	address the details of those or other experiments, our results (together with
	other recent work \cite{martinis:supremacy,bms2}) can help to inform the experiments---for
	example, by showing how the circuit depth, the verification test applied to the
	outputs, and other design choices affect the strength of the computational
	assumptions that are necessary and sufficient to conclude that quantum
	supremacy has been achieved.
	
	We have five main results.\bigskip
	
	\textbf{The Hardness of Quantum Circuit Sampling.} \ Our first result, in
	Section \ref{sec:proposal}, is about the hardness of sampling the output distribution
	of a random quantum circuit, along the general lines of the planned
	Google experiment. \ Specifically, we propose a simple verification
	test to apply to the outputs of a random quantum circuit. \ We then analyze
	the classical hardness of generating \textit{any} outputs that pass that test.
	
	More concretely, we study the following basic problem:
	
	\begin{problem}[HOG, or Heavy Output Generation] Given as input a random quantum
		circuit $C$ (drawn from some suitable ensemble), generate output strings
		$x_{1},\ldots,x_{k}$, at least a $2/3$ fraction of which have greater than the
		median probability in $C$'s output distribution.
	\end{problem}
	
	HOG is a relation problem, for which we can verify a claimed solution in
	classical exponential time, by calculating the ideal probabilities $p_{x_{1}%
	},\ldots,p_{x_{k}}$\ for each $x_{i}$\ to be generated by $C$, and then
	checking whether enough of the $p_{x_{i}}$'s are greater than the median value
	(which we can estimate analytically to extremely high confidence). \ Furthermore, HOG is easy to solve
	on a quantum computer, with overwhelming success probability, by the obvious strategy of
	just running $C$ over and over and collecting $k$ of its outputs.\footnote{Heuristically, one expects the $p_{x_{i}}$'s to be exponentially
		distributed random variables, which one can calculate implies that a roughly $\frac{1+\ln 2}{2}\approx 0.847$ fraction of the outputs will have probabilities exceeding the median value.}
	
	It certainly seems plausible that HOG is exponentially hard for a classical
	computer. \ But we ask: under what assumption could that hardness be
	proven? \ To address that question, we propose a new hardness assumption:
	
	\begin{assumption}[QUATH, or the QUAntum THreshold assumption] There is
		no polynomial-time classical algorithm that takes as input a description of a
		random quantum circuit $C$, and that guesses whether $\left\vert \left\langle
		0^n|C|0^n\right\rangle \right\vert ^{2}$\ is greater or less than the median of
		all $2^{n}$\ of the $\left\vert \left\langle 0^n|C|x\right\rangle \right\vert
		^{2}$\ values, with success probability at least $\frac{1}{2}+\Omega\left(
		\frac{1}{2^{n}}\right)  $\ over the choice of $C$.
	\end{assumption}
	
	Our first result says that if QUATH is true, then HOG is hard. \ While this
	might seem nearly tautological, the important point here is that QUATH makes
	no reference to sampling or relation problems. \ Thus, we can now shift our
	focus from sampling algorithms to algorithms that simply estimate amplitudes,
	with a minuscule advantage over random guessing.\bigskip
	
	\textbf{New Algorithms to Simulate Quantum Circuits.} \ But given what a tiny
	advantage $\Omega\left(  2^{-n} \right)  $\ is, why would
	anyone even conjecture that QUATH might be true? \ This brings us to our
	second result, in Section \ref{sec:new-algorithms-circuits}, which is motivated by the attempt to
	refute QUATH. \ We ask: what \textit{are} the best classical algorithms to
	simulate an arbitrary quantum circuit? \ For special quantum circuits (e.g.,
	those with mostly Clifford gates and a few T gates \cite{bravyigosset}),
	there's been exciting recent progress on improved exponential-time simulation
	algorithms, but for arbitrary quantum circuits, one might think there isn't
	much to say. \ Nevertheless, we \textit{do} find something basic to say that, to our knowledge, had been overlooked earlier.
	
	For a quantum circuit with $n$ qubits and $m$ gates, there are two obvious
	simulation algorithms. \ The first, which we could call the
	\textquotedblleft Schr\"{o}dinger\textquotedblright\ algorithm, stores the entire
	state vector in memory, using $\sim m2^{n}$\ time and $\sim2^{n}$\ space.
	\ The second, which we could call the \textquotedblleft
	Feynman\textquotedblright\ algorithm, calculates an amplitude as a sum of terms,
	using $\sim4^{m}$\ time and $\sim m+n$\ space, as in the proof of
	$\mathsf{BQP}\subseteq\mathsf{P}^{\mathsf{\#P}}$ \cite{bv}.
	
	Now typically $m\gg n$, and the difference between $m$ and $n$ could
	matter enormously in practice. \ For example, in the
	planned Google setup, $n$ will be roughly $40$ or $50$, while $m$
	will ideally be in the thousands. \ Thus, $2^{n}$\ time is reasonable whereas
	$4^{m}$\ time is not. \ So a question arises:
	
	\begin{itemize}
		\item \textit{When }$m\gg n$\textit{, is there a classical algorithm to
			simulate an }$n$\textit{-qubit, }$m$\textit{-gate quantum circuit using both
		}$\operatorname*{poly}\left(  m,n\right)  $\textit{\ space\ and much less than
	}$\exp\left(  m\right)  $\textit{\ time---ideally, more like }$\exp\left(
	n\right)  $\textit{?}
\end{itemize}

We show an affirmative answer. \ In particular, inspired by the proof of
Savitch's Theorem \cite{savitch}, we give a recursive, sum-of-products
algorithm that uses $\operatorname*{poly}\left(  m,n\right)  $\ space and
$m^{O\left(  n\right)  }$ time---or better yet, $d^{O\left(  n\right)  }%
$\ time, where $d$ is the circuit depth. \ We also show how to improve
the running time further for quantum circuits subject to nearest-neighbor
constraints, such as the superconducting systems currently under development.
\ Finally, we show the existence of a \textquotedblleft smooth
tradeoff\textquotedblright\ between our algorithm and the $2^{n}$-memory
Schr\"{o}dinger\ algorithm. \ Namely, starting with the Schr\"{o}dinger
algorithm, for every desired halving of the memory usage, one can multiply the
running time by an additional factor of $\sim d$.

We hope our algorithm finds some applications in quantum simulation. \ In
the meantime, though, the key point for this paper is that neither the Feynman
algorithm, nor the Schr\"{o}dinger algorithm, nor our new recursive algorithm
come close to refuting QUATH. \ The Feynman algorithm fails to refute QUATH
because it yields only a $1/\exp\left(  m\right)  $\ advantage over random
guessing, rather than a $1/2^{n}$\ advantage. \ The Schr\"{o}dinger and
recursive algorithms have much closer to the \textquotedblleft
correct\textquotedblright\ $2^{n}$\ running time, but they also fail to refute
QUATH because they don't calculate amplitudes as straightforward sums, so don't lead to
polynomial-time guessing algorithms at all. \ Thus, in asking whether we can
falsify QUATH, in some sense we're asking how far we can go in combining the
advantages of all these algorithms. \ This might, in turn, connect to
longstanding open problems about the optimality of Savitch's Theorem itself
(e.g., $\mathsf{L}$\ versus $\mathsf{NL}$).

Interestingly, our analysis of quantum circuit simulation algorithms explains
why\ this paper's hardness argument for quantum circuit sampling, based on
QUATH, would \textit{not} have worked for quantum supremacy proposals such as
$\BosonSampling$ or $\mathsf{IQP}$. \ It works only for the more general problem
of quantum circuit sampling. \ The reason is that for the latter, unlike for
$\BosonSampling$ or $\mathsf{IQP}$, there exists a parameter $m\gg n$ (namely,
the number of gates) that controls the advantage that a polynomial-time
classical algorithm can achieve over random guessing, even while $n$ controls
the number of possible outputs. \ Our analysis also underscores the
importance of taking $m\gg n$\ in experiments meant to show quantum supremacy,
and it provides some guidance to experimenters about the crucial question of
what \textit{circuit depth} they need for a convincing quantum
supremacy demonstration.

Note that, the greater the required depth, the
more protected against decoherence the qubits need to be. \ But the tradeoff is that
the depth must be high enough that simulation algorithms that exploit limited entanglement,
such as those based on tensor networks, are ruled out. \ Beyond that requirement, our $d^{O(n)}$ simulation algorithm gives some information about how much additional hardness one can purchase for a given increase in depth. \bigskip

\textbf{Strong Quantum Supremacy Theorems Must Be Non-Relativizing.} \ Next,
in Section \ref{sec:non-relativizing}, we switch our attention to a meta-question. \ Namely,
what sorts of complexity-theoretic evidence we could possibly hope to offer
for $\mathsf{SampBPP}\neq\mathsf{SampBQP}$: in other words, for quantum
computers being able to solve approximate sampling problems that are hard
classically? \ By Aaronson's sampling/searching equivalence theorem
\cite{aaronson2014equivalence}, any such evidence would \textit{also} be evidence for
$\mathsf{FBPP}\neq\mathsf{FBQP}$ (where $\mathsf{FBPP}$\ and $\mathsf{FBQP}%
$\ are the corresponding classes of relation problems), and vice versa.

Of course, an unconditional proof of these separations is out of the question
right now, since it would imply $\mathsf{P}\neq\mathsf{PSPACE}$. \ Perhaps the
next best thing would be to show that, if $\mathsf{SampBPP}=\mathsf{SampBQP}$,
then the polynomial hierarchy collapses. \ This latter is \textit{not} out of
the question: as we said earlier, we already know, by a simple relativizing
argument, that an equivalence between quantum and classical \textit{exact}
sampling implies the collapse $\mathsf{P}^{\#\mathsf{P}}=\mathsf{PH}%
=\mathsf{BPP}^{\mathsf{NP}}$. \ Furthermore, in their work on $\BosonSampling$,
Aaronson and Arkhipov \cite{aark}\ formulated a $\#\mathsf{P}$-hardness
conjecture---namely, their so-called \textit{Permanent of Gaussians Conjecture}, or
PGC---that if true, would imply a generalization of that collapse to the
physically relevant case of approximate sampling. \ More explicitly, Aaronson
and Arkhipov showed that if the PGC holds, then%
\begin{equation}
\mathsf{SampBPP}=\mathsf{SampBQP\Longrightarrow P}^{\#\mathsf{P}}%
=\mathsf{BPP}^{\mathsf{NP}}. \label{imp}%
\end{equation}
They went on to propose a program for proving the PGC, by exploiting the
random self-reducibility of the permanent. \ On the other hand, Aaronson and
Arkhipov also explained in detail why new ideas would be needed to complete
that program, and the challenge remains open.

Subsequently, Bremner, Montanaro, and Shepherd \cite{bms,bms2} gave
analogous $\#\mathsf{P}$-hardness conjectures that, if true, would
\textit{also} imply the implication (\ref{imp}), by going through the
$\mathsf{IQP}$\ model rather than through $\BosonSampling$.

Meanwhile, nearly two decades ago, Fortnow and Rogers \cite{fr}\ exhibited an
oracle relative to which $\mathsf{P}=\mathsf{BQP}$\ and yet the polynomial
hierarchy is infinite. \ In other words, they showed that any proof of the
implication%
\[
\mathsf{P}=\mathsf{BQP}\Longrightarrow\mathsf{PH}\text{ collapses}%
\]
would have to be non-relativizing. \ Unfortunately, their construction was
extremely specific to languages (i.e., total Boolean
functions), and didn't even rule out the possibility that the implication%
\[
\mathsf{P{}romiseBPP}=\mathsf{P{}romiseBQP}\Longrightarrow\mathsf{PH}\text{
	collapses}%
\]
could be proven in a relativizing way. \ Thus, Aaronson and Arkhipov
\cite[see Section 10]{aark}\ raised the question of which quantum supremacy
theorems hold relative to all oracles.

In Section \ref{sec:non-relativizing}, we fill in the final piece needed to resolve their question, by constructing an oracle
$A$\ relative to which $\mathsf{SampBPP}=\mathsf{SampBQP}$\ and yet
$\mathsf{PH}$\ is infinite. \ In other words, we show that \textit{any strong
	supremacy theorem for quantum sampling, along the lines of what Aaronson and
	Arkhipov \cite{aark} and Bremner, Montanaro, and Shepherd \cite{bms,bms2} were
	seeking, must use non-relativizing techniques}. \ In that respect, the
situation with approximate sampling is extremely different from that with
exact sampling.

Perhaps it's no surprise that one would need non-relativizing techniques to
prove a strong quantum supremacy theorem. \ In fact, Aaronson and Arkhipov
\cite{aark}\ were originally led to study $\BosonSampling$ precisely because of
the connection between bosons and the permanent function, and the hope that
one could therefore exploit the famous non-relativizing properties of the
permanent to prove hardness. \ All the same, this is the first time
we have explicit confirmation that non-relativizing techniques will be
needed.\bigskip

\textbf{Maximal Quantum Supremacy for Black-Box Sampling and Relation
	Problems.} \ In Section \ref{sec:sepa-samp-relation}, we turn our attention to the black-box
model, and specifically to the question: \textit{what are the largest possible
	separations between randomized and quantum query complexities for any
	approximate sampling or relation problem?} \ Here we settle another open
question. \ In 2015, Aaronson and Ambainis
\cite{aa:forrelation}\ studied $\fsampling$, in which we're given
access to a Boolean function $f:\left\{  0,1\right\}  ^{n}\rightarrow\left\{
0,1\right\}  $, and the goal is to sample a string $z$ with probability
$\widehat{f}\left(  z\right)  ^{2}$, where $\widehat{f}$ is the Boolean
Fourier transform of $f$, normalized so that $\sum_z \widehat{f}\left(
z\right)  ^{2}=1$. \ This problem is trivially solvable by a quantum algorithm
with only $1$ query to $f$. \ By contrast, Aaronson and Ambainis showed that
there exists a constant $\varepsilon>0$\ such that any classical algorithm
that solves $\fsampling$, to accuracy $\varepsilon$\ in variation
distance, requires $\Omega\left(  2^{n}/n\right)  $\ queries to $f$. \ They
conjectured that this lower bound was tight.

Here we refute that conjecture, by proving a $\Omega\left(  2^{n}\right)
$\ lower bound on the randomized query complexity of $\fsampling$, as long
as $\varepsilon$\ is sufficiently small (say, $\frac{1}{40000}$). \ This
implies that, for approximate sampling problems, the gap between quantum and
randomized query complexities can be as large as imaginable: namely,
$\mathit{1}$\textit{ versus linear (!)}.\footnote{We have learned (personal communication) that recently, and independently of us,
Ashley Montanaro has obtained a communication complexity result that implies this result as a corollary.} \ This sharply contrasts with the
case of partial Boolean functions, for which
Aaronson and Ambainis \cite{aa:forrelation} showed that any $N$-bit problem
solvable with $k$ quantum queries is also solvable with $O\left(
N^{1-1/2k}\right)  $\ randomized queries, and hence a constant versus linear
separation is impossible. \ Thus, our result helps once again to underscore the
advantage of sampling problems over decision problems for quantum supremacy experiments. \ Given the
extremely close connection between $\fsampling$ and the $\mathsf{IQP}%
$\ model \cite{bjs}, our result also provides
some evidence that classically simulating an $n$-qubit $\mathsf{IQP}%
$\ circuit, to within constant error in variation distance, is about as hard
as can be: it might literally require $\Omega\left(  2^{n}\right)  $\ time.

Aaronson and Ambainis \cite{aa:forrelation}\ didn't directly address the
natural relational version of $\fsampling$, which Aaronson \cite{aar:ph}%
\ had called $\ffishing$ in 2009. \ In $\ffishing$, the goal
is to output any string $z$\ such that $\widehat{f}\left(  z\right)
^{2}\ge 1$, with nontrivial success probability. \ Unfortunately, the best lower bound on the
randomized query complexity of $\ffishing$ that follows from
\cite{aar:ph}\ has the form $2^{n^{\Omega\left(  1\right)  }}$. \ As a further
contribution, in Section \ref{sec:sepa-samp-relation}\ we give a\ lower bound of $\Omega\left(
2^{n}/n\right)  $\ on the randomized query complexity of $\ffishing$,
which both simplifies and subsumes the $\Omega\left(  2^{n}/n\right)  $\ lower
bound for $\fsampling$ by Aaronson and Ambainis \cite{aa:forrelation}
(which, of course, we also improve to $\Omega(2^n)$ in this paper).\bigskip

\textbf{Quantum Supremacy Relative to Efficiently-Computable Oracles.} \ In
Section \ref{sec:oracle-sepa-ppoly}, we ask a new question: when proving quantum supremacy
theorems, can we \textquotedblleft interpolate\textquotedblright\ between the
black-box setting of Sections \ref{sec:non-relativizing}\ and \ref{sec:sepa-samp-relation}, and the
non-black-box setting of Sections \ref{sec:proposal}\ and \ref{sec:new-algorithms-circuits}? \ In particular,
what happens if we consider quantum sampling algorithms that can access an
oracle, \textit{but} we impose a constraint that the oracle has to be
\textquotedblleft physically realistic\textquotedblright? \ One natural
requirement here is that the oracle function $f$\ be computable in the class
$\mathsf{P/poly}$:\footnote{More broadly, we could let $f$\ be computable in
	$\mathsf{BQP/poly}$, but this doesn't change the story too much.} in other
words, that there are polynomial-size circuits for $f$, which we imagine that
our sampling algorithms (both quantum and classical) can call as subroutines.
\ If the sampling algorithms \textit{also} had access to explicit descriptions
of the circuits, then we'd be back in the computational setting, where we
already know that there's no hope at present of proving quantum supremacy
unconditionally. \ But what if our sampling algorithms
know only that small circuits for $f$ exist, without
knowing what they are? \ Could quantum supremacy be proven unconditionally
\textit{then}?

We give a satisfying answer to this question. \ First, by adapting constructions
due to Zhandry \cite{zhandry2012construct} and (independently) Servedio and Gortler \cite{servediogortler},
we show that if one-way functions exist, then there are oracles $A\in \mathsf{P/poly}$\ such that $\mathsf{BPP}^A\neq\mathsf{BQP}^A$, and indeed even $\mathsf{BQP}^A \not\subset\mathsf{SZK}^A$.
\ (Here and later, the one-way functions only need to be hard to invert
classically, not quantumly.)

Note that, in the unrelativized world, there
seems to be no hope at present of proving $\mathsf{BPP}\neq\mathsf{BQP}$ under
any hypothesis nearly as weak as the existence of one-way functions. \ Instead
one has to assume the one-wayness of extremely \textit{specific} functions,
for example those based on factoring or discrete log.

Second, and more relevant to near-term experiments, we show that if there
exist one-way functions that take at least subexponential time to invert, then there are
Boolean functions $f\in\mathsf{P/poly}$\ such that approximate $\fsampling$ on those $f$'s requires classical exponential time. \ In other words:
within our \textquotedblleft physically realistic oracle\textquotedblright%
\ model, there are feasible-looking quantum supremacy experiments, along the
lines of the $\mathsf{IQP}$\ proposal \cite{bjs},
such that a very standard and minimal cryptographic assumption is enough to
prove the hardness of simulating those experiments classically.

Third, we show that the above two results are essentially optimal, by proving
a converse result: that even in our $\mathsf{P/poly}$ oracle model,
\textit{some} computational assumption is still needed to prove quantum supremacy.
\ The precise statement is this: if $\mathsf{SampBPP}=\mathsf{SampBQP}$ and
$\mathsf{NP}\subseteq\mathsf{BPP}$, then $\mathsf{SampBPP}^{A}%
=\mathsf{SampBQP}^{A}$\ for all $A\in\mathsf{P/poly}$. \ Or equivalently: if
we want to separate quantum from classical approximate sampling relative to
efficiently computable\ oracles, then we need to assume \textit{something}
about the unrelativized world: either $\mathsf{SampBPP}\neq\mathsf{SampBQP}%
$\ (in which case we wouldn't even need an oracle), or else $\mathsf{NP}%
\not \subset \mathsf{BPP}$\ (which is closely related to the assumption we
\textit{do} make, namely that one-way functions exist).

So to summarize, we've uncovered a \textquotedblleft smooth
tradeoff\textquotedblright\ between the model of computation and the hypothesis needed
for quantum supremacy. \ Relative to \textit{some} oracle (and even a random
oracle), we can prove $\mathsf{SampBPP}\neq\mathsf{SampBQP}$
unconditionally. \ Relative to some efficiently computable oracle, we can
prove $\mathsf{SampBPP}\neq\mathsf{SampBQP}$, but only under a weak
computational assumption, like the existence of one-way functions. \ Finally,
with no oracle, we can currently prove $\mathsf{SampBPP}\neq\mathsf{SampBQP}%
$\ only under special assumptions, such as factoring being hard, or the
permanents of Gaussian matrices being hard to approximate in $\mathsf{BPP}%
^{\mathsf{NP}}$, or our QUATH assumption. \ Perhaps eventually, we'll be able to prove
$\mathsf{SampBPP}\neq\mathsf{SampBQP}$\ under the sole assumption that
$\mathsf{PH}$\ is infinite, which would be a huge step forward---but at any rate
we'll need \textit{some} separation of classical complexity classes.\footnote{Unless, of course, someone
	were to separate $\mathsf{P}$ from $\mathsf{PSPACE}$ unconditionally!}

One last remark: the idea of comparing complexity classes relative
to\ $\mathsf{P/poly}$\ oracles\ seems quite natural even apart from its
applications to quantum supremacy. \ So in Appendix \ref{sec:other_ppoly}, we take an
initial stab at exploring the implications of that idea for other central questions in complexity theory.
\ In particular, we prove the surprising result there that $\mathsf{P}%
^{A}=\mathsf{BPP}^{A}$\ for all oracles $A\in\mathsf{P/poly}$, \textit{if and
	only if} the derandomization hypothesis of Impagliazzo and Wigderson
\cite{iw}\ holds (i.e., there exists a function in $\mathsf{E}$\ with
$2^{\Omega\left(  n\right)  }$\ circuit complexity). \ In our view, this helps to clarify
Impagliazzo and Wigderson's theorem itself, by showing precisely in what way their circuit lower bound hypothesis
is stronger than the desired conclusion $\mathsf{P}=\mathsf{BPP}$. \ We also show that, if
there are quantumly-secure one-way functions, then there exists an
oracle $A\in\mathsf{P/poly}$\ such that $\mathsf{SZK}^{A}\not \subset
\mathsf{BQP}^{A}$.


\subsection{Techniques\label{TECHNIQUES}}

In our view, the central contributions of this work lie in the creation of new
questions, models, and hardness assumptions (such as QUATH and quantum
supremacy relative to $\mathsf{P/poly}$\ oracles), as well as in basic
observations that somehow weren't made before (such as the sum-products
algorithm for simulating quantum circuits)---all of it motivated by the goal
of using complexity theory to inform ongoing efforts in experimental
physics to test the Extended Church-Turing Thesis. \ While some of our proofs
are quite involved, by and large the proof techniques are ones that will be
familiar to complexity theorists. \ Even so, it seems appropriate to say a few
words about techniques here.

To prove, in Section \ref{sec:proposal}, that \textquotedblleft if QUATH is true, then
HOG is hard,\textquotedblright\ we give a fairly straightforward
reduction:\ namely, we assume the existence of a polynomial-time classical
algorithm to find high-probability outputs of a given quantum circuit $C$.
\ We then use that algorithm (together with a random self-reduction trick) to
guess the magnitude of a particular transition amplitude, such as
$\left\langle 0^n |C|0^n \right\rangle $, with probability slightly
better than chance, which is enough to refute QUATH.

One technical step is
to show that, with $\Omega(1)$ probability, the distribution over $n$-bit strings sampled by a random
quantum circuit $C$ is far from the uniform distribution. \ But not only can this be done, we
show that it can be done by examining only the very last gate of $C$, and
ignoring all other gates! \ A challenge that we leave open is to improve this, to show
that the distribution sampled by $C$ is far from uniform, not merely with $\Omega(1)$ probability,
but with $1-1/\exp(n)$ probability. \ In Appendix \ref{sec:tukareta}, we present numerical evidence
for this conjecture, and indeed for a stronger conjecture, that the probabilities appearing in the output
distribution of a random quantum circuit behave like independent, exponentially-distributed random variables. \ (We note that Brandao, Harrow and Horodecki \cite{brandaoharrow} recently proved a closely-related result, which unfortunately is not quite strong enough for our purposes.)

In Section \ref{sec:new-algorithms-circuits}, to give our polynomial-space, $d^{O\left(  n\right)  }%
$-time classical algorithm for simulating an $n$-qubit, depth-$d$\ quantum
circuit $C$, we use a simple recursive strategy, reminiscent of Savitch's
Theorem. \ Namely, we slice the circuit into two layers, $C_{1}$\ and $C_{2}$,
of depth $d/2$\ each, and then express a transition amplitude $\left\langle
x|C|z\right\rangle $\ of interest to us as%
\[
\left\langle x|C|z\right\rangle =\sum_{y\in\left\{  0,1\right\}  ^{n}%
}\left\langle x|C_{1}|y\right\rangle \left\langle y|C_{2}|z\right\rangle .
\]
We then compute each $\left\langle x|C_{1}|y\right\rangle $\ and $\left\langle
y|C_{2}|z\right\rangle $\ by recursively slicing $C_{1}$\ and $C_{2}$\ into
layers of depth $d/4$\ each, and so on. \ What takes more work is to obtain a
further improvement if $C$ has only nearest-neighbor interactions on a grid
graph---for that, we use a more sophisticated divide-and-conquer approach---and
also to interpolate our recursive algorithm with the $2^{n}$-space
Schr\"{o}dinger simulation, in order to make the best possible use of whatever
memory is available.

Our construction, in Section \ref{sec:non-relativizing}, of an oracle relative to which
$\mathsf{SampBPP}=\mathsf{SampBQP}$\ and yet $\mathsf{PH}$\ is infinite\ involves significant technical difficulty. \ As a first step,
we can use a $\mathsf{PSPACE}$\ oracle to collapse $\mathsf{SampBPP}$\ with
$\mathsf{SampBQP}$, and then use one of many known oracles (or, by the recent breakthrough of Rossman, Servedio,
and Tan \cite{rst}, even a \textit{random} oracle) to make $\mathsf{PH}%
$\ infinite. \ The problem is that, if we do this in any na\"{\i}ve way, then
the oracle that makes\ $\mathsf{PH}$\ infinite will also re-separate
$\mathsf{SampBPP}$\ and $\mathsf{SampBQP}$, for example because of the approximate $\fsampling$ problem. \ Thus, we need to \textit{hide} the oracle that makes $\mathsf{PH}%
$\ infinite, in such a way that a $\mathsf{PH}$\ algorithm can still find the
oracle (and hence, $\mathsf{PH}$\ is still infinite), but a $\mathsf{SampBQP}%
$\ algorithm can't find it with any non-negligible probability---crucially,
not even if the $\mathsf{SampBQP}$\ algorithm's input $x$\ provides a clue
about the oracle's location. \ Once one realizes that these are the
challenges, one then has about seven pages of work to ensure that
$\mathsf{SampBPP}$\ and $\mathsf{SampBQP}$\ remain equal, relative to the
oracle that one has constructed. \ Incidentally, we know that this equivalence can't possibly hold for \textit{exact} sampling, so \textit{something} must force small errors
to arise when the $\mathsf{SampBPP}$\ algorithm simulates the $\mathsf{SampBQP}$\ one. \ That something is basically the tiny probability that the quantum algorithm will succeed at finding the hidden oracle,
which however can be upper-bounded using quantum-mechanical linearity.

In Section \ref{sec:sepa-samp-relation}, to prove a $\Omega\left(  2^{n}\right)  $\ lower bound
on the classical query complexity of approximate $\mathsf{Fourier}\newline \mathsf{Sampling}$, we use the
same basic strategy that Aaronson and Ambainis \cite{aa:forrelation}\ used to
prove a $\Omega\left(  2^{n}/n\right)  $\ lower bound, but with a much more
careful analysis. \ Specifically, we observe that any $\fsampling$
algorithm would also yield an algorithm whose probability of accepting, while always small,
is extremely sensitive to some specific Fourier coefficient, say
$\widehat{f}\left(  0\cdots0\right)  $. \ We
then lower-bound the randomized query complexity of accepting with the required sensitivity to $\widehat{f}\left(  0\cdots0\right)  $, taking
advantage of the fact that $\widehat{f}\left(  0\cdots0\right)  $\ is simply
proportional to $\sum_x f\left(  x\right)  $, so that all $x$'s can be
treated symmetrically. \ Interestingly, we also give a different, much simpler
argument that yields a $\Omega\left(  2^{n}/n\right)  $\ lower bound on the
randomized query complexity of $\ffishing$, which then immediately implies
a $\Omega\left(  2^{n}/n\right)  $\ lower bound for $\fsampling$ as well.
\ However, if we want to improve the bound to $\Omega\left(  2^{n}\right)  $,
then the original argument that Aaronson and Ambainis \cite{aa:forrelation}
used to prove $\Omega\left(  2^{n}/n\right)  $\ seems to be needed.

In Section \ref{sec:oracle-sepa-ppoly}, to prove that one-way functions imply the existence of
an oracle $A\in\mathsf{P/poly}$\ such that $\mathsf{P}^{A}\neq\mathsf{BQP}%
^{A}$, we adapt a construction that was independently proposed by Zhandry
\cite{zhandry2012construct} and by Servedio and Gortler \cite{servediogortler}. \ In this
construction, we first use known reductions \cite{hill,ggm} to convert a
one-way function into a classically-secure pseudorandom permutation, say $\sigma$. \ We then define a
new function by $g_r(x) := \sigma(x \operatorname{mod}r )$,
where $x$\ is interpreted as an integer written in binary, and $r$ is a hidden
period. \ Finally, we argue that either Shor's algorithm \cite{shor} leads to a
quantum advantage over classical algorithms in finding the period of $g_r$,
or else $g_r$\ was not pseudorandom, contrary to assumption. \ To show
that subexponentially-secure one-way functions imply the existence of an oracle
$A\in\mathsf{P/poly}$\ relative to which $\fsampling$ is classically hard,
we use similar reasoning. \ The main difference is that now, to construct a
distinguisher against a pseudorandom function $f$, we need classical
exponential time just to \textit{verify} the outputs of a claimed
polynomial-time classical algorithm for $\fsampling$ $f$---and that's
why we need to assume $2^{n^{\Omega(1)}}$ security.

Finally, to prove that
$\mathsf{SampBPP}=\mathsf{SampBQP}$ and $\mathsf{NP}\subseteq\mathsf{BPP}$
imply $\mathsf{SampBPP}^{A}=\mathsf{SampBQP}^{A}$\ for all $A\in
\mathsf{P/poly}$, we design a step-by-step classical simulation of a quantum
algorithm, call it $Q$, that queries an oracle $A\in\mathsf{P/poly}$. \ We use
the assumption\ $\mathsf{SampBPP}=\mathsf{SampBQP}$\ to sample from the
probability distribution over queries to $A$\ that $Q$\ makes\ at any given time
step. \ Then we use the assumption $\mathsf{NP}\subseteq\mathsf{BPP}$\ to guess a
function $f\in\mathsf{P/poly}$\ that's consistent with $n^{O\left(  1\right)
}$ sampled classical queries to $A$. \ Because of the limited number of
functions in $\mathsf{P/poly}$, standard sample complexity bounds for
PAC-learning imply that any such $f$\ that we guess will probably agree with
the \textquotedblleft true\textquotedblright\ oracle $A$ on most inputs. \
Quantum-mechanical linearity then implies that the rare disagreements
between $f$\ and $A$ will have at most a small effect on the future behavior
of $Q$.

\section{Preliminaries}
\label{sec:prelim}

For a positive integer $n$, we use $[n]$ to denote the integers from $1$ to $n$.  Logarithms are base $2$.

\subsection{Quantum Circuits}
\label{sec:quantum-circuits}

We now introduce some notations for quantum circuits, which will be used throughout this paper.

In a quantum circuit, without loss of generality, we assume all gates are unitary and acting on exactly two qubits each\footnote{Except for oracle gates, which may act on any number of qubits.}.

Given a quantum circuit $C$, slightly abusing notation, we also use $C$ to denote the unitary operator induced by $C$. Suppose there are $n$ qubits and $m$ gates in $C$; then we index the qubits from $1$ to $n$. We also index gates from $1$ to $m$ in chronological order for convenience.

For each subset $S \subseteq [n]$ of the qubits, let $\heb_S$ be the Hilbert space corresponding to the qubits in $S$, and $I_S$ be the identity operator on $\heb_S$. Then the unitary operator $U_i$ for the $i$-th gate can be written as $U_i := O_i \otimes I_{[n]\setminus\{a_i,b_i\}}$, in which $O_i$ is a unitary operator on $\heb_{\{a_i,b_i\}}$ (the Hilbert space spanned by the qubits $a_i$ and $b_i$), and $I_{[n]\setminus\{a_i,b_i\}}$ is the identity operator on the other qubits.

We say that a quantum circuit has depth $d$, if its gates can be partitioned into $d$ layers (in chronological order), such that the gates in each layer act on disjoint pairs of qubits. Suppose the $i$-th layer consists of the gates in $[L_i,R_i]$. We define $C_{[r \leftarrow l]} = U_{R_{r}} \cdot U_{R_{r}-1} \dotsc U_{L_{l}+1} \cdot U_{L_{l}}$, that is, the sub-circuit between the $l$-th layer and the $r$-th layer.

\subsubsection*{Base Graphs and Grids}

In Sections \ref{sec:proposal} and \ref{sec:new-algorithms-circuits}, we will sometimes assume {\em locality} of a given quantum circuit. \ To formalize this notion, we define the {\em base graph} of a quantum circuit.

\begin{defi}
	Given a quantum circuit $C$ on $n$ qubits, its base graph  $G_C=(V,E)$ is an undirected graph defined by $V=[n]$, and
$$ E = \{ (a,b) \betw \text{there is a quantum gate that acts on qubits $a$ and $b$.} \}.$$
\end{defi}

We will consider a specific kind of base graph, the grids.

\begin{defi}\label{defi:grid}
	The grid $G$ of size $H \times W$ is a graph with vertices $V = \{ (x,y) \betw x \in [H], y \in [W] \}$ and edges $E = \{(a,b) \betw |a-b|_1 = 1, a\in V,b \in V \}$, and we say that grid $G$ has $H$ rows and $W$ columns.
\end{defi}

\subsection{Complexity Classes for Sampling Problems}

\subsubsection*{Definitions for $\SampBPP$ and $\SampBQP$}

We adopt the following definition for sampling problems from \cite{aaronson2014equivalence}.

\begin{defi}
	[Sampling Problems, $\SampBPP$, and $\SampBQP$]A \textit{sampling problem} $S$ is a collection of probability distributions
	$\left(  \mathcal{D}_{x}\right)  _{x\in\left\{  0,1\right\}  ^{\ast}}$, one
	for each input string $x\in\left\{  0,1\right\}  ^{n}$, where $\mathcal{D}%
	_{x}$\ is a distribution over $\left\{  0,1\right\}  ^{p\left(  n\right)  }$,
	for some fixed polynomial $p$. \ Then $\SampBPP$ is the class of
	sampling problems $S=\left(  \mathcal{D}_{x}\right)  _{x\in\left\{
		0,1\right\}  ^{\ast}}$ for which there exists a probabilistic polynomial-time
	algorithm $B$\ that, given $\left\langle x,0^{1/\varepsilon}\right\rangle $ as
	input, samples from a probability distribution $\mathcal{C}_{x}$\ such that
	$\left\Vert \mathcal{C}_{x}-\mathcal{D}_{x}\right\Vert \leq\varepsilon$.
	\ $\SampBQP$\ is defined the same way, except that $B$\ is a quantum
	algorithm rather than a classical one.
	
	Oracle versions of these classes can also be defined in the natural way.
\end{defi}

\subsubsection*{A Canonical Form of $\SampBQP$ Oracle Algorithms}
\label{sec:canonical-SampBQP}

To ease our discussion about $\SampBQP^{\oracle}$, we describe a canonical form of $\SampBQP$ oracle algorithms. Any other reasonable definitions of $\SampBQP$ oracle algorithms (like with quantum oracle Turing machines) can be transformed into this form easily.

Without loss of generality, we can assume a $\SampBQP$ oracle algorithm $M$ with oracle access to $\oracle_1,\oracle_2,\dotsc,\oracle_k$ ($k$ is a universal constant) acts in three stages, as follows.

\begin{enumerate}
	\item
	Given an input $\langle x,0^{1/\varepsilon}\rangle$, $M$ first uses a classical routine (which does not use the oracles) to output a quantum circuit $C$ with $p(n,1/\varepsilon)$ qubits and $p(n,1/\varepsilon)$ gates in polynomial time, where $p$ is a fixed polynomial. \ Note that $C$ can use the $\oracle_1,\oracle_2,\dotsc,\oracle_k$ gates in addition to a universal set of quantum gates.
	\item
	Then $M$ runs the outputted quantum circuit with the initial state $\spz{0}^{\otimes p(n,1/\varepsilon)}$, and measures all the qubits to get an outcome $z$ in $\{0,1\}^{p(n,1/\varepsilon)}$.
	
	\item
	Finally, $M$ uses another classical routine $A^{\mathsf{output}}$ (which does not use the oracles) on the input $z$, to output its final sample $A^{\mathsf{output}}(z) \in \{0,1\}^*$.
\end{enumerate}

Clearly, $M$ solves different sampling problems (or does not solve any sampling problem at all) given different oracles $\oracle_1,\oracle_2,\dotsc,\oracle_k$. \ Therefore, we use $M^{\oracle_1,\oracle_2,\dotsc,\oracle_k}$ to indicate the particular algorithm when the oracles are $\oracle_1,\oracle_2,\dotsc,\oracle_k$.

\subsection{Distinguishing Two Pure Quantum States}

We also need a standard result for distinguishing two pure quantum states.

\begin{theo}[Helstrom's decoder for two pure states]
	The maximum success probability for distinguishing two pure quantum states $\spz{\varphi_0}$ and $\spz{\varphi_1}$ given with prior probabilities $\pi_0$ and $\pi_1$, is given by
	$$
	p_{succ} = \frac{1+\sqrt{1-4\pi_0\pi_1F}}{2},
	$$
\noindent where $F := |\langle\varphi_0\spz{\varphi_1}|^2$ is the fidelity between the two states.
\end{theo}

We'll also need that for two similar quantum states, the distributions induced by measuring them are close.

\begin{cor}\label{cor:close-dist}
	Let $\spz{\varphi_0}$ and $\spz{\varphi_1}$ be two pure quantum state such that $|\spz{\varphi_0} - \spz{\varphi_1}| \le \varepsilon$. For a quantum state $\varphi$, define $\distr(\varphi)$ be the distribution on $\{0,1\}^*$ induced by some quantum sampling procedure, we have
	
	$$
	\|\distr(\varphi_0) - \distr(\varphi_1)\| \le \sqrt{2\varepsilon}.
	$$
\end{cor}
\begin{proof}
	Fix prior probabilities $\pi_0 = \pi_1 = \frac{1}{2}$.
	
	Note that we have a distinguisher of $\spz{\varphi_0}$ and $\spz{\varphi_1}$ with success probability $\frac{1+\|\distr(\varphi_0) - \distr(\varphi_1)\|}{2}$ by invoking that quantum sampling procedure.
	
	By the assumption, $|\rpz{\varphi_0}\spz{\varphi_1}| = |\rpz{\varphi_0} \cdot (\spz{\varphi_0} + (\spz{\varphi_1}-\spz{\varphi_0})| \ge 1 -\varepsilon$, hence $F = |\langle\varphi_0\spz{\varphi_1}|^2 \ge (1-\varepsilon)^2$. So we have
	
	$$
	\frac{1+\|\distr(\varphi_0) - \distr(\varphi_1)\|}{2} \le \frac{1+\sqrt{1-(1-\varepsilon)^2}}{2}
	$$
	
	\noindent This implies $\|\distr(\varphi_0) - \distr(\varphi_1)\|_1 \le \sqrt{1-(1-\varepsilon)^2} = \sqrt{2\varepsilon - \varepsilon^2} \le \sqrt{2\varepsilon}$.
	
\end{proof}

\subsection{A Multiplicative Chernoff Bound}

\begin{lemma}
	Suppose $X_1,X_2,\dotsc,X_n$ are independent random variables taking values in $[0,1]$. Let $X$ denote their sum and let $\mu = \Ex[X]$. Then for any $\delta > 1$, we have
	
	$$
	\Pr[X \ge (1+\delta)\mu] \le e^{-\frac{\delta\mu}{3}}.
	$$
\end{lemma}

\begin{cor}\label{cor:mul-bound}
	For any $0 < \tau$, suppose $X_1,X_2,\dotsc,X_n$ are independent random variables taking values in $[0,\tau]$. Let $X$ denote their sum and let $\mu = \Ex[X]$. Then for any $\delta > 1$, we have
	
	$$
	\Pr[X \ge (1+\delta)\mu] \le e^{-\frac{\delta\mu}{3\tau}}.
	$$
\end{cor}
\begin{proof}
	Replace each $X_i$ by $X_i/\tau$ and apply the previous lemma.
\end{proof}

\newcommand{\murand}{\mu_{\mathsf{rand}}}
\newcommand{\muharr}{\mu_{\mathsf{Haar}}}
\newcommand{\mugrid}{\mu_{\mathsf{grid}}}
\newcommand{\nugrid}{\nu_{\mathsf{grid}}}
\newcommand{\QGuess}{\mathsf{QGuess}}
\newcommand{\QTest}{\mathsf{QTest}}
\newcommand{\SHA}{\mathsf{SHA}}
\newcommand{\prand}{p_{\mathsf{rand}}}
\newcommand{\pquantum}{p_{\mathsf{quantum}}}
\newcommand{\uphalf}{\mathsf{uphalf}}
\newcommand{\probs}{\mathsf{probList}}

\section{The Hardness of Quantum Circuit Sampling}
\label{sec:proposal}

We now discuss our random quantum circuit proposal for demonstrating quantum supremacy.

\subsection{Preliminaries}
We first introduce some notations. \ We use $\mathbb{U}(N)$ to denote the group of $N \times N$ unitary matrices, $\muharr^N$ for the Haar measure on $\mathbb{U}(N)$, and $\murand^N$ for the Haar measure on $N$-dimensional pure states.

For a pure state $\spz{u}$ on $n$ qubits, we define $\probs(\spz{u})$ to be the list consisting of $2^n$ numbers, $|\langle u\spz{x}|^2$ for each $x \in \{0,1\}^n$.

Given $N$ real numbers $a_1,a_2,\dotsc,a_N$, we use $\uphalf(a_1,a_2,\dotsc,a_N)$ to denote the sum of the largest $N/2$ numbers among them, and we let
$$
\adv(\spz{u}) = \uphalf(\probs(\spz{u})).
$$

\noindent Finally, we say that an output $z \in \{0,1\}^n$ is {\em heavy} for a quantum circuit $C$, if it is greater than the median of $\probs(C\spz{0^n})$.

\subsection{Random quantum circuit on grids}

Recall that we assume a quantum circuit consists of only $2$-qubit gates. \ Our random quantum circuit on grids of $n$ qubits and $m$ gates (assuming $m \ge n$) is generated as follows (though the basic structure of our hardness argument will not be very sensitive to details, and would also work for many other circuit ensembles):

\begin{itemize}
	\item All the qubits are arranged as a $\sqrt{n} \times \sqrt{n}$ grid (see Definition~\ref{defi:grid}), and a gate can only act on two adjacent qubits.
	\item For each $t \in [m]$ with $t \le n$, we pick the $t$-th qubit and a random neighbor of it.\footnote{The purpose here is to make sure that there is a gate on every qubit.}
	\item For each $t \in [m]$ with $t > n$, we pick a uniform random pair of adjacent qubits in the grid $\sqrt{n} \times \sqrt{n}$.
	\item Then, in either case, we set the $t$-th gate to be a unitary drawn from $\muharr^{4}$ acting on these two qubits.
\end{itemize}

Slightly abusing notation, we use $\mugrid^{n,m}$ to denote both the above distribution on quantum circuits and the distribution on $\mathbb{U}(2^n)$ induced by it.

\subsubsection*{Conditional distribution $\nugrid$}

For convenience, for a quantum circuit $C$, we abbreviate $\adv(C\spz{0^n})$ as $\adv(C)$. \ Consider a simple quantum algorithm which measures $C\spz{0^n}$ in the computational basis to get an output $z$. \ Then by definition, $\adv(C)$ is simply the probability that $z$ is heavy for $C$.

We want that, when a quantum circuit $C$ is drawn, $\adv(C)$ is {\em large} (that is, bounded above $1/2$), and therefore the simple quantum algorithm has a substantial advantage on generating a heavy output, compared with the trivial algorithm of guessing a random string.

For convenience, we also consider the following conditional distribution $\nugrid^{n,m}$: it keeps drawing a circuit $C \leftarrow \mugrid^{n,m}$ until the sample circuit $C$ satisfies $\adv(C) \ge 0.7$.

\subsubsection*{Lower bound on $\adv(C)$}

We need to show that a circuit $C$ drawn from $\nugrid^{n,m}$ has a large probability of having $\adv(C) \ge 0.7$. \ In order to show that, we give a cute and simple lemma, which states that the {\em expectation} of $\adv(C)$ is {\em large}. \ Surprisingly, its proof only makes use of the randomness introduced by the {\em very last gate}!
\begin{lemma}\label{lemma:adv-large-exp}
	For $n \ge 2$ and $m \ge n$
	$$
	\Ex_{C \leftarrow \mugrid^{n,m}}[\adv(C)] \ge \frac{5}{8}.
	$$
\end{lemma}

In fact, we conjecture that $\adv(C)$ is large with an {\em overwhelming} probability.

\begin{conjecture}\label{conj:adv-large}
	For $n \ge 2$ and $m \ge n^2$, and for all constants $\eps > 0$,
	$$
	\Pr_{C \leftarrow \mugrid^{n,m}}\left[\adv(C) < \frac{1+\ln 2}{2} - \eps\right] < \exp\left\{-\Omega(n) \right\}.
	$$
\end{conjecture}

We give some numerical simulation evidence for Conjecture~\ref{conj:adv-large} in Appendix~\ref{sec:tukareta}.

\begin{rem}
	Assuming Conjecture~\ref{conj:adv-large}, in practice, one can sample from $\nugrid$ by simply sampling from $\mugrid$, the uniform distribution over circuits\textemdash doing so only introduces an error probability of $\exp \{-\Omega(n)\}$.
\end{rem}

\subsection{The HOG Problem}

Now we formally define the task in our quantum algorithm proposal.

\begin{prob}[HOG, or Heavy Output Generation]
	Given a random quantum circuit $C$ from $\nugrid^{n,m}$ for $m \ge n^2$, generate $k$ binary strings $z_1,z_2,\dotsc,z_k$ in $\{0,1\}^n$ such that at least a $2/3$ fraction of $z_i$'s are heavy for $C$.
\end{prob}

The following proposition states that there is a simple quantum algorithm which solves the above problem with overwhelming probability.

\begin{prop}
	There is a quantum algorithm that succeeds at HOG with probability $1-\exp\{- \Omega(k) \}$.
\end{prop}
\begin{proof}
	The algorithm just simulates the circuit $C$ with initial state $\spz{0^n}$, then measures in the computational basis $k$ times independently to output $k$ binary strings.
	
	From the definition of $\nugrid$, we have $\adv(C) \ge 0.7 > 2/3$. \ So by a Chernoff bound, with probability $1-\exp\{ \Omega(k) \}$, at least a $2/3$ fraction of $z_i$'s are heavy for $C$, in which case the algorithm solves HOG.
\end{proof}

\subsection{Classical Hardness Assuming QUATH}

We now state our classical hardness assumption.

\begin{as}[QUATH, or the \textbf{Qua}ntum \textbf{Th}reshold assumption]
	There is no polynomial-time classical algorithm that takes as input a random
	quantum circuit $C \leftarrow \nugrid^{n,m}$ for $m \ge n^2$ and decides whether
	$0^{n}$ is heavy for $C$ with success probability $1/2+\Omega(2^{-n})$.
\end{as}

\begin{rem}
	Note that $1/2$ is the success probability obtained by always outputting either $0$ or $1$. \ Therefore, the above assumption means that no efficient algorithm can beat the trivial algorithm even by $\Omega(2^{-n})$.
\end{rem}

Next, we show that QUATH implies that no efficient classical algorithm can solve HOG.

\begin{theo}
	Assuming QUATH, no polynomial-time classical algorithm can solve HOG with probability at least $0.99$.
\end{theo}
\begin{proof}
	Suppose by contradiction that there is such a classical polynomial-time algorithm $A$. \ Using $A$, we will construct an algorithm to violate QUATH.
	
	The algorithm is quite simple. \ Given a quantum circuit $C \leftarrow \nugrid^{n,m}$, we first draw a uniform random string $z \in \{0,1\}^n$. Then for each $i$ such that $z_i = 1$, we apply a $\mathsf{NOT}$ gate on the $i$-th qubit. \ Note that this gate can be ``absorbed'' into the last gate acting on the $i$-th qubit in $C$. \ Hence, we still get a circuit $C'$ with $m$ gates. \ Moreover, it is easy to see that $C'$ is distributed exactly the same as $C$ even if conditioning on a particular $z$, and we have $ \rpz{0^n}C\spz{0^n} = \rpz{0^n}C'\spz{z} $, which means that $0^n$ is heavy for $C$ if and only if $z$ is heavy for $C'$.
	
	Next our algorithm runs $A$ on circuit $C'$ to get $k$ outputs $z_1,\dotsc,z_k$, and picks an output $z_{i^{\star}}$ among these $k$ outputs uniformly at random. \ If $z_{i^\star} = z$, then the algorithm outputs $1$; otherwise it outputs a uniform random bit.
	
	Since $A$ solves HOG with probability $0.99$, we have that each $z_k$ is heavy for $C'$ with probability at least $0.99 \cdot 2/3$.
	
	Now, since $z$ is a uniform random string, the probability that our algorithm decides correctly whether $z$ is heavy for $C'$ is
	\begin{align*}
	\Pr[z = z_{i^{\star}}] \cdot 0.99 \cdot \frac{2}{3} + \Pr[z \ne z_{i^{\star}}] \cdot 1/2 &= 2^{-n} \cdot 0.99 \cdot \frac{2}{3} + (1 - 2^{-n}) \cdot 1/2 \\
	&= \frac{1}{2} + \Omega(2^{-n}).
	\end{align*}
	
	\noindent But this contradicts QUATH, so we are done.
\end{proof}

\subsection{Proof for Lemma~\ref{lemma:adv-large-exp}}

\newcommand{\dev}{\mathsf{dev}}

We first need a simple lemma which helps us to lower bound $\adv(\spz{u})$.

For a pure quantum state $\spz{u}$, define
$$
\dev(\spz{u}) = \sum_{w \in \{0,1\}^n} \Big| |\langle u\spz{w}|^2 - 2^{-n} \Big|.
$$
In other words, $\dev(\spz{u})$ measures the {\em non-uniformity} of the distribution obtained by measuring $\spz{u}$ in the computational basis.

The next lemma shows that, when $\dev(\spz{u})$ is large, so is $\adv(\spz{u})$. \ Therefore, in order to establish Lemma~\ref{lemma:adv-large-exp}, it suffices to lower-bound $\dev(\spz{u})$.

\begin{lemma}\label{lemma:dev-to-adv}
	For a pure quantum state $\spz{u}$, we have
	$$
	\adv(\spz{u}) \ge \frac{1}{2} + \frac{\dev(u)}{4}.
	$$
\end{lemma}

We will also need the following technical lemma.

\begin{lemma}\label{lemma:random}
	Let $\spz{u} \leftarrow \murand^2$. \ Then
	$$
	\Ex_{\spz{u} \leftarrow \murand^2}\left[ \Big||\rpz{u} 0 \rangle|^2 - |\rpz{u} 1 \rangle|^2 \Big|\right] = 0.5.
	$$
\end{lemma}

The proofs of Lemma~\ref{lemma:dev-to-adv} and Lemma~\ref{lemma:random} are based on simple but tedious calculations, so we defer them to Appendix~\ref{sec:missing-proofs-proposal}.

Now we are ready to prove Lemma~\ref{lemma:adv-large-exp}.
\begin{proofof}{Lemma~\ref{lemma:adv-large-exp}}
	Surprisingly, our proof only uses the randomness introduced by the very last gate. \ That is, the claim holds even if there is an adversary who fixes all the gates except for the last one.

We use $I_n$ to denote the $n$-qubit identity operator.
	
	Let $C \leftarrow \mugrid^{n,m}$. \ From Lemma~\ref{lemma:dev-to-adv}, it suffices to show that
	$$
	\Ex_{C \leftarrow \mugrid^{n,m}}[\dev(C \spz{0^n})] \ge \frac{1}{2}.
	$$
	
	Suppose the last gate $U \leftarrow \muharr^{4}$ acts on qubits $a$ and $b$. Let the unitary corresponding to the circuit before applying the last gate be $V$, and $\spz{v} = V\spz{0^n}$. Now, suppose we apply another unitary $U_a$ drawn from $\muharr^{2}$ on the qubit $a$. It is not hard to see that $U$ and $ (U_a \otimes I_1) \cdot U$ are identically distributed. So it suffices to show that
	$$
	\Ex_{U \leftarrow \muharr^4, U_a \leftarrow \muharr^2} \left[ \adv\Big( (U_a \otimes I_{n-1})(U \otimes I_{n-2})\spz{v} \Big) \right] \ge 0.6.
	$$
	
	We are going to show that the above holds even for a fixed $U$. \ That is, fix a $U \in \mathbb{U}(4)$ and let $\spz{u} = U \otimes I_{n-2})\spz{v}$. \ Then we will prove that
	$$
	\Ex_{U_a \leftarrow \muharr^2} \left[ \dev\Big( (U_a \otimes I_{n-1})\spz{v}\Big) \right] \ge \frac{1}{2}.
	$$
	Without loss of generality, we can assume that $a$ is the last qubit. Then we write
	$$
	\spz{u} = \sum_{w \in \{0,1\}^n} a_{w} \spz{w},
	$$
	and
	$$
	\spz{z} = (U_a \otimes I_{n-1})\spz{u}.
	$$
	
	Now we partition the $2^n$ basis states into $2^{n-1}$ buckets, one for each string in $\{0,1\}^{n-1}$. \ That is, for each $p \in \{0,1\}^{n-1}$, there is a bucket that consists of basis states $\{ \spz{p0},\spz{p1} \}$. \ Note that since $U_a$ acts on the last qubit, only amplitudes of basis states in the same bucket can affect each other.
	
	For a given $p\in \{0,1\}^{n-1}$, if both $a_{p0}$ and $a_{p1}$ are zero, we simply ignore this bucket. \ Otherwise, we can define a quantum state
	$$
	\spz{t_p} = \frac{a_{p0} \spz{0} + a_{p1} \spz{1}}{\sqrt{|a_{p0}|^2 +　|a_{p1}|^2} },
	$$
	and
	$$
	\spz{z_p} = U_a \spz{t_p}.
	$$
	
	Clearly, we have $\rpz{z} p0 \rangle = \sqrt{|a_{p0}|^2 +　|a_{p1}|^2} \cdot \rpz{z_p} 0 \rangle$ and $\rpz{z} p1 \rangle = \sqrt{|a_{p0}|^2 +　|a_{p1}|^2} \cdot \rpz{z_p} 1 \rangle$. Plugging in, we have
	\begin{align*}
	&\Ex_{U_a \leftarrow \muharr^2}\left[ \Big||\rpz{z} p0 \rangle|^2 - 2^{-n} \Big| + \Big| |\rpz{z} p1 \rangle|^2 - 2^{-n} \Big| \right]\\
	\ge&\Ex_{U_a \leftarrow \muharr^2}\left[ \Big| |\rpz{z} p0 \rangle|^2 - |\rpz{z} p1 \rangle|^2 \Big|\right] \tag{triangle inequality}\\
	=& \left(|a_{p0}|^2 +　|a_{p1}|^2\right) \cdot \Ex_{U_a \leftarrow \muharr^2}\left[ \Big||\rpz{z_p} 0 \rangle|^2 - |\rpz{z_p} 1 \rangle|^2 \Big|\right].
	\end{align*}
	
	Now, since $\spz{t_p}$ is a pure state, and $U_a$ is drawn from $\muharr^{2}$, we see that $\spz{z_p}$ is distributed as a Haar-random pure state. \ So from Lemma~\ref{lemma:random}, we have
	
	$$
	\Ex_{U_a \leftarrow \muharr^2}\left[ \Big||\rpz{z_p} 0 \rangle|^2 - |\rpz{z_p} 1 \rangle|^2 \Big|\right] = 0.5.
	$$
	
	Therefore,
	$$
	\Ex_{U_a \leftarrow \muharr^2}\left[ \Big||\rpz{z} p0 \rangle|^2 - 2^{-n} \Big| + \Big| |\rpz{z} p1 \rangle|^2 - 2^{-n} \Big| \right] \ge  \frac{1}{2} \cdot \left(|a_{p0}|^2 +　|a_{p1}|^2\right).
	$$
	
	\noindent Summing up for each $p \in \{0,1\}^{n-1}$, we have
	$$
	\Ex_{U_a \leftarrow \muharr^2}[\dev(\spz{z})] \ge \frac{1}{2},
	$$
	which completes the proof.
	
\end{proofof}

\section{New Algorithms to Simulate Quantum Circuits}
\label{sec:new-algorithms-circuits}

In this section, we present two algorithms for simulating a quantum circuit with $n$ qubits and $m$ gates: one algorithm for arbitrary circuits, and another for circuits that act locally on grids. \ What's new about these algorithms is that they use both polynomial space and close to $\exp(n)$ time (but despite that, they don't violate the QUATH assumption from Section \ref{sec:proposal}, for the reason pointed out in Section \ref{OURCONT}). \ Previously, it was known how to simulate a quantum circuit in polynomial space and $\exp(m)$ time (as in the proof of $\mathsf{BQP} \subseteq \mathsf{P}^{\mathsf{\#P}} $), or in exponential space and $\exp(n)$ time.

In addition, we provide a time-space trade-off scheme, which enables even faster simulation at the cost of more space usage. See Section~\ref{sec:quantum-circuits} for the quantum circuit notations that are used throughout this section.

\subsection{Polynomial-Space Simulation Algorithms for General Quantum Circuits}

We first present a simple recursive algorithm for general circuits.

\begin{theo}\label{theo:PSPACE-SIMUL}
	Given a quantum circuit $C$ on $n$ qubits with depth $d$, and two computational basis states $\spz{x},\spz{y}$,
	we can compute $\rpz{y} C \spz{x}$ in $O(n\cdot(2d)^{n+1})$ time and $O(n \log d)$ space.
\end{theo}

\begin{proof}
	In the base case $d=1$, the answer can be trivially computed in $O(n)$ time.
	
	When $d>1$, we have
	\begin{align}
	\rpz{y} C \spz{x} &= \rpz{y} C_{[d \leftarrow d/2+1]} \cdot C_{[d/2 \leftarrow 1]} \spz{x} \notag \\
	&=\rpz{y} C_{[d \leftarrow d/2+1]} \left(  \sum_{z \in \{0,1\}^n} \outpt{z} \right) C_{[d/2 \leftarrow 1]} \spz{x} \notag \\
	&=\sum_{z \in \{0,1\}^n} \rpz{y} C_{[d \leftarrow d/2+1]} \spz{z} \cdot \rpz{z} C_{[d/2 \leftarrow 1]} \spz{x}. \label{eq:sum}
	\end{align}
	
	Then, for each $z$, we calculate $\rpz{y} C_{[d \leftarrow d/2+1]} \spz{z} \cdot \rpz{z} C_{[d/2 \leftarrow 1]} \spz{x}$ by recursively calling the algorithm on the two sub-circuits $C_{[d \leftarrow d/2+1]}$ and $C_{[d/2 \leftarrow 1]}$ respectively; and sum them up to calculate \eqref{eq:sum}.
	
	It is easy to see the above algorithm is correct, and its running time can be analyzed as follows: let $F(d)$ be its running time on a circuit of $d$ layers; then we have $F(1) = O(n)$, and by the above discussion
	$$
	F(d) \le 2^{n+1} \cdot F(\ceil{d/2}) = O(n\cdot 2^{(n+1)\ceil{\log d}}) = O(n \cdot (2^{\ceil{\log d}})^{n+1}) \le O(n \cdot (2d)^{n+1}),
	$$
	
	\noindent which proves our running time bound.
	
	Finally, we can see in each recursion level, we need $O(n)$ space to save the indices of $\spz{x}$ and $\spz{y}$, and $O(1)$ space to store an intermediate answer. \ Since there are at most $O(\log d)$ recursion levels, the total space is bounded by $O(n \log d)$.
\end{proof}

\subsection{Faster Polynomial Space Simulation Algorithms for Grid Quantum Circuits}

When a quantum circuit is spatially local, i.e., its base graph can be embedded on a grid, we can further speed up the simulation with a more sophisticated algorithm.

We first introduce a simple lemma which shows that we can find a small balanced cut in a two-dimensional grid.	

\begin{lemma}\label{lm:triv-lemma}
	Given a grid $G=(V,E)$ of size $H \times W$ such that $|V| \ge 2$, we can find a subset $S\subset E$ such that
	\begin{itemize}
		\item $|S| \le O(\sqrt{|V|})$, and
		\item after $S$ is removed, $G$ becomes a union of two disconnected grids with size smaller than $\frac{2}{3} |V|$.
	\end{itemize}
\end{lemma}

\begin{proof}
	We can assume $H \ge W$ without loss of generality and simply set $S$ to be the set of all the edges between the $\floor{H/2}$-th row and
	the $\floor{H/2}+1$-th row; then both claims are easy to verify.
\end{proof}

We now present a faster algorithm for simulating quantum circuits on grids.

\begin{theo}\label{theo:fast-small-depth}
	Given a quantum circuit $C$ on $n$ qubits with depth $d$, and two computational basis states $\spz{x},\spz{y}$, assuming that $G_C$ can be embedded into a two-dimensional grid with size $n$ (with the embedding explicitly specified), we can compute $\rpz{y} C \spz{x}$ in $2^{O(d \sqrt{n})}$ time and $O(d \cdot n \log n)$ space.
\end{theo}

\begin{proof}
	For ease of presentation, we slightly generalize the definition of quantum circuits: now each gate can be of the form $O_i \otimes I_{[n] \setminus \{a_i,b_i\}}$ (a 2-qubit gate) or $O_i \otimes I_{[n] \setminus \{a_i\}} $ (a 1-qubit gate) or simply $I_{[n]}$ (a 0-qubit gate, which is introduced just for convenience).
	
	The algorithm works by trying to break the current large instance into many small instances which we then solve recursively. \ But unlike the algorithm in Theorem~\ref{theo:PSPACE-SIMUL}, which reduces an instance to many sub-instances with fewer {\em gates}, our algorithm here reduces an instance to many sub-instances with fewer {\em qubits}.
	
	{\bf The base case, $n=1$ qubit.} In this case, all the gates are either 1-qubit or 0-qubit; hence the answer can be calculated straightforwardly in $O(m)$ time and constant space.
	
	{\bf Cutting the grid by a small set.} When $n \ge 2$, by Lemma~\ref{lm:triv-lemma}, we can find a subset $S$ of edges with $|S| \le O(\sqrt{n})$. \ After $S$ is removed, the grid becomes a union of two disconnected grids $A$ and $B$ (we use $A,B$ to denote both the grids and the sets of the vertices in the grid for simplicity) with size smaller than $\frac{2}{3} n$.
	
	Let
	$$
	\{ R = i \betw U_i \text{ is of the form } O_{i} \otimes I_{[n] \setminus \{a_i,b_i\}} \text{ and } (a_i,b_i) \in S \},
	$$ that is, the set of the indices of the gates crossing the cut $S$. \ Without loss of generality, we can assume that for each $i \in R$, we have $a_i \in A$ and $b_i \in B$.
	
	Since in a single layer, there is at most one gate acting on a particular adjacent pair of qubits, we have
	$$
	|R| \le O(d \sqrt{n}).
	$$
	
	{\bf Breaking the gates in $R$.} Now, for each $i \in R$, we decompose $O_i$ (which can be viewed as a matrix in $\mathbb{C}^{4 \times 4}$) into a sum of $16$ single-entry matrices $O_{i,1},O_{i,2},\dotsc,O_{i,16}$.
	
	Write $O_i$  as $$ O_i = \sum_{x,y \in \{0,1\}^2} \rpz{y} O_i \spz{x} \cdot \spz{y} \rpz{x}.$$
	Then we set $O_{i,j} = \rpz{y_j} O_i \spz{x_j} \cdot \spz{y_j} \rpz{x_j}$ for each $j \in [16]$, where $(x_j,y_j)$ is the $j$-th ordered pair in $\{0,1\}^2 \times \{0,1\}^2$.
	
	{\bf Decomposing the instance.} Now, we are going to expand each $U_i = O_i \otimes I_{[n] \setminus \{a_i,b_i\}}$ as a sum
	$$
	U_i = \sum_{j=1}^{16} O_{i,j} \otimes I_{[n] \setminus \{a_i,b_i\}}
	$$
	for each $i \in R$, and therefore decompose the answer $ \rpz{y} C \spz{x} = \rpz{y} U_{m}U_{m-1}\cdots U_{1} \spz{x}$ into a sum of $16^{|R|}$ terms. \ More concretely, for a mapping $\tau$ from $R$ to $[16]$ and an index $i \in [m]$, we define
	$$
	U_{i,\tau} = \begin{cases}
	O_{i,\tau(i)} \times I_{[n] \setminus \{a_i,b_i\}} &\quad\quad i \in R.\\
	U_i                                                &\quad\quad i \not\in R.\\
	\end{cases}
	$$
	Let $\mathcal{T}$ be the set of all mappings from $R$ to $[16]$. \ Then we have
	$$
	\rpz{y} C \spz{x} = \rpz{y} U_{m}U_{m-1}\cdots U_{1} \spz{x} = \sum_{\tau \in \mathcal{T}} \rpz{y} U_{m,\tau} U_{m-1,\tau} \cdots U_{1,\tau} \spz{x}.
	$$
	
	{\bf Dealing with the sub-instance.} For each $\tau \in \mathcal{T}$ and an index $i \in [m]$, we are going to show that $U_{i,\tau}$ can be decomposed as $U_{i,\tau}^A \otimes U_{i,\tau}^B$, where $U_{i,\tau}^A$ and $U_{i,\tau}^B$ are operators on $\heb_{A}$ and $\heb_{B}$ respectively.
	
	When $i \in R$, by definition, there exist $x,y \in \{0,1\}^2$ and $\alpha \in \mathbb{C}$ such that
	$$
	U_{i,\tau} = \alpha \cdot \spz{y}\rpz{x} \otimes I_{[n] \setminus \{a_i,b_i\}} = \alpha \cdot \left(\spz{y_{0}}\rpz{x_{0}} \otimes I_{A \setminus \{a_i\}} \right) \otimes \left(\spz{y_{1}}\rpz{x_{1}} \otimes I_{B \setminus \{b_i\}} \right).
	$$
	
	Otherwise $i \notin R$. \ In this case, if $O_i$ is of the form $O_i \otimes I_{[n] \setminus \{a_i,b_i\}}$, then $a_i,b_i$ must be both in $A$ or in $B$ and the claim trivially holds; and the claim is also obvious when $O_i$ is of the form $O_i \otimes I_{[n] \setminus \{a_i\}}$ or $I_{[n]}$.
	
	Moreover, one can easily verify that each $U_{i,\tau}^A$ is of the form $O_i^A \otimes I_{A \setminus \{a_i,b_i\}}$ or $O_i^A \otimes I_{A \setminus \{a_i\}} $ or simply $I_{A}$, in which $O_i^A$ is (respectively) a 2-qubit operator on $\heb_{\{a_i,b_i\}}$ or a 1-qubit operator on $\heb_{\{a_i\}}$), and the same holds for each $U_{i,\tau}^B$.
	
	Hence, we have
	\begin{align}
	&\rpz{y} U_{m}U_{m-1}\cdots U_{1} \spz{x} \notag \\
	=&\sum_{\tau \in \mathcal{T}} \rpz{y} U_{m,\tau} U_{m-1,\tau} \cdots U_{1,\tau} \spz{x}. \notag \\
	=&\sum_{\tau \in \mathcal{T}} \rpz{y} (U_{m,\tau}^A \otimes U_{m,\tau}^B) (U_{m-1,\tau}^A \otimes U_{m-1,\tau}^B) \cdots(U_{1,\tau}^A \otimes U_{1,\tau}^B) \spz{x}. \notag\\
	=&\sum_{\tau \in \mathcal{T}} \rpz{y_A} U_{m,\tau}^A U_{m-1,\tau}^A \cdots U_{1,\tau}^A \spz{x_A} \cdot \rpz{y_B} U_{m,\tau}^B U_{m-1,\tau}^B \cdots U_{1,\tau}^B \spz{x_B}, \label{eq:final-result}
	\end{align}
	where $x_A,x_B$ ($y_A,y_B$) is the projection of $x$ ($y$) on $\heb_A$ and $\heb_B$.
	
	So from the above discussion, we can then calculate $ \rpz{y_A} U_{m,\tau}^A U_{m-1,\tau}^A \cdots U_{1,\tau}^A \spz{x_A}$ with a recursive call with computational basis states $\spz{x_A}$ and $\spz{y_A}$, grid $A$, and $m$ gates $U_{1,\tau}^A,U_{2,\tau}^A,\dotsc,U_{m,\tau}^A$.
	
	The matrix element $\rpz{y_B} U_{m,\tau}^B U_{m-1,\tau}^B \cdots U_{1,\tau}^B \spz{x_B}$ can be computed similarly. \ After that we sum up all the terms in \eqref{eq:final-result} to get the answer.
	
	{\bf Complexity analysis.} Now we are going to bound the running time. \ Let $F(n)$ be an upper bound on the running time when the size of the remaining grid is $n$. \ Then we have
	
	$$
	F(n) =   \begin{cases}
	O(m)       & \quad \text{when }n = 1.\\
	2^{O(d\sqrt{n})}\cdot \max_{k \in [n/3,2n/3]} F(k)  & \quad \text{otherwise}.\\
	\end{cases}	
	$$
	The second case is due to the fact that the sizes of sub-instances (i.e., the sizes of $A$ and $B$) lie in $[n/3,2n/3]$, and $\mathcal{T} = 16^{|R|} = 2^{O(d\sqrt{n})}$. It is not hard to see that $F(n)$ is an increasing function, so we have $F(n) = 2^{O(d\sqrt{n})} F(2n/3)$ for $n>1$, which further simplifies to $F(n) = 2^{O(d\sqrt{n})}$.
	
	Finally, we can see that at each recursion level, we need $O(d \cdot n)$ space to store the circuit, and $O(1)$ space to store the intermediate answer. \ Since there are at most $\log n$ recursion levels, the space complexity is $O(d \cdot n \log n)$.
	
\end{proof}

Interestingly, by using tensor network methods, Markov and Shi \cite{markovshi} gave an algorithm for simulating quantum circuits on grids with similar running time to ours. \ However, the difference is that Markov and Shi's algorithm requires $2^{O(d \sqrt{n})}$ time {\em and} $2^{O(d \sqrt{n})}$ space, whereas ours requires $2^{O(d \sqrt{n})}$ time and only polynomial space.

The algorithm of Theorem~\ref{theo:fast-small-depth} achieves a speedup over Theorem~\ref{theo:PSPACE-SIMUL} only for small $d$, but we can combine it with the algorithm in Theorem~\ref{theo:PSPACE-SIMUL} to get a faster algorithm for the whole range of $d$.

\begin{theo}
	There is a constant $c$ such that, given a quantum circuit $C$ on $n$ qubits with depth $d$, and two computational basis states $\spz{x},\spz{y}$, assuming that $G_C$ can be embedded into a two dimensional grid with size $n$ (with the embedding explicitly specified), we can compute $\rpz{y} C \spz{x}$ in
	$$
	O(2^n\cdot \left[1 + \left(\frac{d}{c \sqrt{n}}\right)^{n+1} \right] )
	$$
	time and $O(d \cdot n \log n)$ space.
\end{theo}

\begin{proof}
	By Theorem~\ref{theo:fast-small-depth}, there is a constant $c$ such that we have an $O(2^{n})$ time and polynomial space algorithm for calculating $\rpz{y} C \spz{x}$ when the depth is at most $c\sqrt{n}$ for circuit on grids. \ So we can use the same algorithm as in Theorem~\ref{theo:PSPACE-SIMUL}, except that we revert to the algorithm in Theorem~\ref{theo:fast-small-depth} when the depth is no more than $c \sqrt{n}$.
	
	We still let $F(d)$ be the running time on a circuit of $d$ layers. \ We then have $F(d) = O(2^n)$ when $d \le c\sqrt{n}$. \ From the above discussion, we can see that for $d>c\sqrt{n}$,
	$$
	F(d) \le 2^{n+1} \cdot F(\ceil{\frac{d}{2}}) = O(2^n\cdot 2^{(n+1)\ceil{\log(d/c\sqrt{n})}}) = O(2^n\cdot\left(\frac{d}{c\sqrt{n}}\right)^{n+1}),
	$$
	which proves the running time bound. \ And it is not hard to see that the algorithm's space usage is dominated by $O(d \cdot n \log n)$.
\end{proof}

\subsection{Space-Time Trade-off Schemes}

We now show how to optimize the running time for whatever space is available.

\begin{theo}\label{theo:ST-trade-off}
	Given a quantum circuit $C$ on $n$ qubits with depth $d$, two computational basis states $\spz{x},\spz{y}$ and an integer $k$, we can compute $\rpz{y} C \spz{x}$ in
	$$
	O(n2^{n-k}\cdot 2^{(k+1)\ceil{\log d}}) \le O(n2^{n-k}\cdot(2d)^{k+1})
	$$
	time and $O(2^{n-k} \log d)$ space.
\end{theo}

\begin{proof}
	
	{\bf Decomposing the whole Hilbert space $\heb_{[n]}$.} We first decompose $\heb_{[n]}$ into a direct sum of many subspaces. \ Let $w_i$ be the $i$-th string in $\{0,1\}^k$ in lexicographic order. \ For each $i \in [2^k]$, let $\heb_i = \mathsf{Span}(\spz{w_i 0^{n-k}},...\spz{w_i 1^{n-k}})$. \ Then we have
	$$
	\heb_{[n]} = \bigoplus_{i=1}^{2^k} \heb_i.
	$$.
	
	Also, let $\pj_i$ be the projection from $\heb_{[n]}$ to $\heb_i$; then
	$$
	I_{[n]} = \sum_{i=1}^{2^k} \pj_i.
	$$
	
	Now we generalize the original problem as follows: given two indices $s,t \in [2^k]$ and a pure state $\spz{u}$ in $\heb_s$, we want to compute $\pj_{t} C\spz{u}$. \ By choosing $s$ and $t$ such that $\heb_s$ contains $\spz{x}$ and $\heb_t$ contains $\spz{y}$, we can easily solve the original problem.
	
	{\bf The base case $d = 1$.} When there is only one layer, $\pj_{t} C \spz{u}$ can be calculated straightforwardly in $O(n\cdot 2^{n-k})$ time and $O(2^{n-k})$ space.
	
	{\bf Recursion.} When $d>1$, we have
	\begin{align*}
	\pj_{t} C \spz{u} &= \pj_{t} C_{[d \leftarrow d/2+1]} \cdot C_{[d/2 \leftarrow 1]} \spz{u}\\
	&=\pj_{t} C_{[d \leftarrow d/2+1]} \left(  \sum_{z \in [2^k]} \pj_{z} \right) C_{[d/2 \leftarrow 1]} \spz{u}\\
	&=\sum_{z \in [2^k]} \pj_{t} C_{[d \leftarrow d/2+1]}  \pj_{z} C_{[d/2 \leftarrow 1]} \spz{u}.
	\end{align*}
	We can then calculate $\pj_{t} C_{[d \leftarrow d/2+1]}  \pj_{z} C_{[d/2 \leftarrow 1]} \spz{u}$ for each $z$ as follows: we first use a recursive call to get $\spz{b} = \pj_{z} C_{[d/2 \leftarrow 1]} \spz{u}$ and a second recursive call to compute $\pj_{t} C_{[d \leftarrow d/2+1]} \spz{b}$ (note that $\spz{b} \in \heb_{z}$).

	{\bf Complexity analysis.} It is easy to see that the total space usage is $O(2^{n-k} \log d)$, since for each $i$, storing a vector in $\heb_{i}$ takes $O(2^{n-k})$ space, and we only need to record $O(1)$ such vectors at each recursion level. \ In addition, when $d=1$, we need only $O(2^{n-k})$ space.
	
	For the running time bound, let $F(d)$ denote the running time on a circuit of $d$ layers; then $F(1) = O(n2^{n-k})$. \ From the above discussion, it follows that
	$$
	F(d) \le 2^{k+1} \cdot F(\ceil{d/2}) = O(n2^{n-k}\cdot 2^{(k+1)\ceil{\log(d)}}) = O(n2^{n-k}\cdot(2d)^{k+1}).
	$$
\end{proof}

The above trade-off scheme can be further improved for quantum circuits on grids.

\begin{theo}
	There is a constant $c$ such that, given a quantum circuit $C$ on $n$ qubits with depth $d$, two computational basis states $\spz{x},\spz{y}$ and an integer $k$, assuming that $G_C$ can be embedded into a two dimensional grid with size $n$, we can compute $\rpz{y} C \spz{x}$ in
	$$
	2^{O(n)}\cdot \left[1 +(2d/c\sqrt{n})^{k+1}\right]
	$$
	time and
	$$
	O\left(2^{n-k} \max(1,\log(d/\sqrt{n}))\right)
	$$ space.
\end{theo}

\begin{proof}
	By Theorem~\ref{theo:fast-small-depth}, there is a constant $c$ such that we have an $O(2^{n})$ time algorithm for calculating $\rpz{y} C \spz{x}$ for circuits on grids with depth at most $c \sqrt{n}$.
	
	Then we use the same algorithm as in Theorem~\ref{theo:ST-trade-off}, with the only modification that when $d \le c \sqrt{n}$, we calculate $\pj_{t} \cdot C \spz{u}$ by $2^{2(n-k)}$ calls of the algorithm in Theorem~\ref{theo:fast-small-depth}.
	
	With the same analysis as in Theorem~\ref{theo:ST-trade-off}, when $d > c\sqrt{n}$, we can see that the total space usage is $O(2^{n-k} \log(d/c\sqrt{n}))$ , and the running time is
	$$
	O(2^{n+2(n-k)+(k+1)\ceil{\log(d/c\sqrt{n})}}) = O(2^{O(n)}\cdot(2d/c\sqrt{n}))^{k+1}).
	$$
	
	Combining with the algorithm for $d \le c \sqrt{n}$ proves our running time and space bound.
\end{proof}

\section{Strong Quantum Supremacy Theorems Must Be Non-Relativizing}
\label{sec:non-relativizing}

In this section we show that there is an oracle relative to which $\SampBPP = \SampBQP$, yet $\PH^\oracle$ is infinite.

Recall that an oracle $\oracle$ is a function $ \oracle : \{0,1\}^* \to \{0,1\}$, and the combination of two oracles $\oracle_0,\oracle_1$, denoted as $\oracle_0 \oplus \oracle_1$, simply maps $z \in \{0,1\}^*$ to $\oracle_{z_1}(z_2,z_3,\dotsc,z_{|z|})$ (cf. \cite{fenner2003oracle}). We use $\oracle_n$ to denote the restriction of $\oracle$ on $\{0,1\}^n$.

\subsection{Intuition}

We have two simultaneous objectives: (1) we need $\SampBPP$ and $\SampBQP$ to be equal; and (2) we also need $\PH$ to be infinite. \ So it will be helpful to review some previous results on (1) and (2) separately.

\begin{itemize}
	\item An oracle $\oracle$ such that $\SampBPP^{\oracle} = \SampBQP^{\oracle}$: in order to make two classes equal, we can use the standard method: {\em adding a much more powerful oracle}~\cite{baker1975relativizations}. \ That is, we set $\oracle$ to be a $\PSPACE$-complete language, like $\TQBF$. \ Then it is easy to see both $\SampBPP^{\TQBF}$ and $\SampBQP^{\TQBF}$ become $\mathsf{SampPSPACE}$ (i.e., the class of approximate sampling problems solvable in polynomial space).
	
	\item An oracle $\oracle$ such that $\PH^{\oracle}$ is infinite: a line of works by Yao~\cite{yao1985separating}, H{\aa}stad~\cite{hastad1986almost}, and others constructed relativized worlds where $\PH$ is infinite, and a very recent breakthrough by Rossman, Servedio, and Tan~\cite{rst} even shows that $\PH$ is infinite relative to a random oracle with probability $1$.
	
\end{itemize}
\subsubsection*{A Failed Attempt: Direct Combination}

The first natural idea is to combine the previous two results straightforwardly by setting the oracle to be $\TQBF \oplus \oracle$, where $\oracle$ is a random oracle.

Alas, it is not hard to see that this does not work: while $\PH$ is still infinite, a $\SampBQP$ algorithm can perform \fsampling\ (cf.\ Definition~\ref{defi:Ffish-Fsamp}) on the random oracle bits, and it is known that no $\SampBPP$ algorithm can do that~\cite{aa:forrelation} (see also Theorem~\ref{theo:fsamp-lowerbound}). \ Hence, in this case $\SampBQP \ne \SampBPP$.

\subsubsection*{Another Failed Attempt: Hiding a ``Secret Random String'' in a Secret Location}

The failure of the naive approach suggests that we must somehow ``hide'' the random oracle bits, since if the $\SampBQP$ algorithm has access to them, then $\SampBPP$ and $\SampBQP$ will not be equal. \ More specifically, we want to hide a ``secret random string'' among the oracle bits so that:

\begin{enumerate}
	\item[(1)] a $\PH$ algorithm can find it, so that $\PH$ is still infinite, but
	\item[(2)] a $\SampBQP$ algorithm cannot find it, so that we can still make $\SampBPP = \SampBQP$ by attaching a $\TQBF$ oracle.
\end{enumerate}

Inspired by the so-called cheat-sheet construction~\cite{aaronson2015separations}, it is natural to consider a direct hiding scheme. \ Imagine that the oracle bits are partitioned into two parts: one part is $\log N$ copies of the $\mathsf{OR}$ function on $N$ bits, and another part is $N$ binary strings $y_1,\dotsc,y_N$, each with length $N$. \ Let $t=a_1,a_2,\dotsc,a_{\log N} \in \{0,1\}^{\log N}$ be the answer to the copies of $\mathsf{OR}$; we can also interpret $t$ as an integer in $[N]$. \ Finally, set $y_{t}$ to be a random input, while other $y_i$'s are set to zero.

Intuitively, a $\PH$ algorithm can easily evaluate the $\log N$ copies of $\mathsf{OR}$ and then get access to the random string; while it is known that $\OR$ is hard for a quantum algorithm, so no quantum algorithm should be able to find the location of the random string efficiently.

Unfortunately, there is a fatal issue with the above approach: a $\SampBQP$ algorithm is also given an input $x \in \{0,1\}^n$ and it may {\em guess} that the input $x$ denotes the location of the random string. \ That is, on some particular input, the $\SampBQP$ algorithm is ``lucky'' and gets access to the random string, which still makes $\SampBPP$ and $\SampBQP$ unequal.

\subsubsection*{Hiding the ``Secret Random String'' in a Bunch of $\OR$'s}

Therefore, our final construction goes further. \ Instead of hiding the random string in a secret location amid the oracle bits, we hide it using a bunch of $\mathsf{OR}$s. \ That is, suppose we want to provide $N$ uniform random bits. \ Then we provide them each as an $\mathsf{OR}$ of $N$ bits. \ In this way, a $\PH$ algorithm is still able to access the random bits, while a quantum algorithm, even if it's ``lucky'' with its additional input, still can't get access to these hidden random bits.

\subsection{Implementation}

\subsubsection*{The Distribution $\distr_\oracle$ on Oracles.}

We first describe formally how to hide a random string inside a bunch of $\OR$'s by defining a distribution $\distrO$ on oracles.

For notational convenience, our constructed oracles always map all odd-length binary strings to $0$. \ So we can alternatively describe such an oracle $\oracle$ by a collection of functions $\{ f_n \}_{n=0}^{+\infty}$, where each $f_n$ is a function from $\{0,1\}^{2n} \to \{0,1\}$. \ That is, $\oracle_{2n}$ is set to be $f_n$ for each $n$, while the $\oracle_{2n+1}$'s are all constant zero functions.

For each string $p \in \{0,1\}^n$, we use $B_{n,p}$ to denote the set of strings in $\{0,1\}^{2n}$ with $p$ as a prefix. \ Now we first define a distribution $\distr_{n}$ on functions $\{0,1\}^{2n} \to \{0,1\}$, from which a sample function $f_n$ is generated as follows: initially, we set $f_n(x) = 0$ for all $x \in \{0,1\}^{2n}$; then for each $p \in \{0,1\}^{n}$, with probability $0.5$, we pick an element $e$ in $B_{n,p}$ at uniformly random and set $f_n(e) = 1$. \ Observe that by taking the $\OR$ of each $B_{n,p}$, we get a function $g(p) := \lor_{x \in B_{n,p}} f_n(x)$, which is a uniform random function from $\{0,1\}^n$ to $\{0,1\}$ by construction.

Finally, the $\distr_{n}$'s induce a distribution $\distrO$ on oracles, which generates an oracle $\oracle$ by drawing $f_n \sim \distr_{n}$ independently for each integer $n$. \ That is, we set $\oracle_{2n}$ to be $f_n$, and $\oracle_{2n+1}$ to be $\mathbf{0}$, for each $n$.

Having defined the distribution $\distrO$, we are ready to state our result formally.

\begin{theo}\label{theo:oralce-result-distrO}
	For an oracle $\oracle$ drawn from the distribution $\distrO$, the following two statements hold with probability $1$:
	
	\begin{itemize}
		\item $\SampBPP^{\TQBF,\oracle} = \SampBQP^{\TQBF,\oracle}$.
		\item $\PH^{\TQBF,\oracle}$ is infinite.
	\end{itemize}
	
\end{theo}

From which our desired result follows immediately.

\begin{cor}
	There exists an oracle $\oracle' = \TQBF \oplus \oracle$ such that $\SampBPP^{\oracle'} = \SampBQP^{\oracle'}$ and $\PH^{\oracle'}$ is infinite.
\end{cor}

The rest of this section is devoted to the proof of Theorem~\ref{theo:oralce-result-distrO}.

\subsection{$\SampBPP^{\TQBF,\oracle} = \SampBQP^{\TQBF,\oracle}$ with Probability $1$.}

We first describe an algorithm for simulating $\SampBQP^{\TQBF,\oracle}$ in $\SampBPP^{\TQBF,\oracle}$, thereby proving the first part of Theorem~\ref{theo:oralce-result-distrO}. \ In the following, we assume that all oracle algorithms are given access to two oracles, $\TQBF$ and $\oracle$.

Given a $\SampBQP$ oracle algorithm $M$, our central task is to give a $\SampBPP$ oracle algorithm that simulates $M$ closely. \ Formally:

\begin{lemma}\label{lemma:simulation}
	For any $\SampBQP$ oracle algorithm $M$, there is a $\SampBPP$ oracle algorithm $A$ such that:
	
	Let $\oracle$ be an oracle drawn from $\distrO$, and let $\distr^{M}_{x,\varepsilon}$ and $\distr^{A}_{x,\varepsilon}$ be the distributions output by  $M^{\TQBF,\oracle}$ and $A^{\TQBF,\oracle}$ respectively on input $\langle x, 0^{1/\varepsilon} \rangle$. \ Then with probability at least $1 - \exp\{ -(2 \cdot |x|+1/\varepsilon) \}$, we have
	$$
	\| \distr^{M}_{x,\varepsilon} - \distr^{A}_{x,\varepsilon} \| \le \varepsilon.
	$$
\end{lemma}

Before proving Lemma~\ref{lemma:simulation}, we show it implies the first part of Theorem~\ref{theo:oralce-result-distrO}.

\begin{proofof}{the first part of Theorem~\ref{theo:oralce-result-distrO}}
	Fix a $\SampBQP$ oracle algorithm $M$, and let $\oracle$ be an oracle drawn from $\distrO$. \ We first show that with probability $1$, there is a classical algorithm $A_M$ such that
	\begin{equation}\label{eq:AM-exist}
	\text{$\|\distr^{M}_{x,\varepsilon} - \distr^{A_M}_{x,\varepsilon} \| \le \varepsilon$ for all $x \in \{0,1\}^*$ and $\varepsilon = 2^{-k}$ for some integer $k$.}
	\end{equation}
	
	Let $A$ be the $\SampBPP$ algorithm guaranteed by Lemma~\ref{lemma:simulation}. \ For an input $x \in \{0,1\}^{*}$ and an integer $k$, we call $(x,k)$ a \textit{bad pair} if $\|\distr^{M}_{x,2^{-k}} - \distr^{A}_{x,2^{-k}}\| > 2^{-k}$. By Lemma~\ref{lemma:simulation}, the expected number of bad pairs is upper-bounded by
	$$
	\sum_{n=1}^{+\infty} 2^{n} \cdot \sum_{k=1}^{+\infty} \exp(-(2n+2^k)) \le
	\sum_{n=1}^{+\infty} \sum_{k=1}^{+\infty} \exp(-(n+k)) \le O(1).
	$$
	
	This means that with probability 1, there are only finitely many bad pairs, so we can handle them by hardwiring their results into the algorithm $A$ to get the algorithm $A_M$ we want.
	
	Since there are only countably many $\SampBQP$ oracle algorithms $M$, we see with probability $1$, for every $\SampBQP$ oracle algorithm $M$, there is a classical algorithm $A_M$ such that \eqref{eq:AM-exist} holds. \ We claim that in that case, $\SampBQP^{\TQBF,\oracle} = \SampBPP^{\TQBF,\oracle}$.
	
	Let $\mathcal{S}$ be a sampling problem in $\SampBQP^{\TQBF,\oracle}$. \ This means that there is a $\SampBQP$ oracle algorithm $M$, such that for all $x\in \{0,1\}^{*}$ and $\varepsilon$, we have $\|\distr^{M}_{x,\varepsilon} - \mathcal{S}_x \| \le \varepsilon$. \ Let $A_M$ be the corresponding $\SampBPP$ algorithm. \ Now consider the following algorithm $A'$: given input $\langle x,0^{1/\varepsilon} \rangle$, let $k$ be the smallest integer such that $2^{-k} \le \varepsilon/2$; then run $A_M$ on input $\langle x, 0^{2^k} \rangle$ to get a sample from $\distr^{A_M}_{x,2^{-k}}$.
	
	Since
	\begin{align*}
	\| \distr^{A'}_{x,\varepsilon} - S_x \|  &=  \| \distr^{A_M}_{x,2^{-k}} - S_x \| \\
	&\le \|\distr^{M}_{x,2^{-k}} - \distr^{A_M}_{x,2^{-k}} \| + \|\distr^{M}_{x,2^{-k}} - \mathcal{S}_x \| \le 2 \cdot 2^{-k} \le \varepsilon,
	\end{align*}
	this means that $A'$ solves $\mathcal{S}$ and $\mathcal{S} \in \SampBPP^{\TQBF,\oracle}$. \ So $\SampBQP^{\TQBF,\oracle} \subseteq \SampBPP^{\TQBF,\oracle}$ with probability $1$, which completes the proof.
\end{proofof}

We now prove Lemma~\ref{lemma:simulation}, which is the most technical part of the whole section.

\begin{proofof}{Lemma~\ref{lemma:simulation}}
	Recall that from the canonical description in Section~\ref{sec:canonical-SampBQP}, there exists a fixed polynomial $p$, such that given input $\langle x,0^{1/\varepsilon} \rangle$, the machine $M$ first constructs a quantum circuit $C$ with $N=p(|x|,1/\varepsilon)$ qubits and $N$ gates classically ($C$ can contain $\TQBF$ and $\oracle$ gates). We first set up some notation.
	
	{\bf Notation.} Recall that $\oracle$ can be specified by a collection of functions $\{ f_n \}_{n=0}^{+\infty}$, where each $f_n$ maps $\{0,1\}^{2n}$ to $\{0,1\}$. \ Without loss of generality, we can assume that all the $\oracle$ gates act on an even number of qubits, and for each $n$, all the $f_n$ gates act on the first $2n$ qubits.

	For a function $f : \{0,1\}^{k} \to \{0,1\}$, we use $U_{f}$ to denote the unitary operator mapping $\spz{i}$ to $(-1)^{f(i)} \spz{i}$ for $i \in \{0,1\}^{k}$.
	
	Suppose there are $T$ $\oracle$-gates in total, and suppose the $i$-th $\oracle$-gate is an $f_{n_i}$ gate. \ Then the unitary operator $U$ applied by the circuit $C$ can be decomposed as
	$$
	U = U_{T+1} (U_{f_{n_T}} \otimes I_{N-2n_T}) \cdots (U_{f_{n_2}} \otimes I_{N-2n_2}) U_2 (U_{f_{n_1}} \otimes I_{N-2n_1}) U_1,
	$$
	\noindent where the $U_i$'s are the unitary operators corresponding to the sub-circuits which don't contain an $\oracle$ gate.
	
	Our algorithm proceeds by replacing each $\oracle$-gate by a much simpler gate, one by one, without affecting the final quantum state too much. \ It then simulates the final circuit with the help of the $\TQBF$ oracle.
	
	{\bf Replacing the $t$-th $\oracle$-gate.} Suppose we have already replaced the first $t-1$ $\oracle$-gates. \ That is, for each $i \in [t-1]$, we replaced the $f_{n_i}$ gate (the $i$-th $\oracle$-gate) with a $g_i$ gate, and now we are going to replace the $t$-th $\oracle$-gate.
	
	Let
	$$
	\spz{v} = U_{t} (U_{g_{t-1}} \otimes I_{N-2n_{t-1}}) \cdots (U_{g_2} \otimes I_{N-2n_2}) U_2 (U_{g_1} \otimes I_{N-2n_1}) U_1 \spz{0}^{\otimes N},
	$$
	\noindent which is the quantum state right before the $t$-th $\oracle$ gate in the circuit after the replacement.
	
	For brevity, we use $f$ to denote the function $f_{n_t}$, and we drop the subscript $t$ of $n_t$ when it is clear from context.
	
	{\bf Analysis of incurred error.} The $t$-th $\oracle$-gate is an $f$ gate. \ If we replace it by a $g$ gate, the change to the quantum state is
	$$
	\| U_f \otimes I_{N-2n} \spz{v} - U_g \otimes I_{N-2n} \spz{v}\| =
	\| (U_f - U_g) \otimes I_{N-2n} \spz{v} \|.
	$$
	
	We can analyze the above deviation by bounding its square. \ Let $H$ be the Hilbert space spanned by the last $N-2n$ qubits, and let $\rho = \Tr_{H}[\outpt{v}]$. \ Then we have
	\begin{align*}
	&\|((U_f-U_g) \otimes I_{N-2n}) \spz{v}\|^2\\
	=&\Tr\left[ \ct{(U_f-U_g)}(U_f-U_g) \otimes I_{N-2n} \spz{v} \rpz{v} \right]\\
	=&\Tr\left[\ct{(U_f-U_g)}(U_f-U_g) \rho\right].
	\end{align*}
	Note that $$\ct{(U_f-U_g)}(U_f-U_g) = 4 \sum_{f(i) \ne g(i)} \outpt{i}$$ from the definition. \ So we can further simplify the above trace as
	
	\begin{equation}\label{eq:error-tmp}
	\Tr\left[\ct{(U_f-U_g)}(U_f-U_g) \rho\right] = 4 \sum_{f(i) \ne g(i)} \Tr\left[\outpt{i} \rho\right] =
	4 \sum_{f(i) \ne g(i)} \rpz{i}\rho \spz{i}.
	\end{equation}
	
	Now, $\rho$ is a (mixed) quantum state on the first $2 n$ bits, and $\rpz{i}\rho \spz{i}$ is the probability of seeing $i$ when measuring $\rho$ in the computational basis. \ So we can define a probability distribution $Q$ on $\{0,1\}^{2n}$ by $Q(i) := \rpz{i}\rho \spz{i}$.
	
	Using the distribution $Q$, the error term \eqref{eq:error-tmp} can finally be simplified as:
	
	\begin{equation}\label{eq:error-term}
	4 \sum_{i \in \{0,1\}^{2n}} Q(i) \cdot [f(i) \ne g(i)] =4 \cdot \Pr_{i \sim Q} [f(i) \ne g(i)],
	\end{equation}
	\noindent where $[f(i) \ne g(i)]$ is the indicator function that takes value $1$ when $f(i) \ne g(i)$ and $0$ otherwise.
	
	\newcommand{\distrpost}{\distr_{n}^{\mathsf{post}}}
	
	{\bf A posterior distribution $\distrpost$ on functions from $\{0,1\}^{2n} \to \{0,1\}$.} Now, recall that $f = f_n$ is a function drawn from the distribution $\distr_{n}$. \ Our goal is to replace $f$ by another simple function $g$, such that with high probability, the introduced deviation \eqref{eq:error-term} is small.
	
	Note that when replacing the $t$-th $\oracle$ gate, we may already have previously queried some contents of $f$ (i.e., it is not the first $f_n$ gate in the circuit). So we need to consider the posterior distribution $\distrpost$ on functions from $\{0,1\}^{2n} \to \{0,1\}$. That is, we want a function $g$, such that with high probability over $f \sim \distrpost$, the error term \eqref{eq:error-term} is small.
	
	\newcommand{\fknown}{f_{\mathsf{known}}}
	
	We use a function $\fknown : \{0,1\}^{2n} \to \{0,1,*\}$ to encode our knowledge: if $f(i)$ is not queried, then we set $\fknown(i):=*$; otherwise we set $\fknown(i):=f(i)$. \ Then $\distrpost$ is simply the distribution obtained from $\distr_{n}$ by conditioning on the event that $f$ is consistent with $\fknown$.
	
	We can now work out the posterior distribution $\distrpost$ from the definition of $\distr_{n}$ and Bayes' rule.
	
	For $f \sim \distrpost$, we can see that all the sets $B_{n,p}$ (recall that $B_{n,p}$ is the set of all strings in $\{0,1\}^{2n}$ with $p$ as a prefix) are still independent. \ So we can consider each set separately.
	
	For each $p \in \{0,1\}^{n}$, if there is an $x \in B_{n,p}$  such that $\fknown(x) = 1$, then by the construction of $\distr_{n}$, all other elements $y \in B_{n,p}$ must satisfy $f(y)=0$.
	
	Otherwise, if there is no $x \in B_{n,p}$ such that $\fknown(x) = 1$, then we set $Z_p = |\{ \fknown(x) = 0 \betw x \in B_{n,p} \}|$ and note that $|B_{n,p}| = 2^n$. \ By Bayes' rule, we see that with probability $\frac{1}{2-Z_p \cdot 2^{-n}}$, all $y \in B_{n,p}$ satisfy $f(y)=0$; and for each $y \in B_{n,p}$ such that $\fknown(y)=*$, with probability $\frac{2^{-n}}{2-Z_p \cdot 2^{-n}}$, we have that $y$ is the only element of $B_{n,p}$ that satisfies $f(y)=1$.
	
	{\bf Construction and Analysis of $g$.}
	Our construction of $g$ goes as follows: we first set $g(x) = \fknown(x)$ for all $x$ such that $\fknown(x) \ne *$. \ Then for a parameter $\tau$ which will be specified later, we query all $x \in \{0,1\}^{2n}$ with $Q(x) \ge \tau$, and set $g(x)=f(x)$ for them. \ For all other positions of $g$, we simply set them to zero. \ Hence, there are at most $O(1/\tau)+W$ ones in $g$, where $W$ denotes the number of ones in $\fknown$.
	
	The following three properties of $g$ are immediate from the construction.
	\begin{flalign}\label{gprop-1}
	\qquad\bullet\qquad\text{$f(x) \ne g(x)$ implies $Q(x) \le \tau.$}&&
	\end{flalign}
	\vspace{-8mm}
	\begin{flalign}\label{gprop-2}
	\qquad\bullet\qquad\text{$g(x) = 1 $ implies $f(x) = g(x)$.}&&
	\end{flalign}
	\vspace{-8mm}
	\begin{flalign}\label{gprop-3}
	\qquad\bullet\qquad\text{For each $p \in \{0,1\}^n$, there is at most one $x \in B_{n,p}$ with $f(x) \ne g(x)$.}&&
	\end{flalign}
	
	{\bf Upper bounding the deviation~\eqref{eq:error-term}.}
	Now we are going to show that $\Pr_{x \sim Q} [f(x) \ne g(x)]$ is very small, with overwhelming probability over the posterior distribution $\distrpost$.
	
	We first define $2^n$ random variables $\{X_{p}\}_{p \in \{0,1\}^{n}}$, where $X_p = \sum_{x \in B_{n,p}} Q(x) \cdot [f(x) \ne g(x)] $ for each $p \in \{0,1\}^n$. By the construction of $\distrpost$, we can see that all $X_{p}$'s are independent. \
	Moreover, by properties~\eqref{gprop-1} and \eqref{gprop-3}, there is at most one $x \in B_{n,p}$ such that $f(x) \ne g(x)$, and that $x$ must satisfy $Q(x) \le \tau$. \ Therefore $X_{p} \in [0,\tau]$ for every $p$.
	
	Let $X = \sum_{p \in \{0,1\}^{n}} X_p$, and $\mu = \Ex[X]$. \ Alternatively, we can write $X$ as
	$$
	X = \sum_{x \in \{0,1\}^{2n}} Q(x) \cdot [f(x) \ne g(x)],
	$$
	so
	$$
	\mu = \sum_{x \in \{0,1\}^{2n}} Q(x) \cdot \Ex[f(x) \ne g(x)].
	$$
	
	We claim that $\Ex[f(x) \ne g(x)] \le 2^{-n}$ for all $x \in \{0,1\}^{2n}$, and consequently $\mu \le 2^{-n}$. Fix an $x \in \{0,1\}^{2n}$, and suppose $x \in B_{n,p}$. \ When $g(x)=1$, we must have $f(x) = g(x)$ by property~\eqref{gprop-2}. \ When $g(x)=0$, by the definition of $\distrpost$, we have $f(x)=1$ with probability at most $\frac{2^{-n}}{2-Z_p \cdot 2^{-n}} \le 2^{-n}$. \ So $\Ex[f(i) \ne g(i)] \le 2^{-n}$ in both cases and the claim is established.
	
	{\bf Applying the Chernoff Bound.} Set $\delta = \frac{\mu^{-1} \varepsilon^{4}}{32 T^2}$. \ If $\delta \le 1$, then we have
	
	$$
	32 T^2 \varepsilon^{-4} \ge \mu^{-1} \ge 2^{n}.
	$$
	
	This means that we can simply query all the positions in $f_{n}$ using $2^{2n} = O(T^{4} \cdot \varepsilon^{-8})$ queries, as this bound is polynomial in $|x|$ and $1/\varepsilon$ (recall that $T \le N = p(|x|,1/\varepsilon)$).
	
	Hence, we can assume that $\delta > 1$. \ So by Corollary~\ref{cor:mul-bound}, we have
	
	$$
	\Pr\left[X \ge 2\delta\mu \right] \le \Pr\left[X \ge (1+\delta)\mu \right] \le \exp \left\{-\frac{\delta\mu}{3\tau} \right\}.
	$$
	
	Finally, we set $\tau = \frac{\varepsilon^4}{96 T^2 \cdot (2n+\varepsilon^{-1}+\ln T)}$.
	
	Therefore, with probability $$1-\exp \left\{-\frac{\delta\mu}{3\tau} \right\} = 1-\exp(-(2n+ \varepsilon^{-1} +\ln T)) = 1-\frac{\exp(-(2n+\varepsilon^{-1})}{T},$$
	\noindent we have
	$$
	\|(U_f-U_g) \otimes I_{N-2n} \spz{v}\|^2 = 4 \cdot X  \le 8\delta\mu = \frac{\varepsilon^{4}}{4T^2},$$
	\noindent which in turn implies
	$$
	\|(U_f-U_g) \otimes I_{N-2n} \spz{v}\| \le \frac{\varepsilon^{2}}{2T}.
	$$
	
	Moreover, we can verify that $g$ only has $O(1/\tau) + W = \operatorname{poly}(n,1/\varepsilon)$ ones.
	
	\newcommand{\Cfinal}{C^{\mathsf{final}}}
	
	{\bf Analysis of the final circuit $\Cfinal$.} Suppose that at the end, for each $t \in [T]$, our algorithm has replaced the $t$-th $\oracle$-gate with a $g_t$ gate, where $g_t$ is a function from $\{0,1\}^{2n_t}$ to $\{0,1\}$. \ Let $\Cfinal$ be the circuit after the replacement.
	
	Let
	$$
	V = U_{T+1} (U_{g_T} \otimes I_{N-2n_T}) \cdots (U_{g_2} \otimes I_{N-2n_2}) U_2 (U_{g_1} \otimes I_{N-2n_1}) U_1
	$$
	\noindent be the unitary operator corresponding to $\Cfinal$. \ Also, recall that $U$ is the unitary operator corresponding to the original circuit $C$. \ We are going to show that $U \spz{0}^{\otimes N}$ and $V \spz{0}^{\otimes N}$, the final quantum states produced by $U$ and $V$ respectively, are very close.
	
	We first define a sequence of intermediate quantum states. \ Let $\spz{u_1} = U_1 \spz{0}^{\otimes N}$. \ Then for each $t > 1$, we define $$\spz{u_t} = U_{t} (U_{f_{n_{t-1}}} \otimes I_{N-2n_{t-1}}) \spz{u_{t-1}}.$$
	That is, $\spz{u_t}$ is the quantum state immediately before applying the $t$-th $\oracle$-gate in the original circuit. \ Similarly, we let $\spz{v_1} = U_1 \spz{0}^{\otimes N}$, and $$\spz{v_t} = U_{t} (U_{g_{t-1}} \otimes I_{N-2n_{t-1}}) \spz{u_{t-1}}$$ for each $ t > 1$.
	
	From the analysis of our algorithm, over $\oracle \sim \distr_{\oracle}$, for each $t \in [T]$, with probability $1-\exp(-(2n+\varepsilon^{-1}))/T$, we have
	
	\begin{equation}\label{eq:close-vt}
	\| U_{f_{n_t}} \otimes I_{N-2n_t} \spz{v_t} - U_{g_t} \otimes I_{N-2n_t} \spz{v_t} \| \le \frac{\varepsilon^{2}}{2T}.
	\end{equation}
	
	So by a simple union bound, with probability at least $1-\exp(-(2n+\varepsilon^{-1}))$, the above bound holds for all $t \in [T]$. \ We claim that in this case, for each $t \in [T+1]$, we have
	
	\begin{equation}\label{eq:close-vt-2}
	\|\spz{v_t} - \spz{u_t}\| \le (t-1) \cdot \frac{\varepsilon^{2}}{2T}.
	\end{equation}
	
	We prove this by induction. \ Clearly it is true for $t=1$. \ When $t>1$, suppose \eqref{eq:close-vt-2} holds for $t-1$; then
	
	\begin{align*}
	\|\spz{v_t} - \spz{u_t}\| =& \| U_{t} (\oracle_{g_{t-1}} \otimes I_{N-2n_{t-1}}) \spz{v_{t-1} } - U_{t} (f_{n_{t-1}} \otimes I_{N-2n_{t-1}}) \spz{u_{t-1} } \| \\
	=& \|U_{g_{t-1}} \otimes I_{N-2n_{t-1}} \spz{v_{t-1} } - U_{f_{n_{t-1}}} \otimes I_{N-2n_{t-1}} \spz{u_{t-1} } \| \\
	\le& \|U_{g_{t-1}} \otimes I_{N-2n_{t-1}} \spz{v_{t-1} } - U_{f_{n_{t-1}}} \otimes I_{N-2n_{t-1}} \spz{v_{t-1} } \| \\
	& + \|U_{f_{n_{t-1}}} \otimes I_{N-2n_{t-1}} \spz{v_{t-1} } - U_{f_{n_{t-1}}} \otimes I_{N-2n_{t-1}} \spz{u_{t-1} } \| \\
	\le& \frac{\varepsilon^{2}}{2T} + \|\spz{u_{t-1}} - \spz{v_{t-1}}\| \le (t-1) \cdot \frac{\varepsilon^{2}}{2T},
	\end{align*}
	\noindent where the second line holds by the fact that $U_t$ is unitary, the third line holds by the triangle inequality, and the last line holds by (\ref{eq:close-vt}) and the induction hypothesis.

	{\bf Upper-bounding the error.} Therefore, with probability at least $1-\exp(-(2n+\varepsilon^{-1}))$, we have
	$$
	\| \spz{v_{T+1}} - \spz{u_{T+1}}\| = \| U \spz{0}^{\otimes N} - V \spz{0}^{\otimes N}\| \le \frac{\varepsilon^2}{2}.
	$$
	
	Now, our classical algorithm $A$ then simulates stage 2 and 3 of the $\SampBQP$ algorithm $M$ straightforwardly. \ That is, it first takes a sample $z$ by measuring $\spz{v_{T+1}}$ in the computational basis, and then outputs $A^{\mathsf{output}}(z)$ as its sample, where $A^{\mathsf{output}}$ is the classical algorithm used by $M$ in stage 3.
	
	From our previous analysis, $A$ queries the oracle only $\operatorname{poly}(n,1/\varepsilon)$ times. \ In addition, it is not hard to see that all the computations can be done in $\PSPACE$, and therefore can be implemented in $\operatorname{poly}(n,1/\varepsilon)$ time with the help of the $\TQBF$ oracle. \ So $A$ is a $\SampBPP$ algorithm.
	
	By Corollary~\ref{cor:close-dist}, with probability at least $1-\exp(-(2n+\varepsilon^{-1}))$, the distribution $\distr^{A}_{x,\varepsilon}$ outputted by $A$ satisfies
	
	$$
	\|\distr^{A}_{x,\varepsilon} - \distr^{M}_{x,\varepsilon} \| \le \sqrt{2 \cdot \frac{\varepsilon^{2}}{2}} = \varepsilon,
	$$
	\noindent and this completes the proof of Lemma \ref{lemma:simulation}.
	
\end{proofof}

\subsection{$\PH^{\TQBF,\oracle}$ is Infinite with Probability 1.}

For the second part of Theorem~\ref{theo:oralce-result-distrO}, we resort to the well-known connection between $\PH$ and constant-depth circuit lower bounds.

\subsection*{The Average Case Constant-depth Circuit Lower Bound.}

For convenience, we will use the recent breakthrough result by Rossman, Servedio, and Tan~\cite{rst}, which shows that $\PH$ is infinite relative to a random oracle with probability $1$. \ (Earlier constructions of oracles making $\PH$ infinite would also have worked for us, but a random oracle is a particularly nice choice.)

\begin{theo}
	\label{theo:average-case}
	Let $2 \leq d \leq {\frac  {c\sqrt{\log n}} {\log \log n}}$, where $c>0$ is an absolute constant. \ Let $\Sipser_d$ be the explicit $n$-variable read-once monotone depth-$d$
	formula described in \cite{rst}. \ Then any circuit $C'$ of depth at most $d-1$ and size at most $S = 2^{n^{{\frac 1 {6(d-1)}}}}$ over $\{0,1\}^n$ agrees with $\Sipser_d$ on at most $({\frac 1 2} + n^{-\Omega(1/d)})\cdot 2^n$ inputs.
\end{theo}

\subsubsection*{$\distr_{n}$ as a Distribution on $\{0,1\}^{2^{2n}}$.}

In order to use the above result to prove the second part of Theorem~\ref{theo:oralce-result-distrO}, we need to interpret $\distr_{n}$ (originally a distribution over functions mapping $\{0,1\}^{2n}$ to $\{0,1\}$) as a distribution on $\{0,1\}^{2^{2n}}$ in the following way.

Let $\tau$ be the bijection between $[2^{2n}]$ and $\{0,1\}^{2n}$ that maps an integer $i \in [2^{2n}]$ to the $i$-th binary string in $\{0,1\}^{2n}$ in lexicographic order. \ Then a function $f : \{0,1\}^{2n} \to \{0,1\}$ is equivalent to a binary string $x^f \in \{0,1\}^{2^{2n}}$, where the $i$-th bit of $x^f$, denoted $x^f_i$, equals $f(\tau(i))$. \ Clearly this is a bijection between functions from $\{0,1\}^{2n}$ to $\{0,1\}$ and binary strings in $\{0,1\}^{2^{2n}}$.

For notational simplicity, when we say a binary string $x \in \{0,1\}^{2^{2n}}$ is drawn from $\distr_{n}$, it means $x$ is generated by first drawing a sample function $f \sim \distr_{n}$ and then setting $x = x^{f}$.

Note that for $p \in \{0,1\}^{n}$, if $p$ is the $i$-th binary string in $\{0,1\}^n$, then the set $B_{n,p}$ corresponds to the bits $x_{(i-1)2^n + 1},\dotsc,x_{i 2^n}$.

\subsubsection*{Distributional Constant-Depth Circuit Lower Bound over $\distr_{n}$}

Now we are ready to state our distributional circuit lower bound over $\distr_{n}$ formally.

\begin{lemma}\label{lm:avg-lowb}
	For an integer $n$, let $N=2^n$ and $\Sipser_d$ be the $N$-variable Sipser function as in Theorem~\ref{theo:average-case}.
	
	Consider the Boolean function $(\Sipser_d \circ \OR)$ on $\{0,1\}^{N^2}$ defined as follows:
	
	Given inputs $x_1,x_2,\dotsc,x_{N^2}$, for each $1 \le i \le N$, set
	$$
	z_i := \lor_{j=(i-1)N+1}^{i N} x_j,
	$$
	and
	$$
	(\Sipser_d \circ \OR)(x) := \Sipser_d(z).
	$$
	
	Then any circuit $C'$ of depth at most $d-1$ and size at most $S = 2^{N^{{\frac 1 {6(d-1)}}}}$ over $\{0,1\}^{N^2}$ agrees with $(\Sipser_d \circ \OR)$ with probability at most ${\frac 1 2} + N^{-\Omega(1/d)}$ when inputs are drawn from the distribution $\distr_{n}$.
\end{lemma}

Before proving Lemma~\ref{lm:avg-lowb}, we show that it implies the second part of Theorem~\ref{theo:oralce-result-distrO} easily.

\begin{proofof}{the second part of Theorem~\ref{theo:oralce-result-distrO}}
	Consider the function $(\Sipser_d \circ \OR)$ defined as in Lemma~\ref{lm:avg-lowb}. \ It is easy to see that it has a polynomial-size circuit (in fact, a formula) of depth $d+1$; and by Lemma~\ref{lm:avg-lowb}, every polynomial size circuit of depth $d-1$ has at most $\frac{1}{2}+o(1)$ correlation with it when the inputs are drawn from the distribution $\distr_{n}$. \ So it follows from the standard connection between $\PH$ and $\ACz$ that $\PH^{\oracle}$ is infinite with probability $1$ when $\oracle \sim \distrO$.
\end{proofof}

Finally, we prove Lemma~\ref{lm:avg-lowb}.

\begin{proofof}{Lemma~\ref{lm:avg-lowb}}
	By Theorem~\ref{theo:average-case}, there is a universal constant $c$, such that any circuit $C$ of depth at most $d-1$ and size at most $S$ over $\{0,1\}^N$ agrees with $\Sipser_d$ on at most $\left({\frac 1 2} + N^{-c/d} \right)\cdot 2^N$ inputs.
	
	We are going to show this lemma holds for the same $c$. \ Suppose not; then we have a circuit $C$ of depth at most $d-1$ and size at most $S = 2^{N^{{\frac 1 {6(d-1)}}}}$ over $\{0,1\}^{N^2}$, such that
	
	$$
	\Pr_{x \sim \distr_{n}}[C(x) = (\Sipser_d \circ \OR)(x)] > {\frac 1 2} + N^{-c/d}.
	$$
	
	Now, for each $y_1,y_2,\dotsc,y_N \in [N]^N$, we define a distribution $\distr_{n}^{y_1,y_2,\dotsc,y_N}$ on $\{0,1\}^{N^2}$ as follows. \ To generate a sample $x \sim \distr_{n}^{y_1,y_2,\dotsc,y_N}$, we first set $x = 0^{N^2}$. \ Then for each $i \in [N]$, we set $x_{(i-1)N + y_i}$ to $1$ with probability $1/2$.
	
	By construction, we can see for all $x$ in the support of $\distr_{n}^{y_1,y_2,\dotsc,y_N}$,
	
	$$(\Sipser_d \circ \OR)(x) = \Sipser_d(x_{y_1},x_{N+y_2},x_{2 N+y_3},\dotsc,x_{(N-1)N + y_N}).$$
	
	Moreover, by definition, $\distr_{n}$ is just the average of these distributions:
	
	$$
	\distr_{n} = N^{-N} \cdot \sum_{y_1,y_2,\dotsc,y_{N}} \distr_{n}^{y_1,y_2,\dotsc,y_N}.
	$$
	
	By an averaging argument, there exist $y_1,y_2,\dotsc,y_{N} \in [N]^N$ such that
	
	$$
	\Pr_{x \sim \distr_{n}^{y_1,y_2,\dotsc,y_N}}[C(x) = (\Sipser_d \circ \OR)(x)] > {\frac 1 2} + N^{-c/d}.
	$$
	
	Setting $x_{(i-1)N + y_i}=z_i$ for each $i$, and all other inputs to 0 in the circuit $C$, we then have a circuit $D$ of size at most $S$ and depth at most $d-1$ over $\{0,1\}^N$. \ And by the construction of $\distr_{n}^{y_1,y_2,\dotsc,y_N}$ and the definition of the function $(\Sipser_d \circ \OR)$, we see that $D$ agrees with $\Sipser_d$ on at least a ${\frac 1 2} + N^{-c/d}$ fraction of inputs. \ But this is a contradiction.
	
\end{proofof}

\newcommand{\uniform}{\mathcal{U}}
\section{Maximal Quantum Supremacy for Black-Box Sampling and Relation
	Problems}
\label{sec:sepa-samp-relation}

In this section we present our results about $\ffishing$ and $\fsampling$.

We will establish an $\Omega(N/\log N)$ lower bound on the classical query complexity of $\ffishing$, as well as an optimal $\Omega(N)$ lower bound on the classical query complexity of $\fsampling$.

\subsection{Preliminaries}

We begin by introducing some useful notations. \ Throughout this section, given a function $f:\{0,1\}^n \to \{-1,1\}$,
we define the Fourier coefficient
$$
\widehat{f}(z) = 2^{-n/2} \sum_{x \in \{0,1\}^n} f(x) \cdot (-1)^{x \cdot z}
$$
for each $z \in \{0,1\}^n$.

We also define
$$
\adv(f) := 2^{-n} \cdot \sum_{z \in \{0,1\}^n,|\widehat{f}(z)| \ge 1} \widehat{f}(z)^2,
$$
and set $N=2^n$.

The following two constants will be used frequently in this section.

$$
\Qsucc = \frac{2}{\sqrt{2\pi}}\int_{1}^{+\infty} x^2e^{-x^2/2} dx \approx 0.801 \text{ and } \Rsucc = \frac{2}{\sqrt{2\pi}}\int_{1}^{+\infty} e^{-x^2/2} dx \approx 0.317.
$$

Finally, we use $\uniform_n$ to denote the uniform distribution on functions $f : \{0,1\}^{n} \to \{-1,1\}$.

\subsubsection*{An Approximate Formula for the Binomial Coefficients}

We also need the following lemma to approximate the binomial coefficients to ease some calculations in our proofs.

\begin{lemma}\label{lm:binom-approx} ((5.41) in \cite{spencer2014asymptopia})
	For value $n$ and $|k - n/2| = o(n^{2/3})$, we have
	
	$$
	\binom{n}{k} \approx \binom{n}{n/2} \cdot e^{-\frac{(k-n/2)^2}{n/2}}
	$$
	and
	$$
	\ln \binom{n}{k} = \ln \binom{n}{n/2} - \frac{(k-n/2)^2}{n/2} + o(1).
	$$
\end{lemma}

\subsection{\ffishing\ and \fsampling}

We now formally define the \ffishing\ and the \fsampling\ problems.

\begin{defi}\label{defi:Ffish-Fsamp}
	We are given oracle access to a function $f:\{0,1\}^n \to \{-1,1\}$.
	
	In \fsampling\ (or $\fsamp$ in short),  our task is to sample from a distribution $\distr$ over $\{0,1\}^n$ such that $\|\distr-\distr_f\| \le \varepsilon$, where $\distr_f$ is the distribution defined by
	$$
	\Pr_{\distr_f}[y] = 2^{-n}\widehat{f}(y)^2 = \left(\frac{1}{2^n} \sum_{x \in \{0,1\}^n} f(x) (-1)^{x \cdot y}\right)^2.
	$$
	
	In \ffishing\ (or $\ffish$ in short), we want to find a $z$ such that $|\widehat{f}(z)| \ge 1$. \
	We also define a promise version of \ffishing\ ($\pffish$ for short), where the function $f$ is promised to satisfy $\adv(f) \ge \Qsucc - \frac{1}{n}$.
\end{defi}

\subsubsection*{A Simple 1-Query Quantum Algorithm}

Next we describe a simple $1$-query quantum algorithm for both problems. \ It consists of a round of Hadamard gates, then a query to $f$, then another round of Hadamard gates, then a measurement in the computational basis.

The following lemma follows directly from the definitions of $\fsamp$ and $\ffish$.

\begin{lemma}\label{lemma:quantum-algo}
	Given oracle access to a function $f:\{0,1\}^n \to \{-1,1\}$, the above algorithm solves $\fsamp$ exactly (i.e.\ with $\varepsilon = 0$), and $\ffish$ with probability $\adv(f)$.
\end{lemma}

We can now explain the meanings of the constants $\Qsucc$ and $\Rsucc$. \ When the function $f$ is drawn from $\uniform_n$, by a simple calculation, we can see that $\Qsucc$ is the success probability for the above simple quantum algorithm on $\ffishing$, and $\Rsucc$ is the success probability for an algorithm outputting a uniform random string in $\{0,1\}^n$.

\subsection{The $\Omega(N/\log N)$ Lower Bound for \ffishing}

We begin with the $\Omega(N/\log N)$ randomized lower bound for $\ffishing$. \ Formally:

\begin{theo}\label{theo:pffish-lower-bound}
	There is no $o(N/\log N)$-query randomized algorithm that solves $\pffish$ with $\Rsucc + \Omega(1)$ success probability.
\end{theo}

To prove Theorem~\ref{theo:pffish-lower-bound}, we first show that when the function $f$ is drawn from $\uniform_n$, no classical algorithm with $o(N/\log N)$ queries can solve $\ffish$ with probability $\Rsucc + \Omega(1)$; we then show with high probability, a function $f \leftarrow \uniform_n$ satisfies the promise of $\pffish$. Formally, we have the following two lemmas.

\begin{lemma}\label{lemma:con-adv}
	For large enough $n$,
	
	$$
	\Pr_{f \leftarrow \uniform_n}\left[ \adv(f) < \Qsucc - \frac{1}{n} \right] < \frac{1}{n}.
	$$
\end{lemma}

\begin{lemma}\label{lemma:uniform-lowb}
	Over $f \leftarrow \uniform_n$, no randomized algorithm with $o(N/\log N)$ queries can solve $\ffish$  with probability
	$$
	\Rsucc + \Omega(1).
	$$
\end{lemma}

Before proving these two technical lemmas, we show that they together imply Theorem~\ref{theo:pffish-lower-bound} easily.

\begin{proofof}{Theorem~\ref{theo:pffish-lower-bound}}
	Suppose by contradiction that there is an $o(N/\log N)$ query randomized algorithm $A$ which has a $\Rsucc + \Omega(1)$ success probability for $\pffish$. \ From Lemma~\ref{lemma:con-adv}, a $1-o(1)$ fraction of all functions from $\{0,1\}^n \to \{-1,1\}$ satisfy the promise of $\pffish$. \ Therefore, when the sample function $f$ is drawn from $\uniform_f$, with probability $1-o(1)$ it satisfies the promise of $\pffish$, and consequently $A$ has a $\Rsucc + \Omega(1)$ success probability of solving $\ffish$ with that $f$. \ This means that $A$ has a success probability of
	
	$$
	(1 - o(1)) \cdot (\Rsucc + \Omega(1)) = \Rsucc + \Omega(1)
	$$
	\noindent when $f \leftarrow \uniform_n$, contradicting Lemma~\ref{lemma:uniform-lowb}.
\end{proofof}

The proof of Lemma~\ref{lemma:con-adv} is based on a tedious calculation so we defer it to Appendix~\ref{sec:missing-proofs-sepa-samp}. \ Now we prove Lemma~\ref{lemma:uniform-lowb}.

\newcommand{\hatfseen}{\widehat{f}_{\mathsf{seen}}}
\newcommand{\hatfunseen}{\widehat{f}_{\mathsf{unseen}}}
\newcommand{\badevent}{\mathcal{E}_{\mathsf{bad}}}

\begin{proofof}{Lemma~\ref{lemma:uniform-lowb}}
	By Yao's principle, it suffices to consider only deterministic algorithms, and we can assume the algorithm $A$ makes exactly $t = o(N/\log N)$ queries without loss of generality.
	
	{\bf Notations.} Suppose that at the end of the algorithm, $A$ has queried the entries in a subset $S \subseteq \{0,1\}^n$ such that $|S| = t$.
	
	For each $z \in \{0,1\}^n$, we define
	$$
	\hatfseen(z) = \frac{1}{\sqrt{t}} \sum_{x \in S} f(x) \cdot (-1)^{x \cdot z}
	$$
	and similarly
	$$
	\hatfunseen{f}(z) = \frac{1}{\sqrt{N - t}} \sum_{x \in \{0,1\}^n \setminus S} f(x) \cdot (-1)^{x \cdot z}.
	$$
	
	From the definitions of $\widehat{f}(z)$, $\hatfseen(z)$ and $\hatfunseen(z)$, and note that $N/t = \omega(\log N) = \omega(\ln N)$, we have
	\begin{align}
	\widehat{f}(z) &= \left( \sqrt{t} \cdot \hatfseen(z) + \sqrt{N-t} \cdot \hatfunseen(z) \right) \Big/ \sqrt{N} \notag \\
	&= \hatfseen(z) /\omega(\sqrt{\ln N}) + \hatfunseen(z)\cdot(1-o(1)). \label{eq:fz}
	\end{align}
	
	{\bf W.h.p. $\hatfseen(z)$ is small for all $z \in \{0,1\}^n$.} We first show that, with probability at least $1 - o(1)$ over $f \leftarrow \uniform_n$, we have $|\hatfseen(z)| \le 2\sqrt{\ln N}$ for all $z \in \{0,1\}^n$.
	
	Fix a $z \in \{0,1\}^n$, and note that for the algorithm $A$, even though which position to query next might depend on the history, the value in that position is a uniform random bit in $\{-1,1\}$. So $\hatfseen(z)$ is a sum of $t$ uniform i.i.d.\ random variables in $\{-1,1\}$.
	
	Therefore, the probability that $|\hatfseen(z)|> 2\sqrt{\ln N}$ for this fixed $z$ is
	
	$$
	\frac{2}{\sqrt{2\pi}} \int_{2\sqrt{\ln N}}^{+\infty} e^{-x^2/2} dx = o\left(\frac{1}{N}\right).
	$$
	
	Then by a simple union bound, with probability $1-o(1)$, there is no $z \in \{0,1\}^n$ such that $|\hatfseen(z)| > 2\sqrt{\ln N}$ at the end of $t$ queries. \ We denote the nonexistence of such a $z$ as the event $\badevent$.
	
	{\bf The lower bound.} In the following we condition on $\badevent$. \ We show in this case, $A$ cannot solve $\ffish$ with a success probability better than $\Rsucc$, thereby proving the lower bound.
	
	From \eqref{eq:fz}, for each $z \in \{0,1\}^n$, we have
	$$
	\widehat{f}(z) = o(1) + \hatfunseen(z) \cdot (1 - o(1)).
	$$
	
	Therefore, the probability of $|\widehat{f}(z)| \ge 1$ is bounded by the probability that $|\hatfunseen(z)| \ge 1-o(1)$. \ Since $\hatfunseen(z)$ is independent of all the seen values in $S$, we have
	
	\begin{align*}
	\Pr\left[ \hatfunseen(z) \ge 1 - o(1) \right] &= \frac{2}{\sqrt{2\pi}}\int_{1-o(1)}^{+\infty} e^{-x^2/2} dx \\
	&= \frac{2}{\sqrt{2\pi}}\int_{1}^{+\infty} e^{-x^2/2} dx + o(1)\\
	&= \Rsucc + o(1).
	\end{align*}
	
	Hence, no matter which $z$ is outputted by $A$, we have $|\widehat{f}(z)| \ge 1$ with probability at most $\Rsucc + o(1)$. \ That means that if we condition on $\badevent$, then $A$ cannot solve $\ffish$ with probability $\Rsucc + \Omega(1)$. \ As $\badevent$ happens with probability $1-o(1)$, this finishes the proof.
\end{proofof}

\subsection{The Optimal $\Omega(N)$ Lower Bound for \fsampling}

We first show that in fact, Lemma~\ref{lemma:uniform-lowb} already implies an $\Omega(N /\log N)$ lower bound for $\fsampling$, which holds for a quite large $\varepsilon$.

\begin{theo}\label{theo:large-error}
	For any $\varepsilon < \Qsucc - \Rsucc \approx 0.483$, the randomized query complexity for $\fsamp$ is $\Omega(N/\log N)$.
\end{theo}
\begin{proof}
	Note when $f \leftarrow \uniform_n$, an exact algorithm for $\fsamp$ can be used to solve $\ffish$ with probability $\Qsucc$. Hence, a sampling algorithm for $\fsamp$ with total variance $\le \varepsilon$ can solve $\ffish$ with probability at least $\Qsucc - \varepsilon$, when $f \leftarrow \uniform_n$.
	
	Then the lower bound follows directly from Lemma~\ref{lemma:uniform-lowb}.
\end{proof}

Next we prove the optimal $\Omega(N)$ lower bound for \fsampling.

\begin{theo}\label{theo:fsamp-lowerbound}
	There is a constant $\varepsilon>0$, such that any randomized algorithm solving $\fsamp$ with error at most $\varepsilon$ needs $\Omega(N)$ queries.
\end{theo}

\begin{proof}
	
	{\bf Reduction to a simpler problem.} Sampling problems are hard to approach, so we first reduce to a much simpler problem with Boolean output (``accept'' or ``reject'').
	
	Let $A$ be a randomized algorithm for $\fsamp$ with total variance $\le \varepsilon$. \ For a function $f : \{0,1\}^n \to \{-1,1\}$ and $y \in \{0,1\}^n$, we set $p_{f,y}$ to be the probability that $A$ outputs $y$ with oracle access to $f$.
	
	By the definition of $\fsamp$, for all $f$, we have
	
	$$
	\frac{1}{2} \sum_{y \in \{0,1\}^n} \left| p_{f,y} - 2^{-n}\widehat{f}(y)^2 \right| \le \varepsilon.
	$$
	
	By an averaging argument, this implies that there exists a $y^* \in \{0,1\}^n$ such that
	$$
	\Ex_{f \leftarrow \uniform_n}\left[ \left| p_{f,y^*} - 2^{-n}\widehat{f}(y^*)^2 \right| \right] \le \frac{2\varepsilon}{N}.
	$$
	
	Then by Markov's inequality, we have
	
	$$
	\left| p_{f,y} - 2^{-n}\widehat{f}(y^*)^2 \right| \le \frac{400\varepsilon}{N},
	$$
	\noindent for at least a $199/200$ fraction of $f$'s. Now we set $\varepsilon = \frac{1}{400}\cdot\frac{1}{100}$.
	
	Without loss of generality, we can assume that $y^{*}=0^n$. Let $z_i := \frac{1+f(x_i)}{2}$ (where $x_1,x_2,\dotsc,x_N$ is a lexicographic ordering of inputs), $Z := (z_1,\dotsc,z_N)$ and $|Z| := \sum_{i=1}^{N} z_i$. Then we have
	$$
	2^{-n}\widehat{f}(0^n)^2 = \left(\frac{2|Z|}{N} - 1\right)^2.
	$$
	
	Now we can simplify the question to one of how many $z_i$'s the algorithm $A$ needs to query, in order to output $0^n$ (we call it ``accept" for convenience) with a probability $p_Z = p_{f,0^n}$ that satisfies
	
	\begin{equation}\label{eq:req}
	\left|p_Z - \left(\frac{2|Z|}{N} - 1\right)^2\right| \le \frac{400\varepsilon}{N} \le \frac{0.01}{N}
	\end{equation}
	\noindent with probability at least $199/200$ over $Z \in \{0,1\}^N$.
	
	{\bf Analysis of the acceptance probability of $A$.} Without loss of generality, we can assume that $A$ non-adaptively queries $t$ randomly-chosen inputs $z_{i_1},z_{i_2},\dotsc,z_{i_t}$, and then accepts with a probability $q_k$ that depends solely on $k:=z_{i_1}+\cdots+z_{i_t}$. \ The reason is that we can change any other algorithm into this restricted form by averaging over all $N!$ permutations of $Z$ without affecting its correctness.
	
	Let $p_w$ be the probability that $A$ accepts when $|Z|=w$. \ Then
	$$
	p_w = \sum_{k=0}^{t} q_k \cdot r_{k,w},
	$$
	where  $r_{k,w} := \binom{t}{k}\binom{N-t}{w-k}\Big/\binom{N}{w}$, is the probability that $z_{i_1}+\cdots+z_{i_t}=k$ conditioned on $|Z| = w$.
	
	{\bf Construction and Analysis of the sets $U,V,W$.} Now, consider the following three sets:
	
	$$
	U := \left\{Z : \left||Z| - \frac{N}{2}\right| \le \frac{\sqrt{N}}{20} \right\},
	$$

	$$
	V := \left\{Z : \left(1-\frac{1}{20}\right) \frac{\sqrt{N}}{2} \le |Z| - \frac{N}{2} \le \frac{\sqrt{N}}{2} \right\},
	$$

	$$
	W := \left\{Z : \left(1-\frac{1}{20}\right)\sqrt{N} \le |Z| - \frac{N}{2} \le \sqrt{N} \right\}.
	$$
	
	We calculate the probability that a uniform random $Z$ belongs to these three sets. \ For a sufficiently large $N$, we have
	
	\newcommand{\erf}{\mathrm{erf}}
	
	$$
	\Pr_Z[Z \in U] \ge \erf\left(\frac{\sqrt{2}}{20}\right) - o(1) > 0.075,
	$$
	$$
	\Pr_Z[Z \in V] \ge \frac{1}{2} \cdot \left( \erf\left(\frac{\sqrt{2}}{2}\right) -\erf\left(\frac{\sqrt{2}}{2}\cdot \frac{19}{20}\right) \right) - o(1) > 0.01,
	$$
	$$
	\Pr_Z[Z \in W] \ge \frac{1}{2} \cdot \left( \erf(\sqrt{2}) -\erf\left(\sqrt{2}\cdot \frac{19}{20}\right) \right) - o(1) > 0.005.
	$$
	
	{\bf Construction and Analysis of $w_0,w_1,w_2$.}
	Since all $ \Pr_Z[Z \in U], \Pr_Z[Z \in V], \Pr_Z[Z \in W]$ $>0.005$, and recall that for at least a $1-0.005$ fraction of $Z$, we have $$\left|p_Z - \left(\frac{2|Z|}{N} - 1\right)^2\right| \le \frac{400\varepsilon}{N} \le \frac{0.01}{N}.$$
	So there must exist $w_0 \in U,w_1 \in V,w_2 \in W$ such that
	\begin{equation} \label{eq:wi-temp}
	\left|p_{w_i} - 4\cdot\left(\frac{w_i - N/2}{N}\right)^2 \right| \le \frac{0.01}{N}
	\end{equation}
	\noindent for each $i \in \{0,1,2\}$.
	
	To ease our calculation, let $u_i = \frac{w_i-N/2}{\sqrt{N}}$, then we have $w_i = N/2 + u_i \sqrt{N}$. \ By the definition of the $u_i$'s, we also have $|u_0| \le \frac{1}{20}, u_1 \in [0.475,0.5],u_2 \in [0.95,1]$.
	
	Plugging in $u_i$'s, for each $i \in \{0,1,2\}$, equation \eqref{eq:wi-temp} simplifies to
	\begin{equation} \label{eq:ui-temp}
	\left|p_{w_i} - \frac{4u_i^2}{N} \right| \le \frac{0.01}{N}.
	\end{equation}
	
	We can calculate the ranges of the $p_{w_i}$'s by plugging the ranges of the $u_i$'s,
	$$
	p_{w_0} \le \frac{0.02}{N},
	$$
	$$
	p_{w_1} \in \left[\frac{0.95^2-0.01}{N},\frac{1+0.01}{N}\right] \subseteq \left[\frac{0.89}{N},\frac{1.01}{N}\right],
	$$
	$$
	p_{w_2} \in \left[\frac{4\cdot0.95^2-0.01}{N},\frac{4+0.01}{N}\right] \subseteq \left[\frac{3.6}{N},\frac{4.01}{N}\right].
	$$
	
	We are going to show that the above is impossible when $t=o(N)$. \ That is, one cannot set the $q_k$'s in such a way that all $p_{w_i}$'s satisfy the above constraints when $t = o(N)$.
	
	{\bf It is safe to set $q_k$ to zero when $|k-t/2|$ is large.} To simplify the matters, we first show that we can set nearly all the $q_k$'s to zero. \ By the Chernoff bound without replacement, for each $w_i$ and large enough $c$ we have
	\begin{align*}
	&\sum_{k:|k-t/2|> c \sqrt{t}} r_{k,w_i}\\
	&  =\Pr\left[  \left|z_{i_{1}}%
	+\cdots+z_{i_{t}} -\frac{t}{2}\right|\geq c\sqrt{t}:\left\vert Z\right\vert =w_i = \frac{N}{2} + u_i \sqrt{N}\right] \\
	&  \le \Pr\left[  \left|z_{i_{1}}%
	+\cdots+z_{i_{t}} -\left(\frac{t}{2} + \frac{u_i t}{\sqrt{N}}\right)\right|\geq c\sqrt{t}-\left|\left(\frac{t}{2} + \frac{u_i t}{\sqrt{N}}\right) - \frac{t}{2}\right|:\left\vert Z\right\vert =w_i = \frac{N}{2} + u_i \sqrt{N}\right] \\
	&  \le \Pr\left[  \left|z_{i_{1}}%
	+\cdots+z_{i_{t}} -\left(\frac{t}{2} + \frac{u_i t}{\sqrt{N}}\right)\right|\geq c\sqrt{t}-|u_i|\sqrt{t}:\left\vert Z\right\vert = \frac{N}{2} + u_i \sqrt{N}\right] \tag{$\frac{t}{\sqrt{N}} \le \sqrt{t}$} \\
	&  \le \exp \left\{ - 2 \frac{(c\sqrt{t}-|u_i|\sqrt{t})^2}{t} \right\}\\
	&  = \exp \left\{ - 2(c-|u_i|)^2 \right\} \\
	&  \le \exp \left\{\Omega(c^2) \right\}.
	\end{align*}
	
	Then we can set $c = c_1 \sqrt{\ln N}$ for a sufficiently large constant $c_1$, so that for all $w_i$'s,
	$$
	\sum_{k:|k-t/2| > c \sqrt{t}} r_{k,w_i} \le \frac{1}{N^2}.$$
	This means that we can simply set all $q_k$'s with $|k-t/2| > c\sqrt{t}=c_1 \sqrt{t\ln N}$ to zero, and only consider $k$ such that $|k-t/2| \le c_1 \sqrt{t\ln N}$, as this only changes each $p_{w_i}$ by a negligible value. \ From now on, we call an integer $k$ {\em valid}, if $|k - t/2| \le c_1 \sqrt{t \ln N}$.
	
	{\bf Either $\frac{r_{k,w_0}}{r_{k,w_1}} \ge 0.05$ or $\frac{r_{k,w_2}}{r_{k,w_1}} \ge 10$.} Now, we are going to show the most technical part of this proof: for all valid $k$, we have either
	
	\begin{equation}\label{eq:k-claim}
	\frac{r_{k,w_0}}{r_{k,w_1}} \ge 0.05\text{ or }\frac{r_{k,w_2}}{r_{k,w_1}} \ge 10.
	\end{equation}
	
	Suppose for contradiction that there is a valid $k$ that satisfies
	\begin{equation}\label{eq:k-req}
	\frac{r_{k,w_0}}{r_{k,w_1}} < 0.05 \text{ and } \frac{r_{k,w_2}}{r_{k,w_1}} < 10.
	\end{equation}
	
	{\bf Estimation of $r_{k,w_i}$'s.} We first use Lemma~\ref{lm:binom-approx} to derive an accurate estimate of $\ln r_{k,w_i}$ for each $w_i$.
	
	We set $N_t = N - t$ for simplicity. \ Recall that
	$$
	r_{k,w} = \binom{t}{k}\binom{N_t}{w-k}\Big/\binom{N}{w}.
	$$
	
	For each $w_i$, since $|k-t/2| \le c_1 \sqrt{t\ln N}$ and $t=o(N)$, we have
	$$\left|w_i-k-\frac{N_t}{2}\right| \le \left|w_i-\frac{N}{2}\right|+\left|k-\frac{t}{2}\right|\le u_i\sqrt{N}+c_1 \sqrt{t\ln N}=o(N_t^{2/3}),$$
	\noindent and note that $|w_i-N/2| = |u_i \sqrt{N}| = o(N^{2/3})$. \ So we can apply Lemma~\ref{lm:binom-approx} to derive
	
	\begin{align*}
	\ln r_{k,w_i} &= \ln \binom{t}{k} + \ln \binom{N_t}{w_i-k} - \ln \binom{N}{w_i} \\
	&= -\frac{(w_i-k-N_t/2)^2}{N_t/2} + \frac{(w_i-N/2)^2}{N/2} + C + \ln \binom{t}{k} + o(1),
	\end{align*}
	\noindent in which $C$ is a constant that does not depend on $k$ or $w_i$.
	
	Let $d=(k-t/2)/\sqrt{t}$ (so $k = t/2 + d\sqrt{t}$), and recall that $w_i = N/2 + u_i \sqrt{N}$ for each $w_i$. \ We can further simplify the expression as
	\begin{align*}
	\ln r_{k,w_i}
	&= - \frac{(N/2+u_i\sqrt{N}-t/2-d\sqrt{t}-N_t/2)^2}{N_t/2} + \frac{(N/2+u_i\sqrt{N}-N/2)^2}{N/2} + C + \ln \binom{t}{k} + o(1) \\
	&= -\frac{(u_i \sqrt{N}-d\sqrt{t})^2}{N_t/2} +2 u_i^2 + C + \ln \binom{t}{k} + o(1).
	\end{align*}
	
	{\bf Estimation of $\frac{r_{k,w_j}}{r_{k,w_i}}$.} Note that $N_t = N-t = (1-o(1))N$. \ So we can approximate the ratio between two $r_{k,w_i}$ and $r_{k,w_j}$ by
	\begin{align*}
	\ln \frac{r_{k,w_j}}{r_{k,w_i}}
	&= \ln r_{k,w_j} - \ln r_{k,w_i}\\
	&= -\frac{(u_j \sqrt{N}-d\sqrt{t})^2}{N_t/2} +2 u_j^2 +\frac{(u_i \sqrt{N}-d\sqrt{t})^2}{N_t/2} -2 u_i^2 + o(1) \\
	&= 2u_j^2 - 2u_i^2 + \frac{((u_i+u_j)\sqrt{N}-2d\sqrt{t})(u_i-u_j)\sqrt{N}}{N_t/2} + o(1) \\
	&= 2u_j^2 - 2u_i^2 + 2(u_i^2-u_j^2) - 4d\frac{\sqrt{tN}}{N_t}(u_i-u_j) + o(1) \\
	&= - 4d\frac{\sqrt{tN}}{N_t}(u_i-u_j) + o(1).
	\end{align*}
	
	{\bf Verifying \eqref{eq:k-claim}.} Finally, to simplify matters further, we set $x = -4d\frac{\sqrt{tN}}{N_t}$, and substitute it in \eqref{eq:k-req} for $k$. \ We have
	
	$$
	\ln \frac{r_{k,w_0}}{r_{k,w_1}} = x(u_1-u_0) + o(1) < -\ln 20,
	$$
	which simplifies to
	$$
	x < \frac{-\ln 20}{u_1-u_0} + o(1) \le \frac{-\ln 20}{0.505} + o(1) \le -5.93 + o(1).
	$$
	
	Similarly, we have
	$$
	\ln \frac{r_{k,w_2}}{r_{k,w_1}} = x(u_1-u_2) + o(1) < \ln 10
	$$
	and
	$$
	x > -\frac{\ln 10}{u_2-u_1} - o(1) \ge -\frac{\ln 10}{0.45} - o(1) \ge -5.12 -o(1).
	$$
	\noindent contradiction.
	
	{\bf The lower bound.} So $\eqref{eq:k-claim}$ holds for all valid $k$, which means for all $k$ such that $|k-t/2| \le c_1\sqrt{t\ln N}$, either $\frac{r_{k,w_0}}{r_{k,w_1}} \ge 0.05$ or $\frac{r_{k,w_2}}{r_{k,w_1}} \ge 10$.
	
	Let $H$ be the set of all valid integers $k$. We set
	$$S = \left\{  k \in H : \frac{r_{k,w_0}}{r_{k,w_1}} \ge 0.05 \right\}\text{ and }T = H \setminus S.$$
	
	By \eqref{eq:k-claim}, for any $k \in T$, we have $\frac{r_{k,w_2}}{r_{k,w_1}} \ge 10$.
	
	Since $p_{w_1} = \sum_{k \in S} q_k \cdot r_{k,w_1} + \sum_{k \in T} q_k \cdot r_{k,w_1} \ge \frac{0.89}{N}$ (recall we have set all $q_k$'s to zero for $k \notin H$), we must have either $\sum_{k \in S} q_k \cdot r_{k,w_1} \ge \frac{0.445}{N}$ or $\sum_{k \in T} q_k \cdot r_{k,w_1} \ge \frac{0.445}{N}$.
	
	If $\sum_{k \in S} q_k \cdot r_{k,w_1} \ge \frac{0.445}{N}$, we have $$p_{w_0} \ge \sum_{k \in S} q_k \cdot r_{k,w_1} \cdot \frac{r_{k,w_0}}{r_{k,w_1}} \ge \frac{0.445}{N} \cdot 0.05 \ge \frac{0.022}{N},$$ which contradicts the constraint that $p_{w_0} \le \frac{0.02}{N}$. \ Otherwise, $\sum_{k \in T} q_k \cdot r_{k,w_1} \ge \frac{0.445}{N}$; then $$p_{w_2} \ge \sum_{k \in T} q_k \cdot r_{k,w_1} \cdot \frac{r_{k,w_2}}{r_{k,w_1}} \ge \frac{0.445}{N} \cdot 10 \ge \frac{4.45}{N},$$ which violates the requirement that $p_{w_2} \le \frac{4.01}{N}$.
	
	Since both cases lead to a contradiction, $A$ needs to make $\Omega(N)$ queries and this completes the proof.
	
\end{proof}

\section{Quantum Supremacy Relative to Efficiently-Computable Oracles}
\label{sec:oracle-sepa-ppoly}

We now discuss our results about quantum supremacy relative to oracles in $\ppoly$.

Building on work by Zhandry~\cite{zhandry2012construct} and Servedio and Gortler \cite{servediogortler}, we first show that, if (classical) one-way functions exist, then there exists an oracle $\oracle \in \ppoly$ such that $\BPP^{\oracle} \ne \BQP^{\oracle}$. \ Then we make a connection to the previous section by showing that, assuming the existence of (classical) subexponentially strong one-way functions, $\ffishing$ and $\fsampling$ are hard even when it is promised that the oracle is in $\ppoly$.

We also study several other complexity questions relative to $\ppoly$ oracles: for example, $\PTIME$ vs $\NP$, $\PTIME$ vs $\BPP$, and $\BQP$ vs $\SZK$. \ Since these questions are not connected directly with quantum supremacy, we will discuss them in Appendix~\ref{sec:other_ppoly}.

\subsection{Preliminaries}

Recall that an oracle $\oracle : \{0,1\}^{*} \to \{0,1\}$ is itself a language, so we say that an oracle $\oracle$ is in $\ppoly$ when the corresponding language belongs to $\ppoly$, and we use $\oracle_n$ to denote its restriction to $\{0,1\}^n$.

Given two sets $\domain$ and $\image$, we define $\image^{\domain}$ as the set of functions $f : \domain \to \image$. \ For a set $\domain$, we will sometimes abuse notation and write $\domain$ to denote the uniform distribution on $\domain$.

\subsubsection*{(Quantum) Pseudorandom Functions and Permutations}

We are going to use pseudorandom functions and permutations throughout this section, so we first review their definitions.

\begin{defi}[PRF and PRP]
	A pesudorandom function is a function $\PRF : \key \times \domain \to \image$, where $\key$ is the key-space, and $\domain$ and $\image$ are the domain and the range. $\key,\domain,\image$ are implicitly functions of the security parameter $n$.\footnote{We denote them by $\key_n,\domain_n,\image_n$ when we need to be clear about the security parameter $n$.} We write $y = \PRF_{k}(x)$.
	
	Similarly, a pesudorandom permutation is a function $\PRP : \key \times \domain \to \domain$, where $\key$ is the key-space, and $\domain$ is the domain of the permutation.
	\ $\key$ and $\domain$ are implicitly functions of the security parameter $n$. \ We write $y = \PRP_{k}(x)$. \ It is guaranteed that $\PRP_{k}$ is a permutation on $\domain$ for each $k \in \key$.
\end{defi}

For simplicity, we use $\PRF_{\key}$ to denote the distribution on functions $f : \domain \to \image$ by drawing $k \leftarrow \key$ and set $f := \PRF_k$.

We now introduce the definitions of classical and quantum security.

\begin{defi}[Classical-Security]
	A pseudorandom function $\PRF : \key \times \domain \to \image$ is (classically) secure if no classical adversary $A$ can distinguish between a truly random function and
	the function $\PRF_k$ for a random $k$ in polynomial time. That is, for every such $A$, there exists a negligible function $\varepsilon = \varepsilon(n)$ such that
	$$
	\left| \Pr_{k \leftarrow \key}[A^{\PRF_k}() = 1] - \Pr_{f \leftarrow \image^{\domain}}[A^{f}() = 1] \right| < \varepsilon.
	$$
	Also, we say that a pseudorandom function $\PRF$ is exponentially-secure, if the above holds even for classical adversaries that take $2^{O(n)}$ time.
	
	Similarly, a pseudorandom permutation $\PRP$ is (classically) secure if no classical adversary $A$ can distinguish between a truly random permutation and
	the function $\PRP_k$ for a random $k$ in polynomial time.
\end{defi}

Sometimes, especially in the context of one-way functions, we will talk about \textit{subexponential} security. \ By this we simply mean that there is no adversary running in $2^{n^{o(1)}}$ time.

\begin{defi}[Quantum-Security]
	A pseudorandom function $\PRF$ is quantum-secure if no quantum adversary $A$ making quantum queries can distinguish between a truly random function and
	the function $\PRF_k$ for a random $k$ in polynomial time.
	
	Also, a pseudorandom permutation $\PRP$ is quantum-secure if no quantum adversary $A$ making quantum queries can distinguish between a truly random permutation and
	the function $\PRP_k$ for a random $k$ in polynomial time.
	
\end{defi}

\subsubsection*{On the Existence of PRFs}

It is well-known that the existence of one-way functions implies the existence of PRFs and PRPs.

\begin{lemma}[\cite{hill,ggm,goldreich1989hard,luby1988construct}]\label{lemma:cond-exist}
	If one-way functions exist, then there exist secure PRFs and PRPs. \ Similarly, if subexponentially-secure one-way functions exist, then there exist exponentially-secure PRFs.
\end{lemma}

We remark here that these are all {\em purely classical} assumptions, which make no reference to quantum algorithms. \ Also, the latter assumption is the same one as in the famous {\em natural proofs} barrier~\cite{rr}.

\subsection{A Construction from Zhandry~\cite{zhandry2012construct}}
\label{sec:zhandry-result}

To prove our separations, we will use a construction from Zhandry~\cite{zhandry2012construct} with some modifications. \
We first construct a PRP and a PRF, and summarize some of their useful properties.

\subsubsection*{Definitions of $\PRPraw$ and $\PRFmod$}

Assuming one-way functions exist, by Lemma~\ref{lemma:cond-exist}, let $\PRPraw$ be a secure pesudorandom permutation with key-space $\keyraw$ and domain $\domainraw$. We interpret $\domainraw$ as $[N]$, where $N = N(n) = |\domainraw|$.

Then we define another pseudorandom function $\PRFmod_{(k,a)}(x) = \PRPraw_{k}((x-1) \bmod a + 1)$ where:

\begin{itemize}
	\item The key space of $\PRFmod$ is $\keymod = \keyraw \times \moduliset$ where $\moduliset$ is the set of primes in $[\sqrt{N}/4,\sqrt{N}/2]$.
	\item The domain and image are both $\domainraw$, that is, $\domainmod = \domainraw$ and $\imagemod = \domainraw$.
\end{itemize}

Note that we denote the latter one by $\PRFmod$ (not $\PRP^{\mathsf{mod}}$) because it is no longer a PRP.

\subsubsection*{Properties of $\PRPraw$ and $\PRFmod$}

We now summarize several properties of $\PRPraw$ and $\PRFmod$, which can be proved along the same lines as \cite{zhandry2012construct}.

\begin{lemma}[Implicit in Claim 1 and Claim 2 of \cite{zhandry2012construct}]\label{lemma:zhandry}
	The following statements hold when $\PRPraw$ is classical secure.
	\begin{enumerate}
		\item Both $\PRPraw$ and $\PRFmod$ are classical secure PRFs. \ Consequently, no classical algorithm $A$ can distinguish them with a non-negligible advantage.
		\item Given oracle access to $\PRFmod_{(k,a)}$ where $(k,a) \leftarrow \keymod$, there is a quantum algorithm that can recover $a$ with probability at least $1-\varepsilon$.
		\item There is a quantum algorithm that can distinguish $\PRPraw$ from $\PRFmod$ with advantage $1-\varepsilon$.
	\end{enumerate}
	
	Here $\varepsilon = \varepsilon(n)$ is a negligible function.
\end{lemma}

For completeness, we prove Lemma~\ref{lemma:zhandry} in Appendix~\ref{sec:missing-proofs-oracle-sepa}, by adapting the proofs of Claims 1 and 2 in \cite{zhandry2012construct}.

\subsection{$\BPP$ vs $\BQP$}

Next we discuss whether there is an oracle $\oracle \in \ppoly$ that separates $\BPP$ from $\BQP$. \ We show that the answer is yes provided that one-way functions exist.

\begin{theo}\label{theo:BPPvsBQP}
	Assuming one-way functions exist, there exists an oracle $\oracle \in \ppoly$ such that $\BPP^{\oracle} \ne \BQP^{\oracle}$.
\end{theo}
\begin{proof}
	
	We are going to use $\PRPraw$ and $\PRFmod$ from Section~\ref{sec:zhandry-result}.
	
	The oracle $\oracle$ will encode the truth tables of functions $f_1,f_2,\dotsc$, where each $f_n$ is a function from $\domainraw_n$ to $\imageraw_n$. \ For each $n$, with probability $0.5$ we draw $f_n$ from $\PRPraw_{\keyraw}$, that is, draw $k \leftarrow \keyraw$ and set $f_n := \PRPraw_{k}$, and with probability $0.5$ we draw $f_n$ from $\PRFmod_{\keymod}$ similarly. \ We set $L$ to be the unary language consisting of all $0^n$ for which $f_n$ is drawn from $\PRPraw_{\keyraw}$.
	
	By Lemma~\ref{lemma:zhandry}, there exists a $\BQP$ machine $M^{\oracle}$ that decides $L$ correctly on all but finite many values of $n$ with probability $1$. \ Since we can simply hardwire the values of $n$ on which $M^{\oracle}$ is incorrect, it follows that $L \in \BQP^{\oracle}$ with probability $1$.
	
	On the other hand, again by Lemma~\ref{lemma:zhandry}, no $\BPP$ machine can distinguish $\PRPraw_{\keyraw}$ and $\PRFmod_{\keymod}$ with a non-negligible advantage. \ So let $M$ be a $\BPP$ machine, and let $E_{n}(M)$ be the event that $M$ decides whether $0^{n} \in L$ correctly. \ We have
	$$
	\Pr_{\oracle} [E_{n}(M)] = \frac{1}{2} + o(1),
	$$
	even conditioning on events $E_{1}(M),\dotsc,E_{n-1}(M)$. \ Therefore, we have $\Pr_{\oracle} [\land_{i=1}^{+\infty} E_n(M)] = 0$, which means that a $\BPP$ machine $M$ decides $L$ with probability $0$. \ Since there are countably many $\BPP$ machines, it follows that $L \notin \BPP^{\oracle}$ with probability $1$. \ Hence $\BPP^{\oracle} \ne \BQP^{\oracle}$ with probability $1$.
	
	Finally, note that each $f_n$ has a polynomial-size circuit, and consequently $\oracle \in \ppoly$.
\end{proof}

\subsection{\ffishing\ and \fsampling}
\label{sec:ffish-fsamp-ppoly}

Finally, we discuss \ffishing\ and \fsampling. \ We are going to show that, assuming the existence of subexponentially-secure one-way functions, \ffishing\ and \fsampling\ are hard even when it is promised that the oracle belongs to $\ppoly$.

\begin{theo}
	Assuming the existence of subexponentially strong one-way functions, there is no polynomial-time classical algorithm that can solve \pffish\ with probability
	$$
	\Rsucc + \Omega(1),
	$$
	even when it is promised that the oracle function belongs to $\ppoly$.
\end{theo}
\begin{proof}
	By Lemma~\ref{lemma:cond-exist}, we can use our one-way function to construct an exponentially-secure pseudorandom function, $\PRF : \key \times \domain \to \image$. \ Without loss of generality, we assume that $|\image| = 2$ and $|\domain| = 2^n$. \ Then we interpret $\domain$ as the set $\{0,1\}^n$, and $\image$ as the set $\{-1,1\}$.
	
	{\bf A Concentration Inequality.} Now, consider the distribution $\PRF_{\key}$ on functions $\{0,1\}^n \to \{-1,1\}$.
	We claim that
	\begin{equation}
	\Pr_{f \leftarrow \PRF_{\key}}[ \adv(f) > \Qsucc - 1/n ] > 1 - \frac{1}{n} - o(1). \label{eq:con-adv-PRF}
	\end{equation}
	
	To see this: from Lemma~\ref{lemma:con-adv}, we have
	$$
	\Pr_{f \leftarrow \image^{\domain}}[ \adv(f) > \Qsucc - 1/n ] > 1 - \frac{1}{n}.
	$$
	
	Therefore, if $\eqref{eq:con-adv-PRF}$ does not hold, then we can construct a distinguisher between $\PRF_{\key}$ and truly random functions $\domain^{\image}$ by calculating $\adv(f)$ in $2^{O(n)}$ time. \ But this contradicts the assumption that $\PRF$ is exponentially-secure.
	
	{\bf A distributional lower bound.} Next, we show that for every polynomial-time algorithm $A$, we have
	\begin{equation}
	\Pr_{f \leftarrow \PRF_{\key}}[A^{f} \text{ solves \ffish\ correctly}] \le \Rsucc + o(1). \label{eq:distr-lowb-PRF}
	\end{equation}
	This is because when $f$ is a truly random function, from Lemma~\ref{lemma:uniform-lowb}, we have
	$$
	\Pr_{f \leftarrow \image^{\domain}}[A^{f} \text{ solve \ffish\ correctly}] \le \Rsucc + o(1).
	$$
	So if \eqref{eq:distr-lowb-PRF} does not hold, then we can construct a distinguisher between $\PRF_{\key}$ and truly random functions $\domain^{\image}$ by simulating $A^{f}$ to get its output $z$, and then checking whether $z$ is a correct solution to $\ffish$ in $2^{O(n)}$ time. \ This again contradicts our assumption that $\PRF$ is exponentially-secure.
	
	{\bf The lower bound.} Finally, we prove the theorem. Suppose for contradiction that there is such a polynomial-time algorithm $A$. Then when $f \leftarrow \PRF_{\key}$, from \eqref{eq:con-adv-PRF}, with probability $1-1/n-o(1)$, we have that $f$ satisfies the promise of $\pffish$. \ Thus, $A$ solves $\ffish$ when $f \leftarrow \PRF_{\key}$ with probability at least
	$$
	(1-o(1)) \cdot (\Rsucc + \Omega(1)) = \Rsucc + \Omega(1),
	$$
	which contradicts \eqref{eq:distr-lowb-PRF}.
\end{proof}

By a similar reduction, we can show that \fsampling\ is also hard.

\begin{cor}
	Assuming the existence of subexponentially-secure one-way functions, no polynomial-time classical algorithm can solve $\fsamp$ with error
	$$
	\varepsilon <\Qsucc - \Rsucc \approx 0.483,
	$$
	even if it is promised that the oracle function belongs to $\ppoly$.
\end{cor}
\begin{proof}
	For a function $f$, an exact algorithm for $\fsamp$ can be used to solve $\ffish$ with probability $\adv(f)$. \ Hence, a polynomial-time sampling algorithm $A$ for $\fsamp$ with error at most $\varepsilon$ can solve $\ffish$ with probability at least $\adv(f) - \varepsilon$.
	
	Note that by \eqref{eq:con-adv-PRF}, when $f\leftarrow \PRF_{\key}$, the algorithm $A$ can solve $\ffish$ with probability at least $$(\Qsucc - \frac{1}{n}  -\varepsilon) \cdot (1-o(1)) = \Qsucc - o(1) - \varepsilon.$$
	Therefore, by \eqref{eq:distr-lowb-PRF}, we must have $\varepsilon \ge \Qsucc - \Rsucc$, which completes the proof.
\end{proof}

\section{Complexity Assumptions Are Needed for Quantum Supremacy Relative to Efficiently-Computable Oracles}
\label{sec:needed}

In Section~\ref{sec:ffish-fsamp-ppoly}, we showed that the existence of subexponentially-secure one-way functions implies that \fsampling\ and \ffishing\ are classically hard, even when it is promised that the oracle function belongs to $\ppoly$. \ We also showed that if one-way functions exist, then there exists an oracle $\oracle \in \ppoly$ which separates $\BPP$ from $\BQP$.

It is therefore natural to ask whether we can prove the same statements {\em unconditionally}. \ In this section, we show that at least {\em some} complexity assumptions are needed.

\begin{theo}\label{theo:assump-needed}
	Suppose $\SampBPP = \SampBQP$ and $\NP \subseteq \BPP$. \ Then for every oracle $\oracle \in \ppoly$, we have $\SampBPP^{\oracle} = \SampBQP^{\oracle}$ (and consequently $\BPP^{\oracle} = \BQP^{\oracle}$).
\end{theo}

Much like in the proof of Theorem~\ref{theo:oralce-result-distrO}, we need to show that under the stated assumptions, every $\SampBQP$ algorithm $M$ can be simulated by a $\SampBPP$ algorithm $A$.

\begin{lemma}\label{lemma:simulation-2}
	Suppose $\SampBPP = \SampBQP$ and $\NP \subseteq \BPP$. \ Then for any polynomial $q(n)$ and any $\SampBQP$ oracle algorithm $M$, there is a $\SampBPP$ oracle algorithm $A$ such that:
	
	For every $\oracle \in \mathsf{SIZE}(q(n)),\footnote{A language is in $\mathsf{SIZE}(q(n))$ if it can be computed by circuits of size $q(n)$.}$ let $\distr^{M}_{x,\varepsilon}$ and $\distr^{A}_{x,\varepsilon}$ be the distributions output by $M^{\oracle}$ and $A^{\oracle}$ respectively on input $\langle x, 0^{1/\varepsilon} \rangle$. \ Then
	$$
	\| \distr^{M}_{x,\varepsilon} - \distr^{A}_{x,\varepsilon} \| \le \varepsilon.
	$$
\end{lemma}

Before proving Lemma~\ref{lemma:simulation-2}, we show that it implies Theorem~\ref{theo:assump-needed}.

\begin{proofof}{Theorem~\ref{theo:assump-needed}}
	Let $\oracle \in \ppoly$ be an oracle. \ Then there exists a polynomial $q(n)$ such that $\oracle \in \mathsf{SIZE}(q(n))$.
	
	Let $\mathcal{S}$ be a sampling problem in $\SampBQP^{\oracle}$. \ This means that there is a $\SampBQP$ oracle algorithm $M$, such that for all $x\in \{0,1\}^{*}$ and $\varepsilon$, we have $\|\distr^{M}_{x,\varepsilon} - \mathcal{S}_x \| \le \varepsilon$. \ Let $A_M$ be the corresponding $\SampBPP$ algorithm whose existence we've assumed, and consider the following algorithm $A'$: given input $\langle x,0^{1/\varepsilon} \rangle$, run $A_M$ on input $\langle x, 0^{2/\varepsilon} \rangle$ to get a sample from $\distr^{A_M}_{x,\varepsilon/2}$.
	
	Then we have
	\begin{align*}
	\| \distr^{A'}_{x,\varepsilon} - \mathcal{S}_x \|  &=  \| \distr^{A_M}_{x,\varepsilon/2} - \mathcal{S}_x \| \\
	&\le \|\distr^{M}_{x,\varepsilon/2} - \distr^{A_M}_{x,\varepsilon/2} \| + \|\distr^{M}_{x,\varepsilon/2} - \mathcal{S}_x \| \le 2 \cdot \frac{\varepsilon}{2} \le \varepsilon.
	\end{align*}
	This means that $A'$ solves $\mathcal{S}$ and $\mathcal{S} \in \SampBPP^{\oracle}$. \ Hence $\SampBQP^{\oracle} \subseteq \SampBPP^{\oracle}$.
\end{proofof}

Now we prove Lemma~\ref{lemma:simulation-2}. \ The simulation procedure is similar to that in Lemma~\ref{lemma:simulation}: that is, we replace each oracle gate, one by one, by a known function while minimizing the introduced error. \ The difference is that, instead of the brute-force method as in Lemma~\ref{lemma:simulation}, here we use a more sophisticated PAC learning subroutine to find an ``approximator'' to replace the oracle gates.

\begin{proofof}{Lemma~\ref{lemma:simulation-2}}
	Let $\oracle \in \mathsf{SIZE}(q(n))$; we let $f_n = \oracle_n$ for simplicity.
	
	Recall that there exists a fixed polynomial $p$, such that given input $\langle x,0^{1/\varepsilon} \rangle$, the machine $M$ first constructs a quantum circuit $C$ with $N=p(|x|,1/\varepsilon)$ qubits and $N$ gates classically ($C$ can contain $\oracle$ gates). \ Without loss of generality, we can assume for each $n$, all $f_n$ gates act only on the first $n$ qubits.
	
	For a function $f : \{0,1\}^{k} \to \{0,1\}$, recall that $U_{f}$ denotes the unitary operator mapping $\spz{i}$ to $(-1)^{f(i)} \spz{i}$ for $i \in \{0,1\}^{k}$.
	
	Suppose there are $T$ $\oracle$-gates in total, and the $i$-th $\oracle$-gate is an $f_{n_i}$ gate. \ Then the unitary operator $U$ applied by the circuit $C$ can be decomposed as
	$$
	U = U_{T+1} (U_{f_{n_T}} \otimes I_{N-n_T}) \cdots (U_{f_{n_2}} \otimes I_{N-n_2}) U_2 (U_{f_{n_1}} \otimes I_{N-n_1}) U_1,
	$$
	\noindent where the $U_i$'s are the unitary operators corresponding to the sub-circuits which don't contain an $\oracle$-gate.
	
	Again, the algorithm proceeds by replacing each $\oracle$-gate by a much simpler gate one by one, without affecting the resulting quantum state too much, and then simulating the final circuit to get a sample to output.
	
	{\bf Replacing the $t$-th $\oracle$-gate.} Suppose we have already replaced the first $t-1$ $\oracle$-gates: that is, for each $i \in [t-1]$, we replaced the $f_{n_i}$ gate (the $i$-th $\oracle$-gate) with a $g_i$ gate. \ Now we are going to replace the $t$-th $\oracle$-gate.
	
	Let
	$$
	\spz{v} = U_{t} (U_{g_{t-1}} \otimes I_{N-n_{t-1}}) \cdots (U_{g_2} \otimes I_{N-n_2}) U_2 (U_{g_1} \otimes I_{N-n_1}) U_1 \spz{0}^{\otimes N},
	$$
	\noindent which is the quantum state right before the $t$-th $\oracle$ gate in the circuit after the replacement.
	
	For brevity, we use $f$ to denote the function $f_{n_t}$, and we drop the subscript $t$ of $n_t$ when it is clear from context.
	
	{\bf Analysis of incurred error.} The $t$-th $\oracle$-gate is an $f$ gate. \ If we replace it by a $g$ gate, then the deviation caused to the quantum states is
	$$
	\| U_f \otimes I_{N-n} \spz{v} - U_g \otimes I_{N-n} \spz{v}\| =
	\| (U_f - U_g) \otimes I_{N-n} \spz{v} \|.
	$$
	
	Let $H$ be the Hilbert space corresponding to the last $N-n$ qubits, and let $\rho = \Tr_{H}[\outpt{v}]$. \ Then proceeding exactly as in Lemma~\ref{lemma:simulation}, we have
	\begin{equation}\label{eq:error-333}
	\|((U_f-U_g) \otimes I_{N-n}) \spz{v}\|^2 = 4 \cdot \Pr_{i \sim Q} [f(i) \ne g(i)],
	\end{equation}
	\noindent where $Q$ is the probability on $\{0,1\}^{n}$ defined by $Q(i) = \rpz{i}\rho \spz{i}$, and $[f(i) \ne g(i)]$ is the indicator function that takes value $1$ when $f(i) \ne g(i)$ and $0$ otherwise.
	
	\newcommand{\poly}{\operatorname{poly}}
	
	{\bf Upper bounding the deviation~\eqref{eq:error-333} vis PAC learning.} Now, we want to replace $f$ by another function $g$, so that the deviation term \eqref{eq:error-333} is minimized.
	
	By a standard result of PAC learning (cf. the book of Vapnik \cite{vapnik1998statistical}), for parameters $\varepsilon_1$ and $\delta_1$, we can take a $\poly(n,\varepsilon_1^{-1},\ln \delta_1^{-1})$ number of i.i.d.\ samples from $Q$, and then find a function $g$ in $\mathsf{SIZE}(q(n))$ which agrees with $f$ on those samples. \ Then with probability at least $1-\delta_1$, we will have
	$$
	\Pr_{i \sim Q} [f(i) \ne g(i)] \le \varepsilon_1.
	$$
	
	The choice of $\varepsilon_1$ and $\delta_1$ will be made later. \ In any case, with probability at least $1-\delta_1$, we have
	$$
	\|(U_f-U_g) \otimes I_{N-n} \spz{v}\|^2 \le 4\varepsilon_1,
	$$
	\noindent which in turn implies
	$$
	\|(U_f-U_g) \otimes I_{N-n} \spz{v}\| \le 2 \cdot \sqrt{\varepsilon_1}.
	$$
	
	\newcommand{\Cfinal}{C^{\mathsf{final}}}
	
	{\bf Analysis of the final circuit $\Cfinal$.} Suppose that at the end, for each $t \in [T]$, our algorithm has replaced the $t$-th $\oracle$-gate with a $g_t$ gate, where $g_t$ is a function from $\{0,1\}^{n_t}$ to $\{0,1\}$. \ Let $\Cfinal$ be the circuit after the replacement. \ Also, let
	$$
	V = U_{T+1} (U_{g_T} \otimes I_{N-n_T}) \cdots (U_{g_2} \otimes I_{N-n_2}) U_2 (U_{g_1} \otimes I_{N-n_1}) U_1
	$$
	\noindent be the unitary operator corresponding to $\Cfinal$.
	
	Now we set $\delta_1 = \frac{\varepsilon}{2T}$, and $\varepsilon_1 = \frac{\varepsilon^{4}}{256 T^2}$. \ Then by a union bound over all rounds, and following exactly the same analysis as in Lemma~\ref{lemma:simulation}, with probability at least $1-T \cdot \delta_1 = 1 - \varepsilon/2$, we have
	$$
	\| U \spz{0}^{\otimes N} - V \spz{0}^{\otimes N}\| \le 2 T \cdot \sqrt{\varepsilon_1} = \frac{\varepsilon^2}{8}.
	$$
	
	Our classical algorithm $A$ then simulates stages 2 and 3 of the $\SampBQP$ algorithm $M$ straightforwardly. \ It first takes a sample $z$ by measuring $V \spz{0}^{\otimes N}$ in the computational basis, and then outputs $A^{\mathsf{output}}(z)$ as its sample, where $A^{\mathsf{output}}$ is the classical algorithm used by $M$ in stage 3.
	
	By Corollary~\ref{cor:close-dist}, with probability at least $1-\varepsilon/2$, the final distribution $\distr$ on which $A$ takes samples satisfies
	
	$$
	\|\distr - \distr^{M}_{x,\varepsilon} \| \le \sqrt{2 \cdot \frac{\varepsilon^{2}}{8}} = \frac{\varepsilon}{2}.
	$$
	
	Hence, the outputted distribution $\distr^{A}_{x,\varepsilon}$ satisfies
	$$
	\|\distr^{A}_{x,\varepsilon} - \distr^{M}_{x,\varepsilon} \| \le  \varepsilon.
	$$
	
	{\bf Showing that $A$ is a $\SampBPP$ algorithm.} We still have to show that $A$ is a $\SampBPP$ oracle algorithm. \ From the previous discussion, $A$ needs to do the following non-trivial computations.
	
	\begin{itemize}
		\item Taking a polynomial number of samples from $Q$. This task is in $\SampBQP$ (no oracle involved) by definition. \ By our assumption $\SampBQP = \SampBPP$, it can be done in $\SampBPP$.
		
		\item Finding a $g \in \mathsf{SIZE}(q(n))$ such that $g$ agrees with $f$ on all the samples. \ This can be done in $\NP$, so by our assumption $\NP \subseteq \BPP$, it can be done in $\BPP$.
		
		\item Taking a sample by measuring $V \spz{0}^{\otimes N}$. \ Again, this task is in $\SampBQP$, and hence can be done in $\SampBPP$ by our assumption.
	\end{itemize}
	
	Therefore, $A$ is a $\SampBPP$ oracle algorithm.
	
\end{proofof}

\section{Open Problems}
There are many exciting open problems left by this paper; here we mention just a few.

\begin{enumerate}
	\item[(1)] Is QUATH (our assumption about the hardness of guessing whether
	$|\langle 0|C|0\rangle|^2$ is greater or less than the median) true or
	false?
	
    \item[(2)] Is Conjecture \ref{conj:adv-large} true?  That is, does a random quantum circuit on $n$ qubits sample an unbalanced distribution over $n$-bit strings with $1-1/\exp(n)$ probability?

	\item[(3)] We showed that there exists an oracle relative to which $\SampBPP = \SampBQP$ but $\PH$ is infinite. \ Can we nevertheless show that $\SampBPP=\SampBQP$ would collapse $\PH$ in the unrelativized world? \ (An affirmative answer would, of course, follow from Aaronson and Arkhipov's Permanent-of-Gaussians Conjecture \cite{aark}, as mentioned in Section \ref{CHALLENGES}.)
	
	\item[(4)] Is our classical algorithm to simulate a quantum circuit with $n$
	qubits and $m$ gates optimal?  \ Or could we reduce the complexity, say
	from $m^{O(n)}$ to $2^{O(n)} \cdot m^{O(1)}$, while keeping the space usage
	polynomial? \ Does it matter if we only want to sample from the output
	distribution, rather than actually calculating the probabilities? \ What about if we only want
    to guess an amplitude with small bias, as would be needed to refute QUATH?
	
	\item[(5)] For random quantum circuit sampling, we proved a conditional
	hardness result that talks directly about the observed outputs of a
	sampling process, rather than about the unknown distribution that's
	sampled from. \ Can we get analogous hardness results for the
	$\BosonSampling$ or $\mathsf{IQP}$ models, under some plausible hardness conjecture? \
	Note that the argument from Section \ref{sec:proposal} doesn't work directly for
	$\BosonSampling$ or $\mathsf{IQP}$, for the simple reason that in those models, the
	advantage over chance in guessing a given amplitude is at least
	$1/\exp(n)$, rather than $1/\exp(m)$ for some $m \gg n$ as is the case for random
	circuits.
	
	\item[(6)] We proved a lower bound of $\Omega(N)$ on the classical query complexity of \fsampling, for a rather small error $\varepsilon = \frac{1}{40000}$. \ The error constant does matter for sampling problems, since there is no efficient way to reduce the error in general. \ So can we discover the {\em exact threshold} $\varepsilon$ for an $\Omega(N)$ lower bound? \ That is, find the constant $\varepsilon$ such that there is an $o(N)$ query classical algorithm solving \fsampling\ with error $\varepsilon$, but any classical algorithm with error $<\varepsilon$ needs $\Omega(N)$ queries?
	
	\item[(7)] In Section~\ref{sec:oracle-sepa-ppoly}, we showed that there is an oracle $\oracle$ in $\ppoly$ separating $\BPP$ from $\BQP$, assuming that one-way functions exist. \ Is it possible to weaken the assumption to, say, $\NP \not\subset \BPP$?
\end{enumerate}

\section*{Acknowledgments}
We thank Shalev Ben-David, Sergio Boixo, Yuzhou Gu, Greg Kuperberg, John Martinis, Ashley Montanaro, John Preskill, Vadim Smelyansky, Ronald de Wolf, and Mark Zhandry for helpful discussions about the subject of this paper.

\bibliographystyle{alpha}
\bibliography{team}

\newcommand{\etalchar}[1]{$^{#1}$}
\begin{thebibliography}{PGHAG15}

\bibitem[A{\etalchar{+}}]{aar:zoo}
S.~Aaronson et~al.
\newblock {The Complexity Zoo}.
\newblock www.complexityzoo.com.

\bibitem[AA13]{aark}
S.~Aaronson and A.~Arkhipov.
\newblock The computational complexity of linear optics.
\newblock {\em Theory of Computing}, 9(4):143--252, 2013.
\newblock Earlier version in Proc. ACM STOC'2011. ECCC TR10-170,
  arXiv:1011.3245.

\bibitem[AA15]{aa:forrelation}
Scott Aaronson and Andris Ambainis.
\newblock Forrelation: A problem that optimally separates quantum from
  classical computing.
\newblock In {\em Proceedings of the Forty-Seventh Annual ACM on Symposium on
  Theory of Computing}, pages 307--316. ACM, 2015.

\bibitem[Aar10]{aar:ph}
S.~Aaronson.
\newblock {BQP} and the polynomial hierarchy.
\newblock In {\em Proc. ACM STOC}, 2010.
\newblock arXiv:0910.4698.

\bibitem[Aar14]{aaronson2014equivalence}
Scott Aaronson.
\newblock The equivalence of sampling and searching.
\newblock {\em Theory of Computing Systems}, 55(2):281--298, 2014.

\bibitem[Aar15]{aar:dwave}
S.~Aaronson.
\newblock {G}oogle, {D}-wave, and the case of the factor-10\^{}8 speedup for
  {WHAT}?, 2015.
\newblock http://www.scottaaronson.com/blog/?p=2555.

\bibitem[AB09]{arora2009computational}
Sanjeev Arora and Boaz Barak.
\newblock {\em Computational complexity: a modern approach}.
\newblock Cambridge University Press, 2009.

\bibitem[ABDK15]{aaronson2015separations}
Scott Aaronson, Shalev Ben-David, and Robin Kothari.
\newblock Separations in query complexity using cheat sheets.
\newblock {\em arXiv preprint arXiv:1511.01937}, 2015.

\bibitem[ABKM16]{abkm}
Scott Aaronson, Adam Bouland, Greg Kuperberg, and Saeed Mehraban.
\newblock The computational complexity of ball permutations.
\newblock {\em arXiv preprint arXiv:1610.06646}, 2016.

\bibitem[ABO97]{ab}
D.~Aharonov and M.~Ben-Or.
\newblock Fault-tolerant quantum computation with constant error.
\newblock In {\em Proc. ACM STOC}, pages 176--188, 1997.
\newblock quant-ph/9906129.

\bibitem[ABOE08]{abe}
D.~Aharonov, M.~Ben-Or, and E.~Eban.
\newblock Interactive proofs for quantum computations.
\newblock arXiv:0810.5375, 2008.

\bibitem[Amb05]{ambainis2005polynomial}
Andris Ambainis.
\newblock Polynomial degree and lower bounds in quantum complexity: Collision
  and element distinctness with small range.
\newblock {\em Theory of Computing}, 1(1):37--46, 2005.

\bibitem[AS04]{as}
S.~Aaronson and Y.~Shi.
\newblock Quantum lower bounds for the collision and the element distinctness
  problems.
\newblock {\em J. of the ACM}, 51(4):595--605, 2004.

\bibitem[AW09]{awig}
S.~Aaronson and A.~Wigderson.
\newblock Algebrization: a new barrier in complexity theory.
\newblock {\em ACM Trans. on Computation Theory}, 1(1), 2009.
\newblock Earlier version in Proc. ACM STOC'2008.

\bibitem[BBBV97]{bbbv}
C.~Bennett, E.~Bernstein, G.~Brassard, and U.~Vazirani.
\newblock Strengths and weaknesses of quantum computing.
\newblock {\em SIAM J. Comput.}, 26(5):1510--1523, 1997.
\newblock quant-ph/9701001.

\bibitem[BFK09]{bfk}
A.~Broadbent, J.~Fitzsimons, and E.~Kashefi.
\newblock Universal blind quantum computation.
\newblock In {\em Proc. IEEE FOCS}, 2009.
\newblock arXiv:0807.4154.

\bibitem[BG16]{bravyigosset}
Sergey Bravyi and David Gosset.
\newblock Improved classical simulation of quantum circuits dominated by
  clifford gates.
\newblock {\em arXiv preprint arXiv:1601.07601}, 2016.

\bibitem[BGS75]{baker1975relativizations}
Theodore Baker, John Gill, and Robert Solovay.
\newblock Relativizations of the {P}=?{NP} question.
\newblock {\em SIAM Journal on computing}, 4(4):431--442, 1975.

\bibitem[BHH16]{brandaoharrow}
Fernando~GSL Brand{\~a}o, Aram~W Harrow, and Micha{\l} Horodecki.
\newblock Local random quantum circuits are approximate polynomial-designs.
\newblock {\em Communications in Mathematical Physics}, 346(2):397--434, 2016.

\bibitem[BIS{\etalchar{+}}16]{martinis:supremacy}
Sergio Boixo, Sergei~V Isakov, Vadim~N Smelyanskiy, Ryan Babbush, Nan Ding,
  Zhang Jiang, John~M Martinis, and Hartmut Neven.
\newblock Characterizing quantum supremacy in near-term devices.
\newblock {\em arXiv preprint arXiv:1608.00263}, 2016.

\bibitem[BJS10]{bjs}
M.~Bremner, R.~Jozsa, and D.~Shepherd.
\newblock Classical simulation of commuting quantum computations implies
  collapse of the polynomial hierarchy.
\newblock {\em Proc. Roy. Soc. London}, A467(2126):459--472, 2010.
\newblock arXiv:1005.1407.

\bibitem[BL95]{boneh1995quantum}
Dan Boneh and Richard~J Lipton.
\newblock Quantum cryptanalysis of hidden linear functions.
\newblock In {\em Annual International Cryptology Conference}, pages 424--437.
  Springer, 1995.

\bibitem[BMS15]{bms}
Michael~J Bremner, Ashley Montanaro, and Dan~J Shepherd.
\newblock Average-case complexity versus approximate simulation of commuting
  quantum computations.
\newblock {\em arXiv preprint arXiv:1504.07999}, 2015.

\bibitem[BMS16]{bms2}
Michael~J Bremner, Ashley Montanaro, and Dan~J Shepherd.
\newblock Achieving quantum supremacy with sparse and noisy commuting quantum
  computations.
\newblock {\em arXiv preprint arXiv:1610.01808}, 2016.

\bibitem[BV97]{bv}
E.~Bernstein and U.~Vazirani.
\newblock Quantum complexity theory.
\newblock {\em SIAM J. Comput.}, 26(5):1411--1473, 1997.
\newblock Earlier version in Proc. ACM STOC'1993.

\bibitem[Che16]{chen2016note}
Lijie Chen.
\newblock A note on oracle separations for {BQP}.
\newblock {\em arXiv preprint arXiv:1605.00619}, 2016.

\bibitem[CHS{\etalchar{+}}15]{carolan}
Jacques Carolan, Christopher Harrold, Chris Sparrow, Enrique
  Mart{\'\i}n-L{\'o}pez, Nicholas~J Russell, Joshua~W Silverstone, Peter~J
  Shadbolt, Nobuyuki Matsuda, Manabu Oguma, Mikitaka Itoh, Graham~D Marshall,
  Mark~G Thompson, Jonathan C~F Matthews, Toshikazu Hashimoto, Jeremy~L
  O'Brien, and Anthony Laing.
\newblock Universal linear optics.
\newblock {\em Science}, 349(6249):711--716, 2015.

\bibitem[FFKL03]{fenner2003oracle}
Stephen Fenner, Lance Fortnow, Stuart~A Kurtz, and Lide Li.
\newblock An oracle builder’s toolkit.
\newblock {\em Information and Computation}, 182(2):95--136, 2003.

\bibitem[FH16]{farhiharrow}
Edward Farhi and Aram~W Harrow.
\newblock Quantum supremacy through the quantum approximate optimization
  algorithm.
\newblock {\em arXiv preprint arXiv:1602.07674}, 2016.

\bibitem[FR99]{fr}
L.~Fortnow and J.~Rogers.
\newblock Complexity limitations on quantum computation.
\newblock {\em J. Comput. Sys. Sci.}, 59(2):240--252, 1999.
\newblock cs.CC/9811023.

\bibitem[Fuj16]{fujii}
Keisuke Fujii.
\newblock Noise threshold of quantum supremacy.
\newblock {\em arXiv preprint arXiv:1610.03632}, 2016.

\bibitem[GGM86]{ggm}
O.~Goldreich, S.~Goldwasser, and S.~Micali.
\newblock How to construct random functions.
\newblock {\em J. of the ACM}, 33(4):792--807, 1986.
\newblock Earlier version in Proc. IEEE FOCS'1984, pp. 464-479.

\bibitem[GL89]{goldreich1989hard}
Oded Goldreich and Leonid~A Levin.
\newblock A hard-core predicate for all one-way functions.
\newblock In {\em Proceedings of the twenty-first annual ACM symposium on
  Theory of computing}, pages 25--32. ACM, 1989.

\bibitem[Has86]{hastad1986almost}
Johan Hastad.
\newblock Almost optimal lower bounds for small depth circuits.
\newblock In {\em Proceedings of the eighteenth annual ACM symposium on Theory
  of computing}, pages 6--20. ACM, 1986.

\bibitem[HILL99]{hill}
J.~H{\aa}stad, R.~Impagliazzo, L.~A. Levin, and M.~Luby.
\newblock A pseudorandom generator from any one-way function.
\newblock {\em SIAM J. Comput.}, 28(4):1364--1396, 1999.

\bibitem[IW97]{iw}
R.~Impagliazzo and A.~Wigderson.
\newblock {P=BPP} unless {E} has subexponential circuits: derandomizing the
  {XOR} {L}emma.
\newblock In {\em Proc. ACM STOC}, pages 220--229, 1997.

\bibitem[JVdN14]{jvn}
Richard Jozsa and Marrten Van~den Nest.
\newblock Classical simulation complexity of extended clifford circuits.
\newblock {\em Quantum Information \& Computation}, 14(7\&8):633--648, 2014.

\bibitem[Kal11]{kalai}
Gil Kalai.
\newblock How quantum computers fail: quantum codes, correlations in physical
  systems, and noise accumulation.
\newblock {\em arXiv preprint arXiv:1106.0485}, 2011.

\bibitem[KBF{\etalchar{+}}15]{kelly}
J~Kelly, R~Barends, AG~Fowler, A~Megrant, E~Jeffrey, TC~White, D~Sank,
  JY~Mutus, B~Campbell, Yu~Chen, et~al.
\newblock State preservation by repetitive error detection in a superconducting
  quantum circuit.
\newblock {\em Nature}, 519(7541):66--69, 2015.

\bibitem[Kut05]{kutin2005quantum}
Samuel Kutin.
\newblock Quantum lower bound for the collision problem with small range.
\newblock {\em Theory of Computing}, 1(1):29--36, 2005.

\bibitem[Lev03]{levin}
Leonid~A Levin.
\newblock The tale of one-way functions.
\newblock {\em Problems of Information Transmission}, 39(1):92--103, 2003.

\bibitem[LR88]{luby1988construct}
Michael Luby and Charles Rackoff.
\newblock How to construct pseudorandom permutations from pseudorandom
  functions.
\newblock {\em SIAM Journal on Computing}, 17(2):373--386, 1988.

\bibitem[MFF14]{mff}
Tomoyuki Morimae, Keisuke Fujii, and Joseph~F Fitzsimons.
\newblock Hardness of classically simulating the one-clean-qubit model.
\newblock {\em Physical review letters}, 112(13):130502, 2014.

\bibitem[MS08]{markovshi}
Igor~L Markov and Yaoyun Shi.
\newblock Simulating quantum computation by contracting tensor networks.
\newblock {\em SIAM Journal on Computing}, 38(3):963--981, 2008.

\bibitem[NW94]{nisan1994hardness}
Noam Nisan and Avi Wigderson.
\newblock Hardness vs randomness.
\newblock {\em Journal of computer and System Sciences}, 49(2):149--167, 1994.

\bibitem[PGHAG15]{pgha}
Borja Peropadre, Gian~Giacomo Guerreschi, Joonsuk Huh, and Al{\'a}n
  Aspuru-Guzik.
\newblock Microwave boson sampling.
\newblock {\em arXiv preprint arXiv:1510.08064}, 2015.

\bibitem[Pre12]{preskill:solvay}
John Preskill.
\newblock Quantum computing and the entanglement frontier.
\newblock {\em arXiv preprint arXiv:1203.5813}, 2012.

\bibitem[RR97]{rr}
A.~A. Razborov and S.~Rudich.
\newblock Natural proofs.
\newblock {\em J. Comput. Sys. Sci.}, 55(1):24--35, 1997.
\newblock Earlier version in Proc. ACM STOC'1994, pp. 204-213.

\bibitem[RST15]{rst}
Benjamin Rossman, Rocco~A Servedio, and Li-Yang Tan.
\newblock An average-case depth hierarchy theorem for boolean circuits.
\newblock In {\em Foundations of Computer Science (FOCS), 2015 IEEE 56th Annual
  Symposium on}, pages 1030--1048. IEEE, 2015.

\bibitem[Rud16]{rudolph:optimistic}
Terry Rudolph.
\newblock Why {I} am optimistic about the silicon-photonic route to quantum
  computing.
\newblock {\em arXiv preprint arXiv:1607.08535}, 2016.

\bibitem[Sav70]{savitch}
W.~J. Savitch.
\newblock Relationships between nondeterministic and deterministic tape
  complexities.
\newblock {\em J. Comput. Sys. Sci.}, 4(2):177--192, 1970.

\bibitem[SG04]{servediogortler}
Rocco~A Servedio and Steven~J Gortler.
\newblock Equivalences and separations between quantum and classical
  learnability.
\newblock {\em SIAM Journal on Computing}, 33(5):1067--1092, 2004.

\bibitem[Sha92]{shamir}
A.~Shamir.
\newblock {IP=PSPACE}.
\newblock {\em J. of the ACM}, 39(4):869--877, 1992.
\newblock Earlier version in Proc. IEEE FOCS'1990, pp. 11-15.

\bibitem[Sho97]{shor}
P.~W. Shor.
\newblock Polynomial-time algorithms for prime factorization and discrete
  logarithms on a quantum computer.
\newblock {\em SIAM J. Comput.}, 26(5):1484--1509, 1997.
\newblock Earlier version in Proc. IEEE FOCS'1994. quant-ph/9508027.

\bibitem[Spe14]{spencer2014asymptopia}
Joel Spencer.
\newblock {\em Asymptopia}, volume~71.
\newblock American Mathematical Soc., 2014.

\bibitem[SV03]{sahai2003complete}
Amit Sahai and Salil Vadhan.
\newblock A complete problem for statistical zero knowledge.
\newblock {\em Journal of the ACM (JACM)}, 50(2):196--249, 2003.

\bibitem[TD04]{td}
B.~M. Terhal and D.~P. DiVincenzo.
\newblock Adaptive quantum computation, constant-depth circuits and
  {A}rthur-{M}erlin games.
\newblock {\em Quantum Information and Computation}, 4(2):134--145, 2004.
\newblock quant-ph/0205133.

\bibitem[Tod91]{toda}
S.~Toda.
\newblock {PP} is as hard as the polynomial-time hierarchy.
\newblock {\em SIAM J. Comput.}, 20(5):865--877, 1991.
\newblock Earlier version in Proc. IEEE FOCS'1989, pp. 514-519.

\bibitem[Vap98]{vapnik1998statistical}
Vladimir~Naumovich Vapnik.
\newblock {\em Statistical learning theory}, volume~1.
\newblock Wiley New York, 1998.

\bibitem[Yao85]{yao1985separating}
Andrew Chi-Chih Yao.
\newblock Separating the polynomial-time hierarchy by oracles.
\newblock In {\em 26th Annual Symposium on Foundations of Computer Science
  (sfcs 1985)}, 1985.

\bibitem[Zha12]{zhandry2012construct}
Mark Zhandry.
\newblock How to construct quantum random functions.
\newblock In {\em Foundations of Computer Science (FOCS), 2012 IEEE 53rd Annual
  Symposium on}, pages 679--687. IEEE, 2012.

\bibitem[Zha16]{zhandry2016note}
Mark Zhandry.
\newblock A note on quantum-secure prps.
\newblock {\em arXiv preprint arXiv:1611.05564}, 2016.

\end{thebibliography}

\appendix

\section{Other Results on Oracle Separations in $\ppoly$}
\label{sec:other_ppoly}

In this section we discuss the rest of our results on complexity theory relative to oracles in $\ppoly$ (see Figure~\ref{fig:ppoly-diagram} for an overview). \ For the definitions of the involved complexity classes, see for example \cite{aar:zoo}.

We first discuss $\PTIME$ and $\NP$. \ We observe that there exists an oracle $\oracle \in \ppoly$ such that $\PTIME^{\oracle} \ne \NP^{\oracle}$ {\em unconditionally}, and no oracle $\oracle \in \ppoly$ can make $\PTIME = \NP$ unless $\NP \subset \ppoly$.

Then we discuss $\PTIME$ and $\BPP$. \ We first prove that the standard derandomization assumption (there exists a function $f \in \textsf{E} = \DTIME(2^{O(n)})$ that requires a $2^{\Omega(n)}$-size circuit) also implies that $\PTIME^{\oracle} = \BPP^{\oracle}$ for all $\oracle \in \ppoly$. \ Then, surprisingly, we show that the converse also holds! \ I.e., if no such $f$ exists, then there exists an oracle $\oracle \in \ppoly$ such that $\PTIME^{\oracle} \ne \BPP^{\oracle}$.

Finally, we discuss $\BQP$ and $\SZK$. \ We show that assuming the existence of one-way functions, there exist oracles in $\ppoly$ that separate $\BQP$ from $\SZK$, and also $\SZK$ from $\BQP$.

We will need to use quantum-secure pseudorandom permutations. \ By a very recent result of Zhandry~\cite{zhandry2016note}, their existence follows from the existence of quantum one-way functions.

\begin{lemma}[~\cite{zhandry2016note}]\label{lemma:zhandry-QPRPs}
	Assuming quantum one way functions exist, there exist quantum-secure PRPs.
\end{lemma}

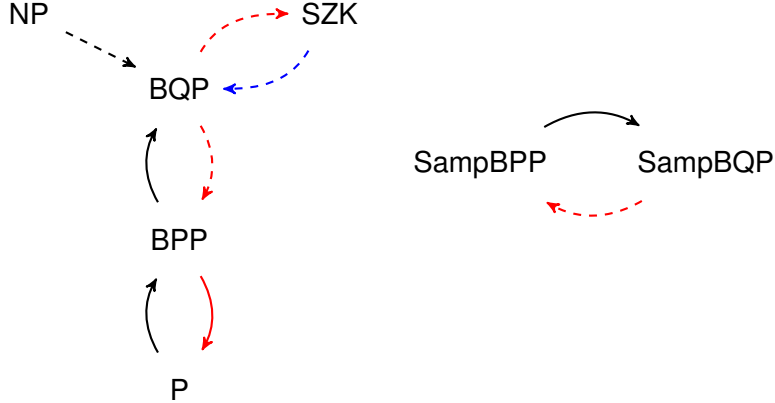
\begin{figure}[H]
	\centering
	\hspace*{-7em} 
	\begin{tikzpicture}[->,>=stealth',shorten >=1pt,auto,
	semithick,scale = 2.0]
	\tikzstyle{every state}=[draw=none, text=black, rectangle]
	\tikzstyle{contain}=[thick]
	\tikzstyle{separation}=[thick, dashed]
	\tikzstyle{contain as}=[thick, color = red]
	\tikzstyle{separation as}=[thick,dashed, color = red, dashed]
	\tikzstyle{separation as quantum}=[thick,dashed, color = blue, dashed]
	
	\node [state] (P) at (0, 1) {$\PTIME$};
	\node [state] (BPP) at (0, 2) {$\BPP$};
	\node [state] (BQP) at (0, 3) {$\BQP$};
	\node [state] (SZK) at (1, 3.5) {$\SZK$};
	\node [state] (NP) at (-1, 3.5) {$\NP$};
	\node [state] (SampBPP) at (2, 2.5) {$\SampBPP$};
	\node [state] (SampBQP) at (3.5, 2.5) {$\SampBQP$};
	\draw[contain] (P) to [out=120, in=240] (BPP);
	\draw[contain as] (BPP) to [out=300, in=60] (P);
	
	\draw[contain] (BPP) to [out=120, in=240] (BQP);
	\draw[separation as] (BQP) to [out=300, in=60] (BPP);
	
	\draw[separation as] (BQP) to [out=60, in=180] (SZK);
	\draw[separation as quantum] (SZK) to [out=240,in=0] (BQP);
	
	\draw[separation] (NP) -- (BQP);

	\draw[contain] (SampBPP) to [out=30,in=150] (SampBQP);
	\draw[separation as] (SampBQP) to [out=210,in=330] (SampBPP);
	\end{tikzpicture}
	
	\newcommand{\class}{\mathcal{C}}
	
	\caption{$\class_1 \to \class_2$ indicates $\class_1$ is contained in $\class_2$ respect to every oracle in $\ppoly$, and $\class_1 \dashrightarrow \class_2$ denotes that there is an oracle $\oracle \in \ppoly$ such that $\class_1^{\oracle} \not\subset \class_2^{\oracle}$. {\color{red}Red} indicates this statement is based on the existence of {\em classical one-way functions}, {\color{blue}Blue} indicates the statement is based on the existence of {\em quantum one-way functions}, and Black indicates the statement holds unconditionally.}
	\label{fig:ppoly-diagram}
\end{figure}

\subsection{$\PTIME$, $\BPP$, $\BQP$ vs $\NP$}
We begin with the relationships of $\PTIME$, $\BPP$, and $\BQP$ to $\NP$ relative to oracles in $\ppoly$.

The first observation is that using the function $\OR$ and standard diagonalization techniques, together with the fact that $\OR$ is hard for quantum algorithms~\cite{bbbv}, we immediately have:

\begin{ob}
	There is an oracle $\oracle \in \ppoly$ such that $\NP^{\oracle} \not\subset \BQP^{\oracle}$.
\end{ob}

On the other side, we also show that unless $\NP \subset \ppoly$ ($\textsf{BQP/poly}$), there is no oracle $\oracle \in \ppoly$ such that $\NP^{\oracle} \subseteq \BPP^{\oracle}$ ($\BQP^{\oracle}$).

\begin{theo}
	Unless $\NP \subset \ppoly$, there is no oracle $\oracle \in \ppoly$ such that $\NP^{\oracle} \subseteq \BPP^{\oracle}$. \ Likewise, there is no oracle $\oracle \in \ppoly$ such that $\NP^{\oracle} \subseteq \BQP^{\oracle}$ unless $\NP \subseteq \textsf{BQP/poly}$.
\end{theo}

\begin{proof}
	Suppose there is an oracle $\oracle \in \ppoly$ such that $\NP^{\oracle} \subseteq \BPP^{\oracle}$. Since $\BPP \subset \ppoly$, and $\mathsf{P^{\oracle}/poly} \subseteq \ppoly$ (since the relevant parts of the oracle $\oracle$ can be directly supplied to the $\ppoly$ algorithm), we have $\NP \subseteq \NP^{\oracle} \subset \ppoly$. \ The second claim can be proved in the same way.
\end{proof}

The following corollary is immediate.

\begin{cor}
	There is an oracle $\oracle \in \ppoly$ such that $\PTIME^{\oracle} \ne \NP^{\oracle}$, and there is no oracle $\oracle \in \ppoly$ such that $\PTIME^{\oracle} = \NP^{\oracle}$ unless $\NP \subset \ppoly$.
\end{cor}

\subsection{$\PTIME$ vs $\BPP$}

Next we consider the relationship between $\PTIME$ and $\BPP$. \ It is not hard to observe that the standard derandomization assumption for $\PTIME = \BPP$ is in fact strong enough to make $\PTIME^{\oracle} = \BPP^{\oracle}$ for every oracle $\oracle$ in $\ppoly$.

Given a function $f : \{0,1\}^{n} \to \{0,1\}$, let $\Hwrs(f)$ be the minimum size of circuits computing $f$ exactly.

\begin{ob}[Implicit in~\cite{nisan1994hardness,iw}, see also Theorem 20.7 in~\cite{arora2009computational}]
	If there exists a function $f \in \textsf{E} = \DTIME(2^{O(n)})$ and $\varepsilon > 0$ such that $\Hwrs(f) \ge 2^{\varepsilon n}$ for sufficiently large $n$, then $\BPP^{\oracle} = \PTIME^{\oracle}$ for every $\oracle \in \ppoly$.
\end{ob}
\begin{proofsketch}
	From~\cite{nisan1994hardness} and \cite{iw}, the assumption leads to a strong PRG which is able to fool circuits of a fixed polynomial size with a logarithmic seed length.
	
	An algorithm with an oracle $\oracle \in \ppoly$ with a certain input can still be represented by a polynomial size circuit, so we can still enumerate all possible seeds to get a deterministic algorithm.
\end{proofsketch}

Surprisingly, we show that condition is not only sufficient, but also necessary.

\begin{theo}
	If for every $f \in \textsf{E} = \DTIME(2^{O(n)})$ and $\varepsilon > 0$, there are infinitely many $n$'s with $\Hwrs(f) < 2^{\varepsilon n}$, then there exists an oracle $\oracle \in \ppoly$ such that $\BPP^{\oracle} \ne \PTIME^{\oracle}$.
\end{theo}

\begin{proof}
	For simplicity, in the following we will specify an oracle $\oracle$ by a sequence of functions $\{f_{i}\}$, where each $f_{i}$ is a function from $\{0,1\}^{n_i} \to \{0,1\}$ and the sequence $\{n_i\}$ is strictly increasing. \ That is, $\oracle_{n_i}$ is set to $f_i$, and $\oracle$ maps all strings with length not in $\{n_i\}$ to $0$.
	
	As there are only countably many $\PTIME$ oracle TM machines, we let $\{A_i\}_{i=1}^{+\infty}$ be an ordering of them.
	
	{\bf The $\GapMaj$ function.} Recall that the gapped-majority function, $\GapMaj : \{0,1\}^N \to \{0,1\}$, which outputs $1$ if the input has Hamming weight $\ge 2N/3$, or $0$ if the input has Hamming weight $\le N/3$, and is undefined otherwise, is the function which separates $\PTIME$ and $\BPP$ in the query complexity world. \ We are going to encode inputs to $\GapMaj$ in the oracle bits to achieve our separation.
	
	We call an oracle \textit{valid}, if for each $n$, either $|\oracle_n^{-1}(0)| \ge \frac{2}{3} \cdot 2^{n}$ or $|\oracle_n^{-1}(1)| \ge \frac{2}{3} \cdot 2^{n}$. \ That is, if we interpret $\oracle_n$ as a binary string with length $2^{n}$, then $\GapMaj(\oracle_n)$ is defined.
	
	{\bf The language $L^{\oracle}$.} For a valid oracle $\oracle$, we define the following language:
	$$
	L_{\oracle} =  \{0^n : \GapMaj(\oracle_n) = 1 \}.
	$$
	Clearly, this language lies in $\BPP^{\oracle}$. \ To prove the theorem, we will construct a valid oracle $\oracle$ such that $L_{\oracle} \not\in \PTIME^{\oracle}$.
	
	{\bf Construction of $\oracle$.} To construct such an oracle, we resort to the standard diagonalization method: for each integer $i$, we find an integer $n_i$ and set the function $\oracle_{n_i}$ so that the machine $A_i$ can't decide $0^{n_i}$ correctly. \ In order to do this, we will make sure that each $A_i$ can only see $0$ when querying the function $\oracle_{n_i}$. \ Since $A_i$ can only see a polynomial number of bits, we can set the remaining bits in $\oracle_{n_i}$ adversarially.
	
	\newcommand{\oraclepart}{\oracle_{\mathsf{part}}}
	
	Let $\oraclepart^i$ be the oracle specified by $\{\oracle_{n_j}\}_{j=1}^{i}$, and let $T_i$ be the maximum integer such that a bit in $\oracle_{T_i}$ is queried by $A_{i}$ when running on input $0^{n_{i}}$. \ Observe that by setting $n_{i+1} > T_i$, we can make sure that $A_{i}^{\oracle}(0^{n_i}) = A_{i}^{\oraclepart^i}(0^{n_i})$ for each $i$.
	
	{\bf Diagonalization against $A_i$.} Suppose we have already constructed $\oracle_{n_1},\dotsc,\oracle_{n_{i-1}}$, and we are going to deal with $A_i$. \ Since $A_i$ is a $\PTIME$ machine, there exists a constant $c$ such that $A_i$ runs in at most $n^{c}$ steps for inputs with length $n$. \ Thus, $A_i$ can query at most $n^{c}$ values in $\oracle_n$ on input $0^n$.
	
	{\bf Construction and Analysis of $f$.} Now consider the following function $f$, which analyzes the behavior of $A_i^{\oraclepart^{i-1}}$:
	
	\begin{itemize}
		\item given an input $x \in \{0,1\}^*$, let $m = |x|$;
		\item the first $m_1 = \floor{m/5c}$ bits of $x$ encode an integer $n \in [2^{m_1}]$;
		\item the next $m_2 = m - m_1$ bits of $x$ encode a string $p \in \{0,1\}^{m_2}$;
		\item $f(x) = 1$ iff $A_i^{\oraclepart^{i-1}}(0^n)$ has queried $\oracle_n(z)$ for an $z \in \{0,1\}^{n}$ with $p$ as a prefix.\footnote{For simplicity, we still use $\oracle_n$ to denote the restriction of $\oraclepart^{i-1}$ on $\{0,1\}^n$.}
	\end{itemize}
	
	It is not hard to see that $f \in \mathsf{E}$: the straightforward algorithm which directly simulates $A_i^{\oraclepart^{i-1}}(0^n)$ runs in $O(n^c) = 2^{O(m/5c \cdot c)} = 2^{O(m)}$ time (note that the input length is $m = |x|$). \ Therefore, by our assumption, there exists an integer $m$ such that $2^{\floor{m/5c}} > \max(T_{i-1},n_{i-1})$ and $\Hwrs(f_m) < 2^{m/c}$. \ Then we set $n_i = 2^{\floor{m/5c}}$.
	
	{\bf Construction and Analysis of $\oracle_{n_i}$.} Now, if $A_i^{\oraclepart^{i-1}}(0^{n_i}) = 1$ , we set $\oracle_{n_i}$ to be the constant function $\mathbf{0}$, so that $L^{\oracle}(0^{n_i}) = 0$.
	
	Otherwise, $A_i^{\oraclepart^{i-1}}(0^{n_i}) = 0$. \ We define a function $g:\{0,1\}^{n_i} \to \{0,1\}$ as follows: $g(z) = 1$ iff $A_i^{\oraclepart^{i-1}}(0^{n_i})$ has queried $\oracle_{n_i}(z')$ for an $z' \in \{0,1\}^{n_i}$ such that $z$ and $z'$ share a prefix of length $m-\floor{m/5c}$. \ Note that $g(z)$ can be implemented by hardwiring $n_i$ and $z_{1\dotsc m - \floor{m/5c}}$ (that is, the first $m-\floor{m/5c}$ bits of $z$) into the circuit for $f_m$, which means that there is a circuit of size $2^{m/c} = n_i^{O(1)}$ for $g$. \ We set $\oracle_{n_i} := \neg g$.
	
	From the definition of $g$ and the fact that $A_i^{\oraclepart^{i-1}}(0^{n_i})$ makes at most $n_i^{c}$ queries, there is at most a
	$$
	\frac{n_i^c}{2^{m-\floor{m/5c}}} <
	\frac{n_i^c}{2^{4c \floor{m/5c}}} = n_i^{-3c}
	$$
	\noindent fraction of inputs that are $0$ in $\neg g$. \ Hence, $\GapMaj(\neg g) = 1$ and $L^{\oracle}(0^{n_i}) = 1$.
	
	We claim that in both cases, we have $A_i^{\oraclepart^{i-1}}(0^{n_i}) = A_i^{\oraclepart^{i}}(0^{n_i})$. \ This holds trivially in the first case since we set $\oracle_{n_i} := \mathbf{0}$. \ For the second case, note from the definition of $g$ that all queries by $A_i^{\oraclepart^{i}}(0^{n_i})$ to $\oracle_{n_i}$ return $0$, and hence $A_i$ will behave exactly the same.
	
	Finally, since we set $n_i > T_{i-1}$ for each $i$, we have $A_i^{\oracle}(0^{n_i}) = A_i^{\oraclepart^{i}}(0^{n_i}) \ne L^{\oracle}(0^{n_i})$, which means that no $A_i$ can decide $L^{\oracle}$.
	
\end{proof}

\subsection{$\BQP$ vs $\SZK$}

Next we investigate the relationship between $\BQP$ and $\SZK$ relative to oracles in $\ppoly$. \ We first show that, by using quantumly-secure pseudorandom permutations, as well as the quantum lower bound for distinguishing permutations from $2$-to-$1$ functions \cite{as}, we can construct an oracle in $\ppoly$ which separates $\SZK$ from $\BQP$.

\begin{theo}
	Assuming quantum-secure one way functions exist, there exists an oracle $\oracle \in \ppoly$ such that $\SZK^{\oracle} \not\subset \BQP^{\oracle}$.
\end{theo}
\begin{proof}
	Let $\PRP$ be a quantum-secure pseudorandom permutation from $\key \times \domain \to \domain$, whose existence is guaranteed by Lemma~\ref{lemma:zhandry-QPRPs}.
	
	\newcommand{\PRFtto}{\PRF^{2 \to 1}}
	\newcommand{\keytto}{\key^{2 \to 1}}
	
	We first build a pseudorandom 2-to-1 function from $\PRP$. \ We interpret $\domain$ as $[N]$ where $N = |\domain|$, and assume that $N$ is even. \ We construct $\PRFtto : (\key \times \key) \times \domain \to \domain$ as follows:
	
	\begin{itemize}
		\item The key space $\keytto$ is $\key \times \key$. That is, a key $k \in \keytto$ is a pair of keys $(k_1,k_2)$.
		\item $\PRFtto_{(k_1,k_2)}(x) := \PRP_{k_2}( (\PRP_{k_1}(x) \bmod N/2) + 1 )$.
	\end{itemize}
	
	Note that $\PRFtto$ would be a uniformly random 2-to-1 function from $[N] \to [N]$, if $\PRP_{k_1}$ and $\PRP_{k_2}$ were replaced by two uniformly random permutations on $[N]$. \ Hence, by a standard reduction argument, $\PRFtto$ is a quantumly-secure pseudorandom 2-to-1 function. \ That is, for any polynomial-time quantum algorithm $A$, we have
	
	$$
	\left| \Pr_{k \leftarrow \keytto}[A^{\PRFtto_k}() = 1] - \Pr_{f \leftarrow \mathsf{F}^{2 \to 1}_{\domain}}[A^{f}() = 1] \right| < \varepsilon,
	$$
	where $\varepsilon$ is a negligible function and $\mathsf{F}^{2 \to 1}_{\domain}$ is the set of 2-to-1 functions from $\domain \to \domain$.
	
	Also, from the definition of $\PRP$, we have
	
	$$
	\left| \Pr_{k \leftarrow \key^\PRP}[A^{\PRP_k}() = 1] - \Pr_{f \leftarrow \mathsf{Perm}_{\domain}}[A^{f}() = 1] \right| < \varepsilon,
	$$
	\noindent where $\mathsf{Perm}_{\domain}$ is the set of permutations on $\domain$.
	
	From the results of Aaronson and Shi~\cite{as}, Ambainis~\cite{ambainis2005polynomial} and Kutin~\cite{kutin2005quantum}, no $o(N^{1/3})$-query quantum algorithm can distinguish a random permutation from a random 2-to-1 function. \ Therefore, we have
	$$
	\left| \Pr_{f \leftarrow \mathsf{F}^{2 \to 1}_{\domain}}[A^{f}() = 1] - \Pr_{f \leftarrow \mathsf{Perm}_{\domain}}[A^{f}() = 1] \right| < o(1).
	$$
	
	Putting the above three inequalities together, we have
	$$
	\left| \Pr_{k \leftarrow \keytto}[A^{\PRFtto_k}() = 1] - \Pr_{k \leftarrow \key^\PRP}[A^{\PRP_k}() = 1] \right| < o(1),
	$$
	\noindent which means $A$ cannot distinguish $\PRFtto_{\keytto}$ and $\PRP_{\key^\PRP}$.
	
	On the other side, an $\SZK$ algorithm can easily distinguish a permutation from a two-to-one function. \ Therefore, we can proceed exactly as in Theorem~\ref{theo:BPPvsBQP} to construct an oracle $\oracle \in \ppoly$ such that $\SZK^{\oracle} \not\subset \BQP^{\oracle}$.
\end{proof}

Very recently, Chen~\cite{chen2016note} showed that, based on a construction similar to the ``cheat-sheet" function by Aaronson, Ben-David and Kothari~\cite{aaronson2015separations}, we can take any function which is hard for $\BPP$ algorithms, and turn it into a function which is hard for $\SZK$ algorithms in a black-box fashion. \ We are going to adapt this construction, together with a PRF, to build an oracle in $\ppoly$ which separates $\BQP$ from $\SZK$.

\begin{theo}
	Assuming one-way functions exist, there exists an oracle $\oracle \in \ppoly$ such that $\BQP^{\oracle} \not\subset \SZK^{\oracle}$.
\end{theo}
\begin{proof}
	
	We will use the $\PRFmod : \keymod \times \domainmod \to \domainmod$ defined in Section~\ref{sec:zhandry-result} here. \ For simplicity, we will use $\domain$ to denote $\domainmod$ in this proof. Recall that $\domain$ is interpreted as $[N]$ for $N = N(n) = |\domain|$.
	
	{\bf Construction of distributions $\distr_n^{i}$.} For each $n$, we define distributions $\distr_n^{0}$ and $\distr_n^{1}$ on $(\domain_n \to \domain_n) \times \{0,1\}^{\sqrt{N}/2}$ as follows. \ We draw a function $f_n : \domain \to \domain$ from $\PRFmod_{\keymod}$, that is, we draw $(k,a) \leftarrow \keymod = \keyraw \times A$, and set $f_n := \PRFmod_{(k,a)}$; then we let $z=0^{\sqrt{N}/2}$ first, and set $z_a = i$ in $\distr_n^{i}$; finally we output the pair $(f,z)$ as a sample.
	
	{\bf Distinguishing $\distr_n^0$ and $\distr_n^1$ is hard for $\SZK$.} Recall that $\SZK$ is a {semantic} class. \ That is, a given protocol $\Pi$ might be invalid with different oracles or different inputs (i.e., the protocol might not satisfy the zero-knowledge constraint, or the verifier might accept with a probability that is neither $\ge 2/3$ nor $\le 1/3$). \ We write $\Pi^{(f,z)}() = \bot$ when $\Pi$ is invalid given oracle access to $(f,z)$.
	
	We claim that for any protocol $\Pi$, one of the following two claims must hold for sufficiently large $n$:
	
	\begin{itemize}
		\item[(A)] $\Pr_{(f,z) \leftarrow \distr_n^{0}}\left[\Pi^{(f,z)}() = \bot\right] > 0.1$ or $\Pr_{(f,z) \leftarrow \distr_n^{1}}\left[\Pi^{(f,z)}() = \bot\right] > 0.1$.
		\item[(B)] $\left| \Pr_{(f,z) \leftarrow \distr_n^{0}}\left[\Pi^{(f,z)}() = 1\right] - \Pr_{(f,z) \leftarrow \distr_n^{1}}\left[\Pi^{(f,z)}() = 1\right] \right| < 0.2$.
	\end{itemize}
	
	That is, either $\Pi$ is invalid on a large fraction of oracles, or else $\Pi$ cannot distinguish $\distr_n^0$ from $\distr_n^1$ with a very good probability.
	
	{\bf Building a $\BPP$ algorithm to break $\PRFmod$.}  Suppose for a contradiction that there are infinitely many $n$ such that none of (A) and (B) hold. \ Without loss of generality, we can assume that
	$$
	\Pr_{(f,z) \leftarrow \distr_n^{1}}\left[\Pi^{(f,z)}() = 1\right] - \Pr_{(f,z) \leftarrow \distr_n^{0}}\left[\Pi^{(f,z)}() = 1\right] \ge 0.2.
	$$
	
	We are going to build a $\BPP$ algorithm which is able to break $\PRFmod$ on those $n$, thereby contradicting Lemma~\ref{lemma:zhandry}.
	
	From (A), we have
	$$
	\Pr_{(f,z) \leftarrow \distr_n^{1}}\left[\Pi^{(f,z)}() = 1\right] - \left(1-\Pr_{(f,z) \leftarrow \distr_n^{0}}\left[\Pi^{(f,z)}() = 0\right]\right) \ge 0.1,
	$$
	which simplifies to
	$$
	\Pr_{(f,z) \leftarrow \distr_n^{1}}\left[\Pi^{(f,z)}() = 1\right] + \Pr_{(f,z) \leftarrow \distr_n^{0}}\left[\Pi^{(f,z)}() = 0\right] \ge 1.1.
	$$
	
	From the definition of $\distr_n^0$ and $\distr_n^1$, the above implies that
	$$
	\Pr_{(k,a) \leftarrow \keymod}\left[ \Pi^{(f,z_1)}() = 1 \text{ and } \Pi^{(f,z_0)}() = 0, f = \PRFmod_{(k,a)}, z_0 = 0^{\sqrt{N}/2}, z_1 = e_{a} \right] \ge 0.1,
	$$
	where $e_a$ denotes the string of length $\sqrt{N}/2$ that is all zero except for the $a$-th bit.
	
	{\bf Analysis of distributions $A_i^{(f,z)}$.} By a result of Sahai and Vadhan~\cite{sahai2003complete}, there are two polynomial-time samplable distributions $A_0^{(f,z)}$ and $A_1^{(f,z)}$ such that $\| A_0^{(f,z)} - A_1^{(f,z)} \| \ge 1 - 2^{-n}$ when $\Pi^{(f,z)}() = 1$; and $\| A_0^{(f,z)} - A_1^{(f,z)} \| \le 2^{-n}$ when $\Pi^{(f,z)}() = 0$.
	
	Hence, with probability $0.1$ over $(k,a) \leftarrow \keymod$, we have
	$$
	\| A_0^{(f,z_1)} - A_1^{(f,z_1)} \| \ge 1 - 2^{-n} \text{ and } \| A_0^{(f,z_0)} - A_1^{(f,z_0)} \| \le 2^{-n}.
	$$
	This means that either $\|A_0^{(f,z_0)} - A_0^{(f,z_1)}\| \ge 1/3$ or $\| A_1^{(f,z_0)} - A_1^{(f,z_1)} \| \ge 1/3$.
	
	Now we show that the above implies an algorithm that breaks $\PRFmod$, and therefore contradicts Lemma~\ref{lemma:zhandry}.
	
	{\bf The algorithm and its analysis.} Given oracle access to a function $f \leftarrow \PRFmod_{\keymod}$, our algorithm first picks a random index $i \in \{0,1\}$. \ It then simulates $A_i$ with oracle access to $(f,z)$ to take a sample from $A_i^{(f,z)}$, where $z=z_0=0^{\sqrt{N}/2}$; it records all the indices in $z$ that are queried by $A_i$. \ Now, with probability at least $0.1/2 = 0.05$, we have $\|A_i^{(f,z_0)} - A_i^{(f,z_1)}\| \ge 1/3$. \ Since $(f,z_0)$ and $(f,z_1)$ only differ at the $a$-th index of $z$, we can see that $A_i^{(f,z_0)}$ must have queried the $a$-th index of $z$ with probability at least $1/3$.
	
	Hence, with probability at least $0.05/3 = \Omega(1)$, one of the values recorded by our algorithm is $a$, and in that case our algorithm can find a collision in $f$ easily. \ However, when $f$ is a truly random function, no algorithm can find a collision with a non-negligible probability. \ Therefore, this algorithm is a distinguisher between $\PRFmod$ and a truly random function, contradicting the fact that $\PRFmod$ is secure by Lemma~\ref{lemma:zhandry}.
	
	{\bf Construction of the oracle $\oracle$.} Finally, we are ready to construct our oracle $\oracle$. \ We will let $\oracle$ encode pairs $(f_1,z_1),(f_2,z_2),\dotsc$, where $f_n$ is a function from $\domain_n$ to $\domain_n$ and $z_n \in \{0,1\}^{\sqrt{N}/2}$.

	For each $n$, we draw a random index $i \leftarrow \{0,1\}$, and then draw $(f_n,z_n) \leftarrow \distr_n^i$. \ We set $L$ to be the unary language consisting of all $0^n$ for which $(f_n,z_n)$ is drawn from $\distr_n^1$.
	
	From Lemma~\ref{lemma:zhandry}, a quantum algorithm can distinguish $\distr_n^0$ from $\distr_n^1$, except with negligible probability, by recovering $a$. \ Therefore, by a similar argument as in the proof of Theorem~\ref{theo:BPPvsBQP}, we have $L \in \BQP^{\oracle}$ with probability $1$.
	
	On the other hand, for a protocol $\Pi$ and a sufficiently large $n$, either (A) happens, which means that $\Pi^{(f_n,z_n)}$ is invalid with probability $0.05$ on input $0^n$, or (B) happens, which means that $\Pi$ cannot distinguish $\distr_n^0$ and $\distr_n^1$ with a constant probability.
	
	In both cases $\Pi$ cannot decide whether $0^n$ belongs to $L$ correctly with bounded error. \ Hence, again by a similar argument as in the proof of Theorem~\ref{theo:BPPvsBQP}, the probability that $\Pi$ decides $L$ is $0$. \ And since there are only countably many protocols, we have $L \notin \SZK^{\oracle}$ with probability $1$, which means that $\BQP^{\oracle} \not\subset \SZK^{\oracle}$ with probability $1$.
	
	Finally, it is easy to see that $\oracle \in \ppoly$, which completes the proof.
\end{proof}

\section{Missing Proofs in Section~\ref{sec:proposal}}
\label{sec:missing-proofs-proposal}

We first prove Lemma~\ref{lemma:dev-to-adv}.

\begin{proofof}{Lemma~\ref{lemma:dev-to-adv}}
	
	Let $N=2^n$ for simplicity and $L$ be a list consisting of $N$ reals: $|\langle u\spz{w}|^2 - 2^{-n}$ for each $w \in \{0,1\}^n$. \ We sort all reals in $L$ in increasing order, and denote them by $a_1,a_2,\dotsc,a_{N}$. \ We also let $\Delta = \dev(\spz{u})$ for brevity.
	
	Then from the definitions of $\adv(\spz{u})$ and $\dev(\spz{u})$, we have
	$$
	\sum_{i=1}^{N} a_i = 0,
	$$
	$$
	\sum_{i=1}^{N} |a_i| = \Delta,
	$$
	and
	$$
	\adv(\spz{u}) = \frac{1}{2} + \sum_{i=N/2+1}^{N} a_i.
	$$
	
	Now, let $t$ be the first index such that $a_t \ge 0$. \ Then we have
	$$
	\sum_{i=t}^{N} a_i = \sum_{i=t}^{N} |a_i| = \frac{\Delta}{2} \qquad \text{and} \qquad \sum_{i=1}^{t-1} a_i = - \sum_{i=1}^{t-1} |a_i| = -\frac{\Delta}{2}.
	$$
	
	We are going to consider the following two cases.
	
	\begin{itemize}
		\item (i): $t \ge N/2+1$. \
		Note that $a_i$'s are increasing and for all $i < t$, $a_i < 0$, we have
		$$
		\sum_{i=1}^{N/2} |a_i| \ge \sum_{i=N/2+1}^{t-1} |a_i|,
		$$
		which means
		$$
		\sum_{i=N/2+1}^{t-1} |a_i| \le \frac{1}{2} \cdot \sum_{i=1}^{t-1} |a_i| \le \frac{\Delta}{4}.
		$$
		
		Therefore,
		$$\sum_{i=N/2+1}^{N} a_i \ge \sum_{i=t}^{N} a_i + \sum_{i=N/2+1}^{t-1} a_i \ge \frac{1}{2} + \frac{\Delta}{2} - \frac{\Delta}{4} \ge \frac{1}{2} + \frac{\Delta}{4}.
		$$
		\item (ii): $t \le N/2$. In this case, note that we have
		$$
		\sum_{i=N/2+1}^{N} a_i \ge \sum_{i=t}^{N/2} a_i.
		$$
		Therefore,
		$$
		\sum_{i=N/2+1}^{N} a_i \ge \frac{1}{2} \cdot \sum_{i=t}^{N} a_i \ge \frac{\Delta}{4}.
		$$
		
	\end{itemize}
	
	Since in both cases we have $\sum_{i=N/2+1}^{N} a_i \ge \frac{\Delta}{4}$, it follows that $$\adv(\spz{u}) = \frac{1}{2} + \sum_{i=N/2+1}^{N} a_i \ge \frac{1}{2} + \frac{\Delta}{4},$$
which completes the proof.
	
\end{proofof}

Now we prove Lemma~\ref{lemma:random}.

\begin{proofof}{Lemma~\ref{lemma:random}}
	The random pure state $\spz{u}$ can be generated as follows: draw four i.i.d. reals $x_1,x_2,x_3,x_4 \sim \mathcal{N}(0,1)$, and set
	$$
	\spz{u} = \frac{(x_1 + x_2 i)\spz{0} + (x_3 + x_4 i)\spz{1}}{\sqrt{x_1^2+x_2^2+x_3^2+x_4^2}}.
	$$

	Hence, we have
	\begin{align*}
	\Ex\left[ \Big||\rpz{u} 0 \rangle|^2 - |\rpz{u} 1 \rangle|^2 \Big|\right]
	=& \int_{-\infty}^{\infty}\int_{-\infty}^{\infty}\int_{-\infty}^{\infty}\int_{-\infty}^{\infty} \frac{1}{(2\pi)^2} \frac{| x_1^2 + x_2^2 - x_3^2 - x_4^2|}{x_1^2 + x_2^2 + x_3^2 + x_4^2}
	\cdot e^{-(x_1^2+x_2^2+x_3^2+x_4^2)/2} dx_1 dx_2 dx_3 dx_4\\
	=& \int_{0}^{2\pi}\int_{0}^{2\pi}\int_{0}^{+\infty}\int_{0}^{+\infty} \frac{1}{(2\pi)^2} \cdot \frac{|\rho_1^2 - \rho_2^2|}{\rho_1^2+\rho_2^2} \cdot\rho_1\rho_2 \cdot e^{-(\rho_1^2+\rho_2^2)/2} d\rho_1 d \rho_2 d\theta_1 d\theta_2
	\tag{$x_1=\rho_1 \sin \theta_1$, $y_1 = \rho_1 \cos \theta_1$, $x_2=\rho_2 \sin \theta_2$, $y_2 = \rho_2 \cos \theta_2$}\\
	=& \int_{0}^{+\infty}\int_{0}^{+\infty} \frac{|\rho_1^2 - \rho_2^2|}{\rho_1^2+\rho_2^2} \cdot\rho_1\rho_2 \cdot e^{-(\rho_1^2+\rho_2^2)/2} d\rho_1 d \rho_2\\
	=& \frac{1}{2}
	\end{align*}
\end{proofof}

\section{Missing Proofs in Section~\ref{sec:sepa-samp-relation}}
\label{sec:missing-proofs-sepa-samp}

\newcommand{\Var}{\operatorname{Var}}

We prove Lemma~\ref{lemma:con-adv} here.

\begin{proofof}{Lemma~\ref{lemma:con-adv}}
	We prove the concentration inequality by bounding the variance,
	$$
	\Var[\adv(f)] = \Ex[\adv(f)^2] - \Ex[\adv(f)]^2.
	$$
	
	Note that
	$$
	\Ex[\adv(f)]^2 = \left( \frac{2}{\sqrt{2\pi}}\int_{1}^{+\infty} x^2e^{-x^2/2} dx \right)^2 = \Qsucc^2.
	$$
	
	We now calculate $\Ex[\adv(f)^2]$. \ We have
	\begin{align*}
	\Ex_{f}[\adv(f)^2] &= \Ex_{f} \left[ \left( \Ex_{z \in \{0,1\}^n} [\widehat{f}^2(z) \cdot \mathbf{1}_{|\widehat{f}(z)| \ge 1} ] \right)^2 \right] \\
	&= \Ex_{f} \left[ \Ex_{z_1,z_2 \in \{0,1\}^n} \left[\widehat{f}^2(z_1) \widehat{f}^2(z_2) \cdot \mathbf{1}_{|\widehat{f}(z_1)| \ge 1 \land |\widehat{f}(z_2)| \ge 1} \right] \right].\\
	&= \Ex_{z_1,z_2 \in \{0,1\}^n} \left[ \Ex_{f} \left[\widehat{f}^2(z_1) \widehat{f}^2(z_2) \cdot \mathbf{1}_{|\widehat{f}(z_1)| \ge 1 \land |\widehat{f}(z_2)| \ge 1} \right] \right].
	\end{align*}
	
	Now there are two cases: $z_1 = z_2$ and $z_1 \ne z_2$. \ When $z_1 = z_2$, let $z=z_1=z_2$; then we have
	\begin{align*}
	Ex_{f} \left[\widehat{f}^2(z_1) \widehat{f}^2(z_2) \cdot \mathbf{1}_{|\widehat{f}(z_1)| \ge 1 \land |\widehat{f}(z_2)| \ge 1} \right]
	=& \Ex_{f}\left[\widehat{f}^4(z) \cdot \mathbf{1}_{|\widehat{f}(z)| \ge 1} \right]\\
	=& \frac{2}{\sqrt{2\pi}}\int_{1}^{+\infty} x^4e^{-x^2/2} dx\\
	=& O(1).
	\end{align*}
	
	Next, if $z_1 \ne z_2$, then without loss of generality, we can assume $z_1 = 0^N$. \ Now we define two sets $A$ and $B$,
	
	$$
	A = \{x\in\{0,1\}^n : (z_2 \cdot x) = 0 \} \text{ and } B = \{x\in\{0,1\}^n : (z_2 \cdot x) = 1 \}.
	$$
	
	We also define
	$$
	\widehat{f}_A := \frac{1}{\sqrt{N/2}} \cdot \sum_{z \in A} f(z) \text{ and }
	\widehat{f}_B := \frac{1}{\sqrt{N/2}} \cdot \sum_{z \in B} f(z).
	$$
	
	Then from the definitions of $\widehat{f}(z_1)$ and $\widehat{f}(z_2)$, we have
	$$
	\widehat{f}(z_1) = \frac{1}{\sqrt{2}} \cdot (\widehat{f}_A + \widehat{f}_B) \text{ and }
	\widehat{f}(z_1) = \frac{1}{\sqrt{2}} \cdot (\widehat{f}_A - \widehat{f}_B).
	$$
	
	Therefore,
	\begin{align*}
	\Ex_{f} \left[\widehat{f}^2(z_1) \widehat{f}^2(z_2) \cdot \mathbf{1}_{|\widehat{f}(z_1)| \ge 1 \land |\widehat{f}(z_2)| \ge 1} \right]
	=& \left(\frac{1}{\sqrt{2\pi}} \right)^2 \cdot \int_{ |a+b| \ge \sqrt{2} \atop |a-b| \ge \sqrt{2}} \frac{1}{4} \cdot (a+b)^2 \cdot (a-b)^2 \cdot e^{-(a^2+b^2)/2} \cdot da db.
	\end{align*}
	
	Let $x=a+b$ and $y=a-b$. \ Then
$$ a = \frac{x+y}{2},~~~b = \frac{x-y}{2},~~~da = \frac{dx + dy}{2},~~~\text{and}~~~db = \frac{dx - dy}{2}.$$
Also note that $x^2 + y^2 = 2(a^2 + b^2)$. \ Plugging in $x$ and $y$, the above can be simplified to
	\begin{align*}
	\frac{1}{2\pi} \int_{|x| \ge \sqrt{2} \atop |y| \ge \sqrt{2}} \frac{1}{4} x^2 y^2 e^{-(x^2+y^2)/4} \cdot \frac{1}{2} dx dy
	=& \frac{1}{2\pi} \left( \int_{|x| \ge \sqrt{2}} \frac{1}{2\sqrt{2}} \cdot x^2 e^{-x^2/4} dx \right)^2\\
	=& \frac{1}{2\pi} \left( \int_{\sqrt{2}}^{+\infty} \frac{1}{\sqrt{2}} \cdot x^2 e^{-x^2/4} dx \right)^2\\
	=& \frac{1}{2\pi} \left( \int_{1}^{+\infty} 2 t^2 e^{-t^2/2} dt \right)^2 \tag{$t = x/\sqrt{2}$}\\
	=&   \left( \frac{2}{\sqrt{2\pi}}\int_{1}^{+\infty} t^2e^{-t^2/2} dt \right)^2
    =& \Qsucc^2.
	\end{align*}
	
	Putting two cases together, we have
	$$
	\Ex_{f}[\adv(f)^2] = \frac{1}{N} \cdot O(1) + \frac{N-1}{N} \cdot \Qsucc^2,
	$$
	
	which in turn implies
	$$
	\Var[\adv(f)] = O(1/N).
	$$
\end{proofof}

\section{Missing Proofs in Section~\ref{sec:oracle-sepa-ppoly}}
\label{sec:missing-proofs-oracle-sepa}

For completeness, we prove Lemma~\ref{lemma:zhandry} here.

\begin{proofof}{Lemma~\ref{lemma:zhandry}}
	In the following, we will always use $\varepsilon=\varepsilon(n)$ to denote a negligible function. \ And we will denote $\domainraw$ as $\domain$ for brevity. \ Recall that we interpret $\domain$ as $[N]$ for $N = N(n) =\domain$.
	
	{\bf Both $\PRPraw$ and $\PRFmod$ are classically-secure PRFs.} It is well-known that a secure PRP is also a secure PRF; therefore $\PRPraw$ is a classically-secure PRF. \ So we only need to prove this for $\PRFmod$.
	
	Recall that $\PRFmod_{(k,a)}(x) = \PRPraw_k( (x-1) \bmod a + 1 )$. \ We first show that if the $\PRPraw$ in the definition of $\PRFmod$ were replaced by a truly random function, then no classical polynomial-time algorithm $A$ could distinguish it from a truly random function. \ That is,
	\begin{equation} \label{eq:neg-lig}
	\left| \Pr_{f \leftarrow \domain^{\domain}, a \leftarrow \moduliset}[A^{f_{\bmod a}}() = 1] - \Pr_{f \leftarrow \domain^{\domain}}[A^{f}() = 1] \right| < \varepsilon,
	\end{equation}
	where $f_{\bmod a}(x) := f((x-1) \bmod a + 1)$.
	
	Clearly, as long as $A$ never queries its oracle on two points $x$ and $x'$ such that $x \equiv x' \pmod a$, the oracle will look random. \ Suppose $A$ makes $q$ queries in total. \ There are $\binom{q}{2}$ possible differences between query points, and each difference is at most $N$. \ So for large enough $N$, each difference can be divisible by at most two different moduli from $\moduliset$ (recall that each number in $\moduliset$ lies in $[\sqrt{N}/4,\sqrt{N}/2]$). \ And since $|\moduliset| \ge \Omega(\sqrt{N}/\log N)$, the total probability of querying two $x$ and $x'$ such that $x \equiv x' \pmod a$ is at most
	$$
	O\left( \frac{q^2 \log N}{\sqrt{N}} \right),
	$$
\noindent which is negligible as $N$ is exponential in $n$. \ This implies \eqref{eq:neg-lig}.
	
	Now, since $\PRPraw$ is a classically-secure PRF, for any polynomial-time algorithm $A$, we have
	\begin{equation}\label{eq:close-33}
	\left| \Pr_{f \leftarrow \domain^{\domain}, a \leftarrow \moduliset}[A^{f_{\bmod a}}() = 1] - \Pr_{f \leftarrow \PRPraw_{\keyraw}, a \leftarrow \moduliset}[A^{f_{\bmod a}}() = 1] \right| < \varepsilon,
	\end{equation}
\noindent since otherwise we can directly construct a distinguisher between $\PRPraw_{\keyraw}$ and $\domain^{\domain}$.
	
	Note that
	$$
	\Pr_{f \leftarrow \PRPraw_{\keyraw}, a \leftarrow \moduliset}[A^{f_{\bmod a}}() = 1] = \Pr_{f \leftarrow \PRFmod_{\keymod}}[A^{f}() = 1]
	$$
	by their definitions. Hence, \eqref{eq:neg-lig} and \eqref{eq:close-33} together imply that
	$$
	\left| \Pr_{f \leftarrow \PRFmod_{\keymod}}[A^{f}() = 1] - \Pr_{f \leftarrow \domain^{\domain}}[A^{f}() = 1] \right| <\varepsilon
	$$
\noindent for any polynomial-time algorithm $A$. \ This completes the proof for the first statement.
	
	{\bf Quantum algorithm for recovering $a$ given oracle access to $\PRFmod_{\keymod}$.} Let $(k,a) \leftarrow \keymod$, $f = \PRFmod_{(k,a)}$ and $g = \PRPraw_{k}$. \ From the definitions, we have $f = g_{\bmod a}$.
	
	Since $g$ is a permutation, there is no collision $(x,x')$ such that $g(x) = g(x')$. \ Moreover, in this case, $f = g_{\bmod a}$ has a unique period $a$. \ Therefore, we can apply Boneh and Lipton's quantum period-finding algorithm~\cite{boneh1995quantum} to recover $a$. \ Using a polynomial number of repetitions, we can make the failure probability negligible, which completes the proof for the second statement.
	
	{\bf Quantum algorithm for distinguishing $\PRPraw$ and $\PRFmod$.} Finally, we show the above algorithm implies a good quantum distinguisher between $\PRPraw$ and $\PRFmod$. \ Given oracle access to a function $f$, our distinguisher $A$ tries to recover a period $a$ using the previously discussed algorithm, and accepts only if $f(1) = f(1+a)$.
	
	When $f \leftarrow \PRPraw_{\keyraw}$, note that $f$ is a permutation, which means $A$ accepts with probability $0$ in this case.
	
	On the other side, when $f \leftarrow \PRFmod_{\keymod}$, from the second statement, $A$ can recover the period $a$ with probability at least $1 - \varepsilon$. \ Therefore $A$ accepts with probability at least $1-\varepsilon$.
	
	Combining, we find that $A$ is a distinguisher with advantage $1-\varepsilon$, and this completes the proof for the last statement.
\end{proofof}

\section{Numerical Simulation For Conjecture~\ref{conj:adv-large} }
\label{sec:tukareta}

\newcommand{\Cthr}{C_{\mathsf{thr}}}

Recall Conjecture \ref{conj:adv-large}, which said that a random quantum circuit $C$ on $n$ qubits satisfies $\adv(C) \ge \Cthr - \epsilon$ with probability $1 - 1/\exp(n)$, where

$$\Cthr := \frac{1 + \ln 2}{2}.$$

We first explain where the magic number $\Cthr$ comes from. \ Suppose $C$ is drawn from $\muharr^{2^n}$ instead of $\mugrid$. \ Then $C \spz{0^n}$ is a random quantum state, and therefore the values $2^n \cdot |\rpz{x} C\spz{0}|$'s, for each $x \in \{0,1\}^n$ are distributed very closely to $2^n$ i.i.d.\ exponential distributions with $\lambda = 1$.

So, assuming that, we can see that the median of $\probs(C\spz{0})$ concentrates around $\ln 2$, as
$$
\int_{0}^{\ln 2} dx e^{-x} = \frac{1}{2},
$$
which also implies that $\adv(C)$ concentrates around
$$
\int_{\ln 2}^{+\infty} x e^{-x} dx = \Cthr = \frac{1 + \ln 2}{2} \approx 0.846574.
$$

In the following, we first provide some numerical evidence that the values in $\probs(C\spz{0})$ \textit{also} behave like exponentially distributed random variables, which explains why the constant should indeed be $\Cthr$. \ Then we provide a direct numerical simulation for the distribution of $\adv(C)$ to argue that $\adv(C)$ approximately follows a nice normal distribution. \ Finally we examine the decreasing rate of the standard variance of $\adv(C)$ to support our conjecture.

\subsection{Numerical Simulation Setting}

In the following we usually set $n = 9$ or $n = 16$ (so that $\sqrt{n}$ is an integer); and we always set $m = n^2$ as in Conjecture~\ref{conj:adv-large}.

\subsection{Distribution of $\probs(C \spz{0})$ : Approximate Exponential Distribution}

In Figure~\ref{fg:distprobC} we plot the histogram of the distribution of the normalized probabilities in $\probs(C\spz{0})$ where $C \leftarrow \mugrid^{16,256}$, that is,
$$
\{ 2^{n} \cdot p  : p \in \probs( C\spz{0}) \}.
$$
And we compare it with the exponential distribution with $\lambda = 1$. \ From Figure~\ref{fg:distprobC}, it is easy to observe that these two distributions are quite similar.

\begin{figure}
	\begin{center}
	\includegraphics[width=12cm]{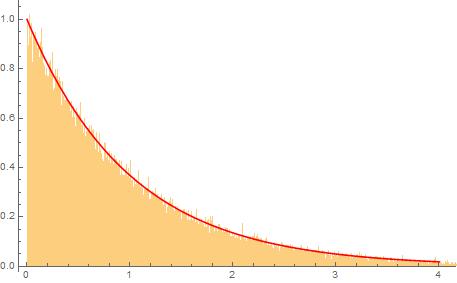}
	\caption{A histogram of (normalized) $\probs(C\spz{0})$, where $C \leftarrow \mugrid^{16,256}$.
		The x-axis represents the probability, and the y-axis represents the estimated density, and the {\color{red} red} line indicates the PDF of the exponential distribution with $\lambda = 1$.}\label{fg:distprobC}
    \end{center}
\end{figure}

\subsection{Distribution of $\adv(C)$ : Approximate Normal Distribution}

Next we perform direct numerical simulation to see how $\adv(C)$ is distributed when $C \leftarrow \mugrid^{n,m}$. \ Our results suggest that $\adv(C)$ approximately follows a normal distribution with mean close to $\Cthr$.

\subsubsection{$\mugrid^{9,81}$, $10^5$ samples}

We first draw $10^5$ i.i.d. samples from $\mugrid^{9,81}$ and plot the distribution of the corresponding $\adv(C)$'s in Figure~\ref{fg:n3m81}. \ From Figure~\ref{fg:n3m81}, we can see that the distribution of $\adv(C)$ follows a nice normal distribution, with mean very close to $\Cthr$.

\begin{figure}
	\begin{center}
		\includegraphics[width=12cm]{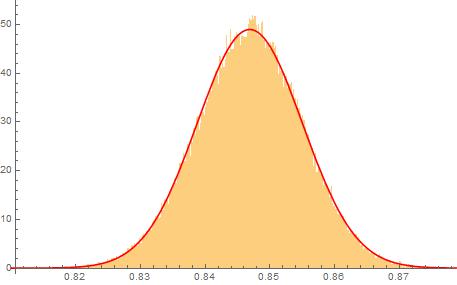}
		\caption{A histogram of the $\adv(C)$'s of the $10^5$ i.i.d. samples from $\mugrid^{9,81}$.
			The x-axis represents the value of $\adv(C)$, and the y-axis represents the estimated density, and the {\color{red} red} line indicates the PDF of the normal distribution $\mathcal{N}(0.846884, 0.00813911^2)$.}\label{fg:n3m81}
	\end{center}
\end{figure}

\subsubsection{$\mugrid^{16,256}$, $10^5$ samples}

Next, we draw $10^5$ i.i.d.\ samples from $\mugrid^{16,256}$ and plot the distribution of the corresponding $\adv(C)$'s in Figure~\ref{fg:n4m256}. \ From Figure~\ref{fg:n4m256}, we can observe that the distribution of $\adv(C)$ in this case also mimics a nice normal distribution, with mean even closer to $\Cthr$ than in the previous case.

\begin{figure}
	\begin{center}
		\includegraphics[width=12cm]{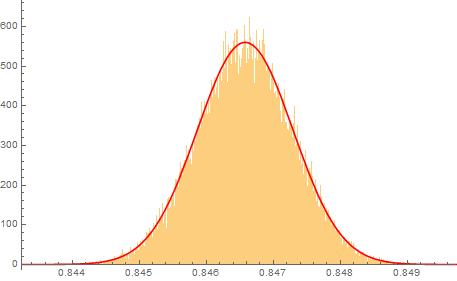}
		\caption{A histogram of the $\adv(C)$'s of the $10^5$ i.i.d. samples from $\mugrid^{16,256}$.
			The x-axis represents the value of $\adv(C)$, the y-axis represents the estimated density, and the {\color{red} red} line indicates the PDF of the normal distribution $\mathcal{N}(0.846579, 0.000712571^2)$.}\label{fg:n4m256}
	\end{center}
\end{figure}

\subsection{The Empirical Decay of Variance}

The previous subsection suggests that $\adv(C)$ follows a normal distribution with mean approaching $\Cthr$. \ If that's indeed the case, then informally, Conjecture \ref{conj:adv-large} becomes equivalent to the conjecture that the variance $\sigma$ of $\Cthr$ becomes $O(1/n)$ as $n \to +\infty$. \ So we wish to verify the latter conjecture for $\mugrid^{n,n^2}$ with some numerical simulation.

\newcommand{\mugeneral}{\mu_{\mathsf{general}}}

\subsubsection*{The circuit distribution $\mugeneral^{n,m}$}

Unfortunately, the definition of $\mugrid^{n,m}$ requires $n$ to be a perfect square, and there are only five perfect squares for which we can perform quick simulations ($n \in \{1,4,9,16,25\}$). \ So we consider the following distribution $\mugeneral^{n,m}$ on $n$ qubits and $m$ circuits instead: each of $m$ gates is a Haar random two-qubit gate acting on two qubits chosen uniformly at random. \ In this case, since we don't need to arrange the qubits in a square grid, $n$ can be any positive integer.

Numerical simulation shows that $\adv(C)$ is distributed nearly the same when $C$ is drawn from $\mugeneral^{n,n^2}$ or $\mugrid^{n,n^2}$ for $n = 3$ or $n = 4$, so it is reasonable to consider $\mugeneral$ instead of $\mugrid$.

For each $n = 2,3,\dotsc,16$, we draw 1000 i.i.d.\ samples from $\mugeneral^{n,n^2}$, and calculate the variance of the corresponding $\adv(C)$'s. \ The results are summarized in Figure~\ref{fg:vardecay}.

\begin{figure}
	\begin{center}
		\includegraphics[width=12cm]{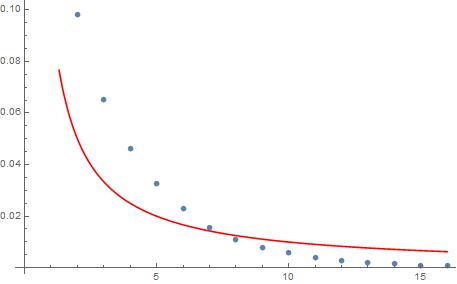}
		\caption{The empirical decay of the variance of $\adv(C)$. \ Here a point $(x,y)$ means that the standard variance of the corresponding $\adv(C)$'s for the 1000 i.i.d.\ samples from $\mugeneral^{x,x^2}$ is $y$. \ Also, the {\color{red} red} line represents the function $y = 0.1/x$.}\label{fg:vardecay}
	\end{center}
\end{figure}

From Figure~\ref{fg:vardecay}, we can observe that the variance decreases faster than the inverse function $1/x$; hence it supports Conjecture~\ref{conj:adv-large}.

\end{document}